\documentclass[fleqn]{article}

\usepackage{newa4}
\usepackage{pinlabel}
\usepackage{latexsym}
\usepackage{amssymb}
\usepackage{amsmath}
\usepackage{amsthm}
\usepackage{amsfonts}
\usepackage{mathabx}
\usepackage{leftidx}
\usepackage{bussproofs}

\allowdisplaybreaks

\newtheorem{lemma}{Lemma}[section]

\newtheorem{defin}[lemma]{Definition}
\newtheorem{proposition}[lemma]{Proposition}
\newtheorem{theorem}[lemma]{Theorem}
\newtheorem{corollary}[lemma]{Corollary}
\newtheorem{example}[lemma]{Example}

\newcounter{claim}[lemma]

\newenvironment{claim}{\begin{list}{}{\listparindent\parindent%
                                      \leftmargin0cm\parsep\parskip}%
                                      \item[] \refstepcounter{claim}%
                                      {\em Claim \arabic{claim} }}%
                                      {\end{list}}

\newcommand{\set}[2]{\{\,#1 \mid #2\,\}}
\newcommand{\fun}[3]{#1\colon #2 \to #3}
\newcommand{\range}{\mathrm{range}}

\newcommand{\subst}{\mathrm{Subst}}
\newcommand{\msubst}{\mathrm{MSubst}}
\newcommand{\sub}{\mathrm{sub}}

\newcommand{\cut}{\mathrm{cut}}
\newcommand{\un}{\mathrm{un}}

\newcommand{\comp}{\mathrm{comp}}
\newcommand{\sv}{\mathrm{sv}}
\newcommand{\sg}{\mathrm{sg}}
\newcommand{\nlg}{\mathrm{lg}}
\newcommand{\bb}{\textsc{Ext}}
\newcommand{\bn}{\mathrm{BN}}
\newcommand{\cs}{\mathrm{CS}}
\newcommand{\sem}{\mathrm{Sem}}
\newcommand{\fix}{\mathrm{fix}}
\newcommand{\fpair}{\mathrm{pair}}
\newcommand{\inv}{\mathrm{inv}}
\newcommand{\im}{\mathrm{im}}

\newcommand{\PPP}{\mathcal{P}}
\newcommand{\RRR}{\mathcal{R}}
\newcommand{\SSS}{\mathcal{S}}
\newcommand{\FFF}{\mathcal{F}}

\newcommand{\Def}{\overset{\mathrm{Def}}{=}}

\makeatletter
\newcommand{\prc}{\mathbin{\mathpalette\prc@inner\relax}}
\newcommand{\prc@inner}[2]{%
  \vbox{\offinterlineskip\m@th
    \ialign{%
      ##\cr
      \hidewidth\raisebox{-1.5\height}[0pt][0pt]{$#1.$}\hidewidth\cr
      $#1-$\cr
    }%
  }%
}
\makeatother

\newcommand{\llbracket}{[\![}
\newcommand{\rrbracket}{]\!]}
\newcommand{\val}[1]{\llbracket #1 \rrbracket}
\newcommand{\pair}[1]{\langle #1 \rangle}

\newcommand{\dom}{\mathrm{dom}}
\newcommand{\tp}{\mathrm{TP}}
\newcommand{\MP}{\mathrm{MP}}
\newcommand{\nf}{\mathrm{nf}}
\newcommand{\narg}{\mathrm{arg}}
\newcommand{\er}{\mathrm{er}}
\newcommand{\lth}{\mathrm{lth}}
\newcommand{\hpt}{\mathrm{HPT}}
\newcommand{\nor}{\mathrm{NOR}}
\newcommand{\qgn}{\mathrm{QGN}}

\newcommand{\conv}{\mathclose\downarrow}

\newcommand{\anf}{\textsc{Anf}}
\newcommand{\graph}{\mathrm{graph}}

\newcommand{\egraph}{\graph_{\mathrm{e}}}

\begin{document}

\title{How Much Partiality Is Needed for a Theory of Computability?}
\author{Dieter Spreen \\
ATHENA Visiting Researcher (AVR) Programme / Le Studium \\
Laboratoire d'Informatique Fondamentale d'Orl\'eans (LIFO) \\
University of Orl\'eans, INSA Centre Val de Loire, France\\ 
and \\
Department of Mathematics, University of Siegen, Germany\\
Email: {\tt spreen@math.uni-siegen.de}}
\date{}

\maketitle

\begin{abstract}
Partiality is a natural phenomenon in computability that we cannot get around. So, the question is whether we can give the areas where partiality occurs, that is, where non-termination happens, more structure. In this paper we consider function classes which besides the total functions only contain finite functions whose domain of definition is an initial segment of the natural numbers. Such functions appear naturally in computation. We show that a rich computability theory can be developed for these functions classes which embraces the central results of classical computability theory, in which all partial (computable) functions are considered. To do so, the concept of a G\"odel number is generalised, resulting in a broader class of numberings. The central algorithmic idea in this approach is to search in enumerated lists. In this way, function computability is reduced to set listability. Besides  the development of a computability theory for the functions classes, the new numberings---called quasi-G\"odel numberings---are studied from a numbering-theoretic perspective: they are complete, and each of the function classes numbered in this way is a retract of the G\"odel numbered set of all partial computable functions. Moreover, the Rogers semi-lattice of all computable numberings of the considered function classes is studied and results as in the case of the computable numberings of the partial computable functions are obtained. The function classes are shown to be effectively given algebraic domains in the sense of Scott-Ershov. The quasi-G\"odel numberings are exactly the admissible numberings of the computable elements of the domain. Moreover,  the domain can be computably mapped onto every other effectively given one so that every admissible numbering of the computable domain elements is generated by a quasi-G\"odel numbering via this mapping. 
\end{abstract}


\tableofcontents

\section{Introduction }\label{sec-intro}

A characteristic feature of the theory of computable functions is the appearance of  functions that are only partially defined. The reason is that the set of all total computable functions cannot be defined constructively, which is however possible for supersets that in addition contain certain partial computable functions. From the constructive description of a function class, an indexing of the functions in this class and a universal algorithm for the computation of the functions can be obtained in a canonical way. The existence of such an indexing and a universal algorithmic procedure is of central importance for the theory of computability. In the early days of computability theory (cf.\ e.g.~\cite{kle36}), only those function schemes were considered that correspond to total functions---thus leading to the class of general recursive functions---but soon one proceeded studying all such schemes and thus the class of partial recursive functions. For the reasons mentioned, it was only for this that the elegant and far-reaching theory as we know it today could be developed. See~\cite{ro67} for an in-depth going discussion. 

On the other hand, in applications one is of course only interested in total computable functions, because who wants an algorithm that runs through an infinite loop. This was probably also the reason why initially only the general recursive functions were considered. Also, in various areas of recursive mathematics one is only interested in the total computable functions, namely in theories that deal with the effective approximability of  infinite objects by finite ones. Examples are recursive analysis and domain theory (cf.~e.g.~\cite{tu3637,tu37,mo63,sc70,er72,ab80,wd80,de83,kt83,ws83}). Here, one examines infinite objects such as the real numbers, which can be effectively approximated by a sequence of finite objects, e.g. the rational numbers. Each such effective sequence is described by a total computable function. So, total computable functions can be considered as names of such effectively approximable objects. Now, if in applications one is more interested in total functions than in partial ones, is it really necessary to consider all partial computable functions in order to obtain a satisfactory theory of computability, or would it be sufficient,  to extend the space of total computable functions by only certain partial ones, namely those whose domain is well structured? In this paper it shall be shown that one can in fact limit oneself to the consideration of certain finite functions. Apart from the total computable functions, it suffices to consider only those finite functions the domain of which is an initial segment of $\omega$. It is also not necessary to add all functions defined on initial segments of the natural numbers to the total computable ones, but only a certain subset of them. Any such subset can be made even more sparse, also by removing infinitely many elements, still we will obtain a satisfying theory of computability.

Functions defined on initial segments of $\omega$ appear quite natural in this context. After all, every total function can be approximated by a sequence of such initial segment functions. This is used on various occasions. If, for instance, one evaluates a function defined by primitive recursion, say for argument $n$, one needs to compute the restriction of the function to the initial segment determined by $n$. In program testing one uses initial segment functions as sample. As we shall see, the class of functions we are considering is of interest not only in programming but also in other areas of computer science. It has been said already that in theories on the effective approximability of infinite objects by finite ones, total computable functions are used as names of the effectively approximable objects. This assignment can now be meaningfully extended to the initial segment functions. Each such function corresponds to a finite sequence of finite objects, which in turn corresponds to a certain vagueness: the object to be approximated has not yet been precisely determined. Initial segment functions are thus names for neighbourhoods in the topological space of the objects that can be approximated. This can be made more precise as follows: The function class considered here can be effectively mapped to the set of computable elements of an algebraic effectively given domain in such a way that the initial segment functions are mapped onto the base elements of the domain, which, as is known, define a basis for the Scott topology of the domain. Also in the interpretation of a domain as a neighbourhood or information system \cite{sc81,sc82}, the basic elements correspond to set systems that define an object only vaguely. The advantage of this approach is that names for the objects as well as for the ``approximating'' neighbourhoods are available in one namespace.

First investigations on the problem whether all partial computable functions have to be considered in order to develop a satisfying theory of computability were together with W.~H.~Kersjes~\cite{ke82}. In this work we started with a modified Turing machine model which for given computably enumerable set $A$ computed all total computable functions and exactly the initial segment functions with a domain of length in $A$. Based on the machine model, a numbering $\gamma^{A}$ of this function class was introduced. Yet, at this time we did not have at hand any characterisation of this numbering as known for G\"odel numberings. As a result, in many cases the computability of index functions had to be proven by constructing suitable machines and could not be derived from properties of the numbering. This could only be given in the habilitation thesis~\cite{sp85} of the current author, of which the present paper contains central, so far unpublished results.

The problem was that although the numbering  $\gamma^{A}$ has a computable universal function, this function is not in the corresponding function class. However, the graph of the universal function can be enumerated by a total computable function. This property corresponds to the Enumeration Theorem for G\"odel numberings~\cite[Theorem IV]{ro67}. The second characteristic property that G\"odel numberings share is that they are maximal with respect to reducibility among all numberings whose universal function is computable.
Similarly, it  turns out  in the case of the numbering $\gamma^{A}$ that  for every family $(r_{i})_{i \in \omega}$ of functions in the class under consideration, the graphs of which are computably enumerable, uniformly in $i$, one can computably pass from $i$ to an index of $r_{i}$ with respect to $\gamma^{A}$. We call numberings that satisfy these two conditions, quasi-G\"odel Numberings. As will be shown, any two such numberings are recursively isomorphic. Furthermore, all important theorems known for G\"odel numberings can be derived for this new type of numberings, without using any results of classical computability theory. As every G\"odel numbering is also a quasi-G\"odel numbering, this shows that the new notion is a reasonable generalisation of the concept of G\"odel numbering.
 
In classical computability theory, the concept of  computable enumerability has proven to be fundamental. For example, in~\cite{ro67}, the concept of a partial recursive operator is traced back to that of a computably enumerable set. As follows already from the definition of quasi-G\"odel numberings, the concept of  computable enumeration is also of central importance for the approach to computability theory that we want to present here. Many of the proofs consist of constructing enumerations for function graphs.

In the subsequent sections we start with developing computability theory on the basis of quasi-G\"odel numberings. In Section~\ref{sec-isf} we introduce the functions classes  which, in addition to the total computable functions, only contain those initial segment functions whose domains have certain given lengths. We then show that these classes have standard numberings with respect to a given G\"odel numbering of all partial computable functions. The standard numberings obtained in this way are in particular quasi-G\"odel numberings. In addition, we will discuss some applications of these function classes.
Section~\ref{sec-syn} shows that a quasi-G\"odel numbering of these function classes can also be obtained without using a G\"odel numbering. We present a machine model which computes exactly the functions in such a class and use this to define a numbering that turns out to be quasi-G\"odel numbering. This shows that a computability theory can be founded on the function classes considered here without prior knowledge of the theory of all partial computable functions.

In the next sections we derive some central results of computability theory from the properties of quasi-G\"odel numberings. In Section~\ref{sec-comp} we show that the \emph{smn}-theorem holds, discuss the effectiveness of substitution, give a normal form theorem and prove the recursion theorem, Rice's theorem and some consequences of these results. In contrast to the theory of all partial computable functions, the partial functions of the kind considered here can be extended to total computable functions, but the extension cannot be effectively computed from an index of the partial function. Also, the length of its domain cannot be computed from the index of an initial segment function. As we shall see, the \emph{smn}-theorem does not have the same importance in this theory as it does in computability theory with G\"odel numbering. There it is mainly used to construct index functions. Here these functions have to be constructed in a different way.

In Section~\ref{sec-ce} computably enumerable sets are introduced in the usual way as sets that are either empty or the range of a total computable function. They can also be characterised as ranges of the functions considered here, not necessarily total ones. With the help of a quasi-G\"odel numbering of this function class we then simultaneously get an enumeration of these sets. On the basis of a few selected theorems it will become clear that in this approach to the theory of computably enumerable sets essentially all classical results apply except for those in the formulation of which the domain characterisation of computably enumerable sets is included. Because of the normalisation of the domains of the considered functions, this characterisation is no longer meaningful here. As will be seen however, in the usual proofs the constructions using this characterisation can be replaced by others. In addition, we show that the numbering of computably enumerable sets introduced in the manner described is computably isomorphic to the numbering of these sets defined via a G\"odel numbering. This shows that the approach to computability theory presented here can be used to study computably enumerable sets in the same way as the classical one.

In Sections~\ref{sec-compII} and \ref{sec-effop}, respectively,  the enumerability of subsets of the functions classes under consideration and the computability of operators between these classes are discussed. Theorems of Rice-Shapiro and Myhill-Shepherdson type are derived.

Section~\ref{sec-semlat} is reserved for a numbering-theoretical investigation of quasi-G\"odel numberings. Here we address the question of the existence of minimal supersets of the class of total computable functions for which still quasi-G\"odel numberings exist. Furthermore, it is shown that all quasi-G\"odel numberings of a function class are computably isomorphic. From this it follows in particular that in Section~\ref{sec-isf} there is no restriction to construct a quasi-G\"odel numbering as a standard numbering for a G\"odel numbering. Finally, the Rogers semi-lattices of computable numberings of the function classes under consideration are examined and  results known for the case of all partial computable functions are transferred, such as Goetze's theorem~\cite{goe74, goe76} that every countable partially ordered set can be embedded isomorphically in this semi-lattice. It also follows from results of Mal'cev~\cite{ma70} and Khutoretski\u{\i}~\cite{kh69} that the considered function classes have infinitely many incomparable Friedberg numberings and that there are positive numberings of these classes to which no Friedberg numbering can be reduced.

In Section \ref{sec-dom} we investigate the connection with domain theory. This theory was established independently by Ershov~\cite{er72} and Scott~\cite{sc81} in their aim to develop a mathematically pleasing way to study the computable functionals of higher type and to construct a model of the untyped lambda calculus, respectively. As will be shown, the functions considered in the present work are just the computable elements of an effectively given algebraic domain, which contains all total number-theoretic functions in addition to the initial segment functions. The quasi-G\"odel numberings studied in the previous sections are precisely the admissible numberings in the sense of domain theory. We show that the domain of number-theoretic functions just described can be computably mapped onto any other effectively given domain. As a consequence, we obtain that every admissible numbering of the computable elements of an effectively given domain can be generated from a quasi-G\"odel numbering.

\section{The function classes $\SSS^{(n)}_{A}$} \label{sec-isf}

In what follows let $\fun{\pair{\cdot, \cdot}}{\omega^2}{\omega}$  be a one-to-one and onto computable pair encoding so that $\pair{x, y} \ge x, y$. Extend the pairing function as usual to an $n$-tuple encoding ($n > 0$) by setting $\pair{x} \Def x$ and $\pair{x_{1}, \ldots, x_{n+1}} \Def \pair{x_{1}, \pair{x_{2}, \ldots, x_{n+1}}}$. 
Let $\pi^{(n)}_{i}$ ($i = 1, \ldots, n$) be the associated decodings such that $\pi^{(n)}_{i}(\pair{x_{1}, \ldots, x_{n}}) = x_{i}$. The sets of all n-ary partial, total, partial computable, and total computable functions, respectively, will be denoted by $\PPP\FFF^{(n)}$, $\FFF^{(n)}$, $\PPP^{(n)}$, and $\RRR^{(n)}$. The arity $n$ of these functions and the dimension of the Cartesian products of $\omega$ that will be considered in the sequel is always assumed to be at least 1. In some cases the case $n=0$ could be included. But we will not discuss this case.

Let $\varphi^{(n)}$ be a G\"odel numbering of $\PPP^{(n)}$ and $W_{i}$ be the domain of $\varphi^{(1)}_{i}$.  Should the arity of the considered functions be clear from the context
or its knowledge be not important, we write $\varphi$ instead of $\varphi^{(n)}$.  We proceed accordingly with all other numberings of function classes considered here. The value of a numbering $\zeta$ at argument $i$ is denoted by $\zeta_i$, but sometimes also by $\zeta(i)$. Instead of $(x_{1}, \ldots, x_{n})$ we also write $\vec x$. Moreover, we write $\varphi(\vec x)\conv$ if the computation of $\varphi(\vec x)$ converges and $\varphi(\vec x)\conv_m$ if it converges in $m$ steps. If $C$ is a fixed non-empty computably enumerable (c.e.) subset of $\omega$, then we assume it to come with a fixed enumeration and denote the finite part of $C$ enumerated after $t$ steps by $C_{t}$.

A subset $B$ of $\omega$ is called \emph{initial segment}  of $\omega$ if for every $m \in B$ we have that also $m' \in B$, for all $m' < m$. The cardinality of $B$ is said to be the \emph{length} of $B$. If $C = B^{n}$, for some initial segment $B$ of $\omega$, then $C$ is called \emph{initial segment} of $\omega^{n}$. The length of $B$ is then also said to be the \emph{edge length} of $C$.

As has already been pointed out, we will consider functions the domain of which is either $\omega^{n}$, the empty set, or an initial segment of $\omega^{n}$ with an edge length in a given 
subset $A$ of $\omega$. Let $\anf^{(n)}_{A}$\footnote{The notation derives from the German word `Anfangsst\"uck' for `initial segment'.} be the set of  functions in $\PPP\FFF^{(n)}$ the domain of which is either empty or an initial segment of $\omega^{n}$ with edge length in $A$. Then we set
\begin{gather*}
\widehat{\SSS}^{(n)}_{A} \Def \FFF^{(n)} \cup \anf^{(n)}_{A}, \quad \widehat{\SSS}_{A} \Def \bigcup \set{\widehat{\SSS}^{(n)}_{A}}{n > 0}, \\
\SSS^{(n)}_{A} \Def \RRR^{(n)} \cup \anf^{(n)}_{A}, \quad \SSS_{A} \Def \bigcup \set{\SSS^{(n)}_{A}}{n > 0}.
\end{gather*}
 
For an infinite c.e.\ set $A$, $\SSS^{(n)}_{A}$ is an enumerable subset of $\PPP^{(n)}$. Let to this end, for $a \in A$, $a^{(n)}$ denote the $n$-tupel $(a, \ldots, a)$. We extend the usual less-or-equal relation on $\omega$ coordinatewise  to $\omega^{n}$ and write $\vec a < \vec b$ to mean $\vec a \le 
\vec b$ and $\vec a \neq \vec b$.  Then there is some $f \in \RRR^{(1)}$ so that
\[
\varphi^{(n)}_{f(i)}(\vec x) = \begin{cases}
                                                \varphi^{(n)}_{i}(\vec x) & \text{if for some $a \in A$, $\vec x < a^{(n)}$ and for all $\vec y < a^{(n)}$, $\varphi^{(n)}_{i}(\vec y)\conv$}, \\
                                                \text{undefined} & \text{otherwise}.
                                            \end{cases}
 \] 
It follows that $\varphi^{(n)}_{f(i)} \in \SSS^{(n)}_{A}$, for all $i \in \omega$. If for all $a \in A$ and all $\vec y < a^{(n)}$, $\varphi_{i}^{(n)}(\vec y)\conv$, then $\varphi^{(n)}_{f(i)} \in \RRR^{(n)}$. In the other case, $b = \max\set{a \in A}{(\forall \vec y < a^{(n)})\, \varphi^{(n)}_{i}(\vec y)\conv}$ exists and the domain of $\varphi^{(n)}_{f(i)}$ is an initial segment of $\omega^{n}$ of edge length $b$. Hence $\varphi^{(n)}_{f(i)} \in \anf^{(n)}_{A}$. Observe the use of excluded  middle in this proof.                                      

For $q \in \PPP\FFF^{(n)}$ let $\dom(q)$ and $\range(q)$, respectively, be the domain and the range of $q$. Moreover, let 
\begin{gather*}
\graph(q) \Def \set{\pair{\pair{\vec x}, z}}{\vec x \in \dom(q) \wedge q(\vec x) = z},  \\
\egraph(q) \Def \set{\pair{\pair{\vec y}, 0}, \pair{\pair{\vec x}, z+1}}{\vec y \in \omega^{n} \wedge \vec x \in \dom(q) \wedge q(\vec x) = z}
\end{gather*}
respectively be the \emph{graph} and the \emph{extended graph} of $q$. Then it is readly verified that the enumeration $\lambda i.\ \varphi^{(n)}_{f(i)}$ has the following two properties
\begin{align}
&\varphi^{(n)}_{i} \in \SSS^{(n)}_{A} \Rightarrow \varphi^{(n)}_{f(i)} = \varphi^{(n)}_{i},  \label{std} \\
&\graph(\varphi^{(n)}_{f(i)}) \subseteq \graph(\varphi^{(n)}_{i}).  \label{stdspec}
\end{align}

Enumerations of a subset of $\PPP^{(n)}$ that satisfy Condition~(\ref{std}) are called \emph{$\varphi$-standard numberings} and those that additionally meet Condition~(\ref{stdspec}) are \emph{special $\varphi$-standard numberings}~\cite{ma71}. Enumerations of this kind have first been considered by Lachlan~\cite{la64b} for the case of classes of c.e.\ sets. The class $\set{\egraph(\varphi^{(n)}_{f(i)})}{i \in \omega}$ is an example for the kind of classes he studied, that is, a standard class.

Set
\[
^{A}\varphi^{(n)}_{i} \Def \varphi^{(n)}_{f(i)},
\]
then we have that $\leftidx{^{A}}\varphi^{(n)}_{i}$ is a special $\varphi$-standard numbering of $\SSS^{(n)}_{A}$, for every infinite c.e.\ $A  \subseteq \omega$.\footnote{Note that the numbering $^{A}\varphi^{(n)}$ depends effectively on (a code) of $A$. We will, however, not make use of this fact}  The $\varphi$-standard numberings of $\SSS^{(n)}_{A}$ can be characterised by conditions similar to those known for G\"odel numberings. Since their universal function is computable, but not contained in $\SSS^{(n+1)}_{A}$, they will not satisfy the conditions for G\"odel numberings. 

\begin{theorem}\label{thm-qgn}
Let $\psi$ be a $\varphi$-standard numbering of $\SSS^{(n)}_{A}$. Then the following two conditions hold:
\begin{description}

\item[\rm (QGN~I)] The extended graph of $\lambda i, \vec x.\ \psi_{i}(\vec x)$ is enumerable by some function in $\RRR^{(1)}$.

\item[\rm (QGN~II)] For all $k \in \RRR^{(2)}$ there exists $v \in \RRR^{(1)}$ such that if, for some $i \in \omega$, $\lambda t.\ k(i,t)$ enumerates the extended graph of a function $r \in \SSS^{(n)}_{A}$, then $ r = \psi_{v(i)}$.

\end{description}
\end{theorem}
\begin{proof}
Because $\psi$ is a $\varphi$-standard numbering there is some $d \in \RRR^{(1)}$ with $\psi_{i} = \varphi_{d(i)}$. Therefore, $\lambda i, \vec x.\ \psi_{i}(\vec x)$ is computable and $\egraph(\lambda i, \vec x.\ \psi_{i}(\vec x))$ is c.e. By construction, $\egraph(\lambda i, \vec x.\ \psi_{i}(\vec x))$ is infinite. Hence, it can be enumerated by a total computable function. Thus, (QGN~I) holds.

For the derivation of (QGN~II) let $v \in \RRR^{(1)}$ so that
\[
\varphi_{v(i)}(\vec x) = \pi^{(2)}_{2}(\mu \pair{t, z}.\ [\pi^{(2)}_{1}(k(i, t)) = \pair{\vec x} \wedge \pi^{(2)}_{2}(k(i, t)) = z \wedge z > 0]) -1.
\]
If $\lambda t.\ k(i,t)$ enumerates the extended graph of a function $r \in \SSS^{(n)}_{A}$ it follows that $\varphi_{v(i)} = r$. Because $r \in \SSS^{(n)}_{A}$ and $\psi$ is a $\varphi$-standard numbering we thus obtain that $r = \varphi_{d(v(i))} = \psi_{v(i)}$.
\end{proof}

In the sequel, a numbering of a countable set of functions $X \supseteq \RRR^{(n)}$ that satisfies Conditions (QGN~I) and (QGN~II) is called \emph{quasi-G\"odel numbering} (replace $\SSS^{(n)}_{A}$ by $X$ in (QGN~II)). As we have just seen, every $\varphi$-standard numbering of such a function class is a quasi-G\"odel numbering. In particular,
\begin{corollary}
Every G\"odel numbering is a quasi-G\"odel numbering.
\end{corollary} 

As will be shown in Section~\ref{sec-semlat}, every quasi-G\"odel numbering of $\SSS^{(n)}_{A}$ is a $\varphi$-standard numbering, up to recursive isomorphism. The analogy of Conditions (QGN~I) and (QGN ~II) to the corresponding conditions for G\"odel enumerations can best be seen by identifying functions with their graphs. (QGN~I) then corresponds to the computability of the universal function of the enumeration and (QGN~II) says that every effective enumeration of functions of the considered set can be computably reduced to the given enumeration. 

As mentioned earlier, the set of extended graphs of a $\varphi$-standard numberable function class is a special case of the classes of c.e.\ sets studied by Lachlan~\cite{la64b}. He calls these classes standard classes. If one interprets the numbering of such a set of functions as an enumeration of the class of the associated extended graphs, then the Conditions (QGN~I) and (QGN~II) correspond precisely to Lachlan's requirements on a standard enumeration of a standard class.

As we see, the classes $\SSS^{(n)}_{A}$ for infinite c.e.\ sets $A$ have very pleasant effectivity properties. This already suggests that we have made the right choice with the initial segment functions in our plan to develop a computability theory for a set of functions which, in addition to the total computable functions, only contains certain partial functions. Since every total function can be approximated by a sequence of such initial segment functions, these functions also occur quite naturally in computability theory. For example, a function defined by primitive recursion is evaluated from the beginning, i.e. starting at 0. You proceed in the same way with the normalised $\mu$ operator
\[
\mu x.\ [q(x) = 0] \Def \begin{cases}
					\min \set{x}{q(x) = 0 \wedge (\forall y < x)\, q(y)\conv} & \text{if $\set{x}{q(x) = 0 \wedge (\forall y < x)\, q(y)\conv}$} \\
					& \text{is not empty}, \\
					\text{undefined} & \text{otherwise}.
				\end{cases}
\]

Moreover, for every infinite $A \subseteq \omega$, the collection of sets $\set{\set{g \in \RRR^{(1)}}{\graph(p) \subseteq \graph(g)}}{p \in \anf^{(1)}_{A}}$ is a basis of a metric topology on $\RRR^{(1)}$, the Baire topology. Finally, this approximation property is used in the definition of computable functionals on sets of total functions (cf.\ \cite{da58, ws83}). We want to clarify this using the example of a functional $\fun{G}{\FFF^{(1)} \times \omega}{\omega}$. Let $\val{\cdot}$ be a one-to-one computable coding of all finite sequences of natural numbers and for $\fun{p}{\{  0, \ldots, a \}}{\omega}$
\[
\val{p} \Def \val{p(0), \ldots, p(a)}.
\]
Then $G$ is called \emph{computable} if there is some $g \in \RRR^{(2)}$ such that for $d \in \FFF^{(1)}$ and $x, y \in \omega$, $G(d, x) = y$, exactly if there exists $p \in \anf^{(1)}_{\omega}$ with $\graph(p) \subseteq \graph(d)$ so that $g(\val{p}, x) = y$. The initial segment function $p$ and in particular the length of its domain is a measure for the amount of information about $d$ that the algorithm $g$ needs to compute $G(d, x)$. Of course one would like to have such algorithms $g$ that manage with as little information as possible. Gordon and Shamir~\cite{gs83} e.g.\ investigated whether such algorithms always exist and how they can be constructed.

There are various approaches to define computability for uncountable sets other than $\FFF^{(1)}$ or $\PPP\FFF^{(1)}$, such as computable analysis~\cite{tu3637,tu37,pr89,we00},  domain theory~\cite{sc70,er72,we87,slg94}, the theory of filter spaces~\cite{hy79}, the theory of effectively given spaces~\cite{ kt83}, the theory of finitely approximable sets~\cite{hi84}. Most of these approaches have in common that the elements of the sets under consideration can be approximated by sequences of other finite objects; the real numbers, for example, by normed Cauchy sequences of rational numbers or descending sequences of closed intervals with rational endpoints. Encoding the finite objects used for the approximation allows to describe the approximating sequences by sequences of natural numbers. In \cite{ha71,ha73,ha76,ha78,ws83,kw85} it is therefore proposed to use functions in $\FFF^{(1)}$ as names for the elements approximated in this way, which led to the now well developed theory of representations. These assignments can be meaningfully extended to $\widehat{\SSS}^{(1)}_{\omega}$, or to put it another way, it makes sense to also use initial segment functions as names. If one considers that all information contained in the approximating finite objects is encoded in the elements to be approximated, then these objects conversely contain only a finite part of the information contained in the approximated elements (cf.~\cite{sc82,hi84}). The approximating finite objects and thus also every finite sequence of these objects therefore corresponds to a certain amount of uncertainty with regard to the element to be approximated: this is not yet uniquely determined by the finite approximation. In many cases these uncertainty sets are open sets in the topological space of the elements to be approximated. This observation suggests to take initial segment functions as names for the uncertainty set generated by finite sequences of approximating objects, which has the advantage that one has names for the elements of the considered set as well as for the uncertainty sets occurring in the approximation in one namespace. In addition, the described extension of representations to $\widehat{\SSS}^{(1)}_{\omega}$ is continuous: if the total function used in the representation is approximated by initial segment functions, then the element named by the total function is approximated by the uncertainty sets corresponding to the initial segment functions. We want to illustrate this with an example below. In  Section~\ref{sec-dom} we will make the statement precise and prove it for the case of effectively given algebraic domains.

\begin{example}
We assume the real numbers $x \in [-1, 1]$ be given in signed digit expansion
\[
x = \sum_{i=0}^{\infty} a_{i} \cdot 2^{-(i+1)}
\]
with $a_{i} \in \{-1, 0, 1 \}$. The information about $x$ which we can read off from this expansion is the sequence $a_{0}, a_{1}, \ldots$ of signed digits. To each finite initial segment of this sequence corresponds a dyadic rational
\[
u_{n} \Def \sum_{i=0}^{n} a_{i} \cdot 2^{-(i+1)}
\]
approximating $x$. The uncertainty set coming with this number is the interval $[u_{n} - 2^{-(n+1)}, u_{n} + 2^{-(n+1)}]$ including all real numbers containing the information about $u_{n}$, that is the sequence $a_{0}, \ldots, a_{n}$, as part of their own information.

Now, for $a \in \omega$, let $\delta(a) \Def (a \bmod 3)-1$ and for $r \in \widehat{\SSS}^{(1)}_{\omega}$ set
\[
G(r) \Def \bigcap \set{[\delta(r(n))-2^{-(n+1)}, \delta(r(n))+2^{-(n+1)}]}{n \in \dom(r)}.
\]
If, in addition, we identify a real number $z$ with the one-point interval $[z, z]$, we have thus obtained a representation of both, the reals in $[-1, 1]$ and the corresponding uncertainty sets. It is such that for $g \in \FFF^{(1)}$ and any sequence $(p_{n})_{n\in \omega}$ of initial segment functions so that $\graph(p_{n}) \subseteq \graph(p_{n+1})$ and $\graph(g) = \bigcup \set{\graph(p_{n})}{n \in \omega}$, 
\[
G(g) = \bigcap \set{G(p_{n})}{n \in \omega}.
\]
\end{example}

\section{A modified Turing machine model} \label{sec-syn}

In the last section we constructed a special $\varphi$-standard numbering of $\SSS^{(n)}_{A}$ for infinite c.e.\ sets $A \subseteq \omega$. In this section we present a machine model for the computation of these functions. The numberings derived from this characterisation will be a quasi-G\"odel numbering.

In what follows let $\Sigma$ be a non-empty finite set and $\tp_{\Sigma}$ be the set of all Turing programs which use $\Sigma$ as tape alphabet. Moreover, let $\fun{\kappa}{\omega}{\omega^{n}}$ be an effective enumeration of $\omega^{n}$ without repetitions and $A \subseteq \omega$ be a non-empty c.e.\ set. Finally, let $P$ be a Turing program and $M^{(n)}_{A}$ be a Turing program realising the following algorithm:\\

\noindent
\emph{Input}: $\vec x$
\begin{enumerate}
 \item Set  $j := 1$.
 
 \item \label{step} Compute $A_{j}$.
 
 \item For $i = 0, \ldots, j-1$: run program $P$ on input $\kappa(i)$ for $j$ steps and store whether there is a result.
 
 \item If there is some $a \in A_{j }$ so that
 \begin{itemize}
 \item $\vec x < a^{(n)}$ 
 
 \item For all $\vec y < a^{(n)}$: $\vec y \in \{ \kappa(0), \ldots, \kappa(j-1) \}$ and program $P$ stops on input $\vec y$ within $j$ steps,
 
 \end{itemize}
 then set $z := \text{output of $P$ on input $\vec x$}$;
 otherwise increase $j$ by 1 and go to (\ref{step}).
 
 \end{enumerate}
 
 \noindent
 \emph{Output}: $z$.\\
 
 \noindent
We write $M^{(n)}_{A}(P)$ to denote the Turing program consisting of program $M^{(n)}_{A}$ and the subprogram $P$. If $\fun{\sem^{(n)}}{\tp_{\Sigma}}{\PPP^{(n)}}$ is  the semantic map that, with regard to a fixed input/output convention, assigns to every Turing program the function it computes, then $\sem^{(n)}(M^{(n)}_{A}(P))$ is precisely the function that  agrees with $\sem^{(n)}(P)$ on the maximal  initial segment of $\omega^{n}$ with an edge length in $A$ that is contained in $\dom(\sem^{(n)}( P))$; and is  undefined, otherwise.

\begin{theorem}
For every infinite c.e.\ set $A \subseteq \omega$,
\[
\SSS^{(n)}_{A} = \set{\sem^{(n)}(M^{(n)}_{A}(P))}{P \in \tp_{\Sigma}}.
\]
\end{theorem}

Let $\MP_{i}$ be the $i$-th program with respect to an effective one-to-one enumeration of the set $M^{(n)}_{A}[\tp_{\Sigma}]$ of all modified Turing programs $M^{(n)}_{A}(P)$ with $P \in \tp_{\Sigma}$ and define
\[
^{A}\psi^{(n)}_{i} \Def \sem^{(n)}(\MP_{i}).
\]

\begin{theorem}\label{thm-qgn2}
For every infinite c.e.\ set $A \subseteq \omega$, $\leftidx{^{A}}\psi^{(n)}$ is a quasi-G\"odel numbering of $\SSS^{(n)}_{A}$.
\end{theorem}

If one compares the above definition of the set $M^{(n)}_{A}[\tp_{\Sigma}]$ with the definition of the numbering $\leftidx{^{A}}\varphi{(n)}$ in the previous section, one sees that the same rule is behind both: algorithms for the computation of functions in $\PPP^{(n)}$ are modified to ones for computing functions in $\SSS^{(n)}_A$. The only difference is that in the first case this is done by direct manipulation of the algorithms formulated in a given algorithmic language, here the language of Turing programs, and in the second case by effective operations on the indices (names) for these algorithms. Therefore, when defining $\leftidx{^{A}}\varphi^{(n)}$, a coding of all algorithms had to be used (G\"odel numbering $\varphi^{(n)}$), the properties of which were exploited in the proof of Theorem~\ref{thm-qgn}. The numbering $\leftidx{^{A}}\psi^{(n)}$, on the other hand, can be defined directly by listing the modified programs. However, in the proof of the above theorem one can no longer fall back on any known properties of a numbering, but must prove the existence of the computable functions stated in Conditions (QGN~I) and (QGN~II) by specifying suitable Turing programs. Since we essentially want to derive all results presented in the next sections from the properties of quasi-G\"odel numberings, this procedure shows that a computability theory for $\SSS_A$ can be established and developed without prior knowledge of the theory of all partial computable functions. The assumption about $A$ in this section is to be understood as an abbreviated way of saying that $A$ is the image of a total, one-to-one function that can be computed by a Turing program. It does not mean that one has to know already what a c.e. set is.

\begin{proof}[Proof of Theorem~\ref{thm-qgn2}]  
(QGN~I). Let $P$ be a Turing program realising  the following algorithm: \\

\noindent
\emph{Input}: t
\begin{enumerate}
\item If $t$ is odd, set $z := \pair{(t-1)/2, 0}$ and go to~(\ref{end}); otherwise, find $i$, $j$ and $\vec x$ so that $\pair{i, j, \vec x} = t/2$.

\item Simulate $j$ steps of the computation of the $i$-th modified Turing program $\MP_{i}$ on input $\vec{x}$.

\item If the computation of $\MP_{i}$ on input $\vec x$ stops within $j$ steps, set\\ $z := \pair{\pair{i, \vec x}, 1+\text{output of $\MP_{i}$ on input $\vec x$}}$; otherwise, set $z := \pair{t/2, 0}$.

\item\label{end} Stop.
\end{enumerate}

\noindent
\emph{Output}: $z$.\\

\noindent
Let $h \Def \sem^{(1)}(P)$. Then $h \in \RRR^{(1)}$ with $\range(h) = \egraph(\lambda i, \vec x.\ ^{A}\psi^{(n)}_{i}(\vec x))$. 

(QGN~II). Let $\widehat{P}$ be a Turing program computing the function
\[
\pi^{(2)}_{2}( \mu \pair{t, z}.\, [\pi^{(2)}_{1}(k(i, t)) = \pair{\vec x} \wedge \pi^{(2)}_{2}(k(i, t)) = z \wedge z >0]) - 1.
\]
 Moreover, for $i \in \omega$, let $P_{i}$ be a program that on input $\vec x$ outputs $\pair{i, \vec x}$ in such a way that it can be read as input by $\widehat{P}$. For two programs $\bar{P}$ and $\tilde{P}$ let $\bar{P};\tilde{P}$ be the program that first executes $\tilde{P}$ and then $\bar{P}$. Finally, let $Q$ be a program that on input of $i$ computes a $j$ with $\MP_{j} = M^{(n)}_{A}(\widehat{P};P_{i})$. Then $v \Def \sem^{(n)}(Q)$ has the property stated in (QGN~II).
\end{proof}

\section{Computability theory for $\SSS_{A}$}\label{sec-comp}

In this and the following sections let $A \subseteq \omega$ be an infinite c.e.\ set. In these sections we will show that a sufficiently rich computability theory can be developed for the set of functions $\SSS_A$. For this purpose we give a selection of important theorems of classical computability theory and show that they also apply to the set of functions considered here. Since this contains in particular all total computable functions, computable sets and relations can be introduced as usual in the theory to be developed here. We do not want to go into this any further.

Let $\theta^{(n)}$ be a numbering of $\SSS^{(n)}_{A}$. If $\theta^{(n)}$ satisfies Condition~(QGN~I), we obtain an analogue to Kleene's Normal Form Theorem.

\begin{theorem}
Let  $\theta^{(n)}$ satisfy Condition~(QGN~I). Then there is a primitive recursive function $q$ and an $(n+2)$-ary computable predicate $T$ so that
\[
\theta^{(n)}_{i}(\vec x) = q(\mu y.\ T(i, \vec x, y)).
\]
\end{theorem}
\begin{proof}
Choose the pair encoding and its decodings as primitive recursive and 
let $h \in \RRR^{(1)}$ enumerate the extended graph of the universal function of the numbering $\theta^{(n)}$. Then
\[
\theta^{(n)}_{i}(\vec x) = \pi^{(2)}_{2}(\mu \pair{t,z}.\, [\pi^{(2)}_{1}(h(t)) = \pair{i, \vec x} \wedge \pi^{(2)}_{2}(h(t)) = z \wedge z > 0]) -1.
\]
Therefore, defining $q(a) \Def \pi^{(2)}_{2}(a) \prc 1$ and
\[
T(i, \vec x, y) \Leftrightarrow \pi^{(2)}_{1}(h(\pi^{(2)}_{1}(y))) = \pair{i, \vec x} \wedge \pi^{(2)}_{2}(h(\pi^{(2)}_{1}(y))) = \pi^{(2)}_{2}(y) \wedge \pi^{(2)}_{2}(y) > 0,
\]
 proves the claim.
\end{proof}

Since the numberings considered here do not have a universal function in the considered  function class, we cannot construct index functions as we do in the  case of  G\"odel numberings, by first defining a function that contains the index parameters as arguments and then applying the \emph{smn}-theorem. We have to obtain such index functions by applying condition (QGN~II). Therefore, in the proofs of this and the following sections, we often need to construct functions that enumerate the extended graph of another function. These enumeration functions are essentially defined in two ways. In order not to have to repeat these constructions in the following, we want to present the first of these definition methods in the form of a scheme. Let $\kappa$ be an effective one-to-one enumeration of $\omega^n$, in which for $a < b$ ($a, b \in \omega$) all elements of the initial segment of $\omega^n$ of edge length $a$ occur before the remaining elements of the initial segment of edge length $b$. In the following we say that $\kappa$ enumerates $\omega^n$  \emph{initial segment by  initial segment}. Also let $\nf, \narg \in \RRR^{(1)}$ \footnote{The notation derives from the German words `Nachfolger' (successor) and `Argument'.} and for all $i,j >0$, $\ast \in \RRR^{(1)}$ with
\begin{gather*}
\nf(\pair{x_{1}, \ldots, x_{n}}) \Def \pair{\kappa(\kappa^{-1}(x_{1}, \ldots, x_{n}) + 1)}, \\
\narg(z) \Def \begin{cases}
			\pair{\vec 0} & \text{if $\pi^{(2)}_{2}(z) = 0$}, \\
			\nf(\pi^{(2)}_{1}(z)) & \text{otherwise},
		  \end{cases}\\
\pair{x_{1}, \ldots, x_{i}} \ast \pair{y_{1}, \ldots, y_{j}} \Def \pair{x_{1}, \ldots, x_{i}, y_{1}, \ldots, y_{j}}.\
\end{gather*}
(With this notation we suppress the dependence of the functions $\kappa$, $\nf$, $\narg$ and $\ast$ on $n$ and $i$ and $j$, respectively.)

 As follows from the definition, $\nf(\pair{\vec x})$ is the encoded successor of $\vec x$ in the enumeration $\kappa$.
 
\begin{lemma}\label{lem-meth1}
Let $f \in \RRR^{(3)}$ and $Q \subseteq \omega^{3}$ be a computable relation so that $f(i,t,z) > 0$ if $Q(i,t,z)$ holds. Moreover, let $k \in \PPP\FFF^{(2)}$ be defined by
\begin{align*}
& k(i, 2t) = \pair{\pair{\kappa(t)},0} \\
& k(i, 2t+1) = g(i, t), \quad \text{where} \\
&g(i,0) = \pair{\pair{\vec 0}, 0}, \\
&g(i, t+1) = \begin{cases}
			\pair{\narg(g(i,t)), f(i,t+1,g(i,t))} & \text{if $Q((i, t+1, g(i,t))$},\\
			g(i,t) & \text{otherwise}.
		  \end{cases}
\end{align*}
Then $k \in \RRR^{(2)}$ and $\lambda t.\ k(i,t)$ enumerates the extended graph of an $n$-ary function.
\end{lemma} 
\begin{proof}
Note that $\lambda t.\ k(i, 2t)$ enumerates alle coded tuples $\pair{\pair{\vec y}, 0}$. The value of $g(i, 0)$ is one of these coded tuples. As long as $Q(i, t+1, g(i, t))$ holds, $\arg(g(i, t))$ computes the successor of $\pi^{(2)}_{1}(g(i, t))$ in the enumeration given by $\kappa$. Since $f(i, t+1, g(i, t)) > 0$ in this case, $g(i, t+1)$ satisfies the requirement for the remaining elements in an extended graph. If $Q(i, t+1, g(i, t))$ does not hold, the value of $g(i, t)$ is repeated. In particular, if $Q$ is empty, we have that $k(i, 2t+1) = g(i, 0) = \pair{\pair{\vec 0}, 0}$, for all $t \in \omega$. Thus, $\lambda t.\ k(i, t)$ enumerates the extended graph of the nowhere defined function in this case.
\end{proof}
		  
As already mentioned, the concept of enumeration is of central importance for the computability theory to be developed here. The numberings of $\SSS^{(n)}_A$ considered here do not have universal functions in $\SSS^{(n+1)}_A$, but the extended graphs of the functions in $\SSS^{(n)}_{A}$ can be enumerated uniformly. We first show a more general result.

\begin{lemma}\label{lem-preunif}
Let $\theta^{(m+n)}$ satisfy Condition~(QGN~I). Then there is some $k \in \RRR^{(2)}$ such that
\[
\range(\lambda t.\ k(\pair{i, \vec y}, t)) = \egraph(\lambda \vec x.\ \theta^{(m+n)}_{i}(\vec y, \vec x)).
\]
\end{lemma}
\begin{proof} 
By Condition~(QGN~I), the extended graph of $\lambda i, \vec y, \vec x.\ \theta^{(m+n)}_{i}(\vec y, \vec x)$ has an enumeration $h \in \RRR^{(1)}$. Define
\begin{gather*}
\widehat{Q}(\pair{i, \vec y}, a, z) \Leftrightarrow \pi^{(2)}_{2}(h(a)) > 0 \wedge \pi^{(2)}_{1}(h(a)) = \pair{i, \vec y} \ast \narg(z), \\
Q(\pair{i, \vec y}, t, z) \Leftrightarrow (\exists a \le t)\, \widehat{Q}(\pair{i, \vec y}, a, z) \quad \text{and} \\
f(\pair{i, \vec y}, t, z) \Def \pi^{(2)}_{2}(h(\mu a \le t.\ \widehat{Q}(\pair{i, \vec y}, a , z))).
\end{gather*}
Now, by applying Lemma~\ref{lem-meth1} we obtain a function $k \in \RRR^{(2)}$. As is readily verified, it has the properties stated. 
\end{proof}

The result we are looking for now follows as special case $m = 0$.

\begin{theorem}\label{thm-qgn1}
Let $\theta^{(n)}$ satisfy Condition~(QGN~I). Then there is some $k \in \RRR^{(2)}$ so that
\[
\range(\lambda t.\ k(i, t)) = \egraph(\theta^{(n)}_{i}).
\]
\end{theorem}

As a further consequence we obtain the \emph{smn}-theorem.

\begin{theorem}[\emph{smn}-Theorem]\label{thm-smn}
Let $m > 0$, $\theta^{(n)}$ satisfy Condition~(QGN~II) and $\theta^{(m+n)}$ meet Condition~(QGN~I). Then there is a function $s \in \RRR^{(m+1)}$ so that
\[
\theta^{(n)}_{s(i, \vec y)}(\vec x) = \theta^{(m+n)}_{i}(\vec y, \vec x).
\]
\end{theorem}
\begin{proof}
By Lemma~\ref{lem-preunif} there is some $k \in \RRR^{(2)}$ such that 
\[
range(\lambda t.\ k(\pair{i, \vec y}, t)) = \egraph(\lambda \vec x.\ \theta^{(m+n)}_{i}(
\vec y, \vec x)).
\]
Now, let $v \in \RRR^{(1)}$ be as in Condition~(QGN~II). Since $\lambda \vec x.\ \theta^{(m+n)}_{i}(\vec y, \vec x) \in \SSS^{(n)}_{A}$, it follows that $\theta^{(m+n)}_{i}(\vec y, \vec x) = \theta^{(n)}_{v(\pair{i, \vec y})}(\vec x)$. Thus, $s \Def \lambda i, \vec y.\ v(\pair{i, \vec y})$ is as wanted.
\end{proof}

Just as the \emph{smn}-theorem is the effective version of the reducibility requirement for G\"odel numberings, there is also an effective version of (QGN~II) for quasi-G\"odel numberings:
\begin{description}
\item[\rm(QGN~E)]  There is a function $d \in \RRR^{(m+1)}$, for all $m > 0$, such that for all $i \in \omega$, $\vec j \in \omega^{m}$ and $r \in \SSS^{(n)}_{A}$, if $\lambda t.\ \theta^{(m+1)}_{i}(\vec j, t)$ enumerates the extended graph of function $r$, then $r = \theta^{(n)}_{d(i, \vec j)}$.
\end{description}

\begin{theorem}\label{thm-qgnequiv}
Let $\theta^{(n)}$ satisfy Condition~(QGN~I), for every $n > 0$. Then the following equivalences hold:
\[
(\ref{thm-qgnequiv-1}) \Leftrightarrow (\ref{thm-qgnequiv-3}) \quad \text{and} \quad (\ref{thm-qgnequiv-2}) \Leftrightarrow (\ref{thm-qgnequiv-4}),
\]
where
\begin{enumerate}

\item \label{thm-qgnequiv-1}
$\theta^{(n)}$ meets Condition~(QGN~II). 

\item \label{thm-qgnequiv-2}
$\theta^{(n)}$ meets Condition~(QGN~II), for all $n >0$.

\item \label{thm-qgnequiv-3}
$\theta^{(n)}$ meets Condition~(QGN~E). 


\item \label{thm-qgnequiv-4}
Requirements ~(\ref{thm-qgnequiv-4-a}) and (\ref{thm-qgnequiv-4-b}) hold, for all $n > 0$:
\begin{enumerate}
\item \label{thm-qgnequiv-4-a}
There is some function $s \in \RRR^{(m+1)}$, for all $m > 0$ , so that for all $\vec y \in \omega^{m}$, $\vec x \in \omega^{n}$ and $i \in \omega$, $\theta^{(n)}_{s(i, \vec y)}(\vec x) = \theta ^{(m+n)}_{i}(\vec y, \vec x)$.

\item \label{thm-qgnequiv-4-b}
There exists a function $g \in \RRR^{(1)}$ such that for all $i \in \omega$ and $r \in \SSS^{(n)}_{A}$, if $\theta^{(1)}_{i}$ enumerates the extended graph of function $r$, then $r = \theta^{(n)}_{g(i)}$.

\end{enumerate}
\end{enumerate}
\end{theorem}
\begin{proof}
$(\ref{thm-qgnequiv-3}) \Rightarrow (\ref{thm-qgnequiv-1})$ holds trivially. We show next that $(\ref{thm-qgnequiv-1}) \Rightarrow (\ref{thm-qgnequiv-3})$. By Lemma~\ref{lem-preunif} there is some $k \in \RRR^{(2)}$ such that $\lambda t.\ k(\pair{i, \vec j}, t)$ enumerates the extended graph of $\lambda a.\ \theta^{(m+1)}_{i}(\vec j, a)$. Let $i \in \omega$, $\vec j \in \omega^{m}$ and $r \in \SSS^{(n)}_{A}$ such that $\lambda a.\ \theta^{(m+1)}_{i}(\vec j, a)$ enumerates the extended graph of function $r$. According to its definition every element of $\egraph(\lambda a.\ \theta^{(m+1)}_{i}(\vec j, a))$ is of the form $\pair{a, 0}$, $\pair{a, \pair{\pair{\vec y}, 0}+1}$ or $\pair{a, \pair{\pair{\vec y}, r(\vec y)+1}+1}$. Therefore, if we set 
\[
\hat{k}(\pair{i, \vec j}, t) \Def \begin{cases}
						\pair{\pair{0},0} & \text{if $\pi^{(2)}_{2}(k(\pair{i, \vec j}, t)) = 0$}, \\ 
						\pi^{(2)}_{2}(k(\pair{i, \vec j}, t)) \prc 1 & \text{otherwise},
					\end{cases}
\]
then $\hat{k} \in \RRR^{(2)}$. Moreover $\lambda t.\ \hat{k}(\pair{i, \vec j}, t)$ enumerates the extended graph of function $r$. Since $\theta^{(n)}$ satisfies Condition~(QGN~II), there is a function $v \in \RRR^{(1)}$ with $r = \theta^{(n)}_{v(\pair{i, \vec j})}$. Thus, $d \Def \lambda i, \vec j.\ v(\pair{i, \vec j})$ has the property stated in Condition~(QGN~E).

In the same way we obtain that (\ref{thm-qgnequiv-2}) implies (\ref{thm-qgnequiv-4-b}): choose $m = 0$ . Statement~(\ref{thm-qgnequiv-4-a}) is just the \emph{smn}-theorem that has been derived from (\ref{thm-qgnequiv-2}) in Theorem~\ref{thm-smn}. For the remaining implication it suffices to show  $(\ref{thm-qgnequiv-4}) \Rightarrow (\ref{thm-qgnequiv-3})$. Let to this end $i \in \omega$, $\vec j \in \omega^{m}$ and $r \in \SSS^{(n)}_{A}$ so that $\lambda t.\ \theta^{(m+1)}_{i}(\vec j, t)$ enumerates the extended graph of $r$. With Condition~(\ref{thm-qgnequiv-4-a}) we have that $\theta^{(m+1)}_{i}(\vec j, t) = \theta^{(1)}_{s(i, \vec j)}(t)$, from which it follows with (\ref{thm-qgnequiv-4-b}) that $r = \theta^{(n)}_{g(s(i, \vec j))}$. Therefore,  $d \Def g \circ s$ has the properties required in (\ref{thm-qgnequiv-3}).
\end{proof}

For what follows, let $\theta^{(n)}$ be a quasi-G\"odel numbering, for all $n > 0$. With the \emph{smn}-theorem we have shown for a known result of computability theory that it also holds  in the computability theory for $\SSS_A$. Next we want to examine whether substitution is an effective operation. Here we first have to state that $\widehat{\SSS}_A$ is not closed under substitution. Because whenever for $p, r \in \widehat{\SSS}^{(1)}_A$ the set $\set{x}{p(x) \in \dom(r)}$ refrains from being a segment of $\omega$ then $r \circ p \notin \widehat{\SSS}^{(1)}_A$ We therefore introduce a modified substitution. Let $\fun{\widehat{M}^{(n)}_{A}}{\PPP\FFF^{n}}{\PPP\FFF^{n}}$ be defined by
\[
\widehat{M}^{(n)}_{A}(q)(\vec x) \Def \begin{cases}
								q(\vec x) & \text{if there exists $a \in A$ so that $\vec x < a^{(n)}$ and } \\
								& \text{for all $\vec y < a^{(n)}$, $\vec y \in \dom(q)$}, \\
								\text{undefined} & \text{othewise}.
							\end{cases}
\]
Then $\widehat{M}^{(n)}_{A}$ is idempotent, $\widehat{M}^{(n)}_{A}[\PPP\FFF^{(n)}] = \widehat{\SSS}^{(n)}_{A}$ and $\widehat{M}^{(n)}_{A}(\varphi^{(n)}_{i})  = \leftidx{^{A}}\varphi^{(n)}_{i}$. If $$\fun{\subst^{(m,n)}}{\PPP\FFF^{(m)} \times (\PPP\FFF^{(n)})^{m}}{\PPP\FFF^{(n)}}$$ is the usual substitution operation, the modified operation is defined by
\[
\msubst^{(m,n)}_{A} \Def \widehat{M}^{(n)}_{A} \circ \subst^{(m,n)}.
\]
For $r \in \widehat{\SSS}^{(m)}_{A}$ and $p_{1}, \ldots, p_{m} \in \widehat{\SSS}^{(n)}_{A}$, $\msubst^{(m,n)}_{A}(r; p_{1}, \ldots, p_{m})$ agrees with  $$\subst^{(m,n)}(r; p_1 , \ldots, p_m)$$ on the maximal initial segment of $\omega^n$ with an edge length in $A$ that is contained in
 $$\dom(\subst^{(m,n)}_{A}(r; p_{1}, \ldots, p_{m}))$$,  and is undefined, otherwise. Therefore, for all functions $r \in \widehat{\SSS}^{(m)}_A$ and $p_1, \ldots, p_m \in \widehat{\SSS}^{(n)}_A$ with $\bigtimes\nolimits_{\nu = 1}^m p_{\nu}(\bigcap\nolimits_{\sigma=1}^m \dom(p_\sigma)) \subseteq \dom(r)$,
\[ 
\msubst^{(m,n)}_{A}(r; p_{1}, \ldots, p_{m}) = \subst^{(m,n)}(r; p_{1}, \ldots, p_{m}).
\]
In this case, $\dom( \subst^{(m,n)}(r; p_{1}, \ldots, p_{m})) = \bigcap\nolimits_{\sigma=1}^m \dom(p_\sigma))$. Since the domains of the functions $p_{\sigma}$ are comparable with respect to set inclusion, the intersection is again an initial segment of $\omega^{n}$ with an edge length in $A$. In particular, the modified substitution agrees with the usual substitution for all total functions. This justifies the introduction of $\msubst^{(m,n)}_A$, as we are essentially concerned with the total functions. Incidentally, it also happens with normal substitution that information gets lost during the substitution process. If for $p,q \in \PPP\FFF^{(1)}$, $q(x) \notin \dom(p)$, then $(p \circ q)(x)$ is undefined. That is, the information coming with $q(x)$ is lost. 
In the case of the modified substitution, in addition all information is lost that comes from points which are not contained in the maximal initial segment of $\omega$ with length in $A$  that is included in $\dom(p \circ q)$. This is in agreement with the algorithms defined in the last section using the Turing program $M^{(n)}_A$.

\begin{theorem}\label{thm-sub}
Let $A \subseteq \omega$ be an infinite c.e.\ set and $\theta^{(n)}$ be a quasi-G\"odel numbering of $\SSS^{(n)}_{A}$, for $n > 0$. Then 
there is a function $\sub \in \RRR^{(m+1)}$ so that
\[
\theta^{(n)}_{\sub(j, j_{1}, \ldots, j_{m})} = \msubst^{(m,n)}_{A}(\theta^{(m)}_{j}; \theta^{(n)}_{j_{1}}, \ldots, \theta^{(n)}_{j_{m}}).
\]
\end{theorem}
\begin{proof}
Let $h, \hat{h} \in \RRR^{(1)}$ be enumerations as in Condition~(QGN~I) for $\theta^{(n)}$ and $\theta^{(m)}$, respectively. Moreover, let
\begin{multline*}
\widehat{Q}(\pair{j, j_{1}, \ldots, j_{m}}, \pair{b, b_{1}, \ldots, b_{m}}, \pair{\vec x}) \Leftrightarrow \\
\qquad \bigwedge_{\nu=1}^{m} [\pi^{(2)}_{1}(h(b_{\nu})) = \pair{j_{\nu}, \vec x} \wedge   \pi^{(2)}_{2}(h(b_{\nu})) > 0] \wedge \mbox{} \\
 \pi^{(2)}_{2}(\hat{h}(b)) = \pair{j, \pi^{(2)}_{2}(h(b_{1}))\prc 1, \ldots, \pi^{(2)}_{2}(h(b_{m}))\prc 1} \wedge \pi^{(2)}_{2}(\hat{h}(b)) > 0,
 \end{multline*}
 \begin{gather*}
Q(i,t,z) \Leftrightarrow (\exists a \in A_{t})\, \bigwedge_{\nu=1}^{n} \pi^{(n)}_{\nu}(\narg(z)) < a \wedge (\forall \vec x < a^{(n)})\,(\exists c \le t)\, \widehat{Q}(i,c,\pair{\vec x}), \\
f(i,t,z) \Def \pi^{(2)}_{1}(\hat{h}(\pi^{(m+1)}_{1}(\mu c \le t.\ \widehat{Q}(i, c, \pi^{(2)}_{1}(z))))). 
\end{gather*}
By now applying Lemma~\ref{lem-meth1} we obtain a function $k \in R^{(2)}$ such that $\lambda t.\ k(\pair{j, j_{1}, \ldots, j_{m}}, t)$ enumerates the extended graph of $\msubst^{(m,n)}_{A}(\theta^{(m)}_{j}; \theta^{(n)}_{j_{1}}, \ldots, \theta^{(n)}_{j_{m}})$. Since $\theta^{(n)}$ meets Condition~(QGN~II), there is some $v \in \RRR^{(1)}$ with 
\[
\msubst^{(m,n)}_{A}(\theta^{(m)}_{j}; \theta^{(n)}_{j_{1}}, \ldots, \theta^{(n)}_{j_{m}}) = \theta^{(n)}_{v(\pair{j, j_{1}, \ldots, j_{m}})}.
\]
By setting $\sub \Def \lambda j, \vec j.\ v(\pair{j, \vec j})$ we are therefore done.
\end{proof}

A central result of computability theory is the recursion respectively fixed point theorem. Our next goal is derive this theorem for the function classes and numberings considered in this paper. To this end we have to consider the family $(\theta^{(n)}_{\theta^{(m+1}_{i}(i, \vec y)}(\vec y))_{(i, \vec y) \in \omega^{m+1}}$. Since numbering $\theta^{(n)}$ is only defined on defined indices in $\omega$, this family is not well defined. With help of the enumeration theorem it can however be extended to situations in which index functions may be not defined. We then obtain
\[
\theta^{(n)}_{\theta^{(m+1)}_{i}(i, \vec y)}(\vec x) = \begin{cases}
										\theta^{(n)}_{\theta^{(m+1)}_{i}(i, \vec y)}(\vec x) & \text{if $\theta^{(m+1)}_{i}(i, \vec y)\conv$}, \\
										\text{undefined} & \text{otherwise}.
									   \end{cases}
\]
Thus, $\theta^{(n)}_{\theta^{(m+1}_{i}(i, \vec y)} \in \SSS^{(n)}_{A}$, for every choice of $(i, \vec y) \in \omega^{m+1}$.

\begin{lemma}\label{lem-helpfp}
Let $A \subseteq \omega$ be an infinite c.e.\ set and $\theta^{(n)}$ be a quasi-G\"odel numbering of $\SSS^{(n)}_{A}$, for $n > 0$. Then
there exists $g \in \RRR^{(m+1)}$ so that
\[
\theta^{(n)}_{g(i, \vec y)} = \theta^{(n)}_{\theta^{(m+1)}_{i}(i, \vec y)}.
\]
\end{lemma}
\begin{proof}
Again, let $h, \hat{h} \in \RRR^{(1)}$ be enumerations as in Condition~(QGN~I) for $\theta^{(n)}$ and $\theta^{(m+1)}$, respectively. Moreover, define
\begin{multline*}
\widehat{Q}(\pair{i, \vec y}, \pair{a, b}, z) \Leftrightarrow \pi^{(2)}_{1}(\hat{h}(a)) = \pair{i,i,\vec y} \wedge \pi^{(2)}_{2}(\hat{h}(a)) >  0 \wedge \mbox{} \\
\pi^{2}_{2}(h(b)) > 0 \wedge \pi^{2}_{1}(h(b)) = \pair{\pi^{2}_{2}(\hat{h}(a)) \prc 1} \ast \narg(z),
\end{multline*}
\begin{gather*}
Q(\pair{i, \vec y}, t, z) \Leftrightarrow (\exists c \le t)\, \widehat{Q}(\pair{i, \vec y}, c, z), \\
f(\pair{i, \vec y}, t, z) \Def \pi^{(2)}_{2}(h(\pi^{2}_{2}(\mu c \le t.\ \widehat{Q}(\pair{i, \vec y}, c, z)))).
\end{gather*}
Then by applying Lemma~\ref{lem-meth1} we obtain a function $k \in \RRR^{(2)}$ so that $\lambda t.\ k(\pair{i, \vec y}, t)$ enumerates the extended graph of $\theta^{(n)}_{\theta^{(m+1)}_{i}(i, \vec y)}$. Now, using Condition~(QGN~II) for $\theta^{(n)}$ we have that there is some $v \in \RRR^{(1)}$ with
\[
\theta^{(n)}_{v(i, \vec y)} = \theta^{(n)}_{\theta^{(m+1)}_{i}(i, \vec y)}.
\]
Set $g(i, \vec y) \Def \lambda i, \vec y.\ v(\pair{i, \vec y})$. 
\end{proof}

This more technical result allows to derive the recursion theorem.

\begin{theorem}[Recursion Theorem]\label{thm-fp}
Let $A \subseteq \omega$ be an infinite c.e.\ set and $\theta^{(n)}$ be a quasi-G\"odel numbering of $\SSS^{(n)}_{A}$, for $n > 0$. Moreover, let $f \in \RRR^{(m+1}$. Then there is a function $e_{f} \in \RRR_{(m)}$ so that
\[
\theta^{(n)}_{f(e_{f}(\vec y), \vec y)} = \theta^{(n)}_{e_{f}(\vec y)}.
\]
\end{theorem}
\begin{proof}
Let $g \in \RRR^{(m+1)}$ be as in Lemma~\ref{lem-helpfp}. Then there is some index $j \in \omega$ with $\theta^{(m+1)}_{j}(i, \vec y) = f(g(i, \vec y), \vec y)$. Set $e_{f}(\vec y) \Def g(j, \vec y)$. Then $e_{f} \in \RRR^{(m)}$ and 
\[
\theta^{(n)}_{e_{f}(\vec y)} = \theta^{(n)}_{g(j, \vec y)} = \theta^{(n)}_{\theta^{m+1}_{j}(j, \vec y)} = \theta^{(n)}_{f(g(j), \vec y), \vec y)} = \theta^{(n)}_{f(e_{f}(\vec y), \vec y)}. \qedhere
\]
\end{proof}

As we see, the proof of this theorem follows essentially the same idea as in the case of classical computability theory (cf.\ e.g.\ \cite{ro67}). This will be the case in several of the subsequent results. In each case, however, there are certain auxiliary functions such as function $g$ above, the existence of which has to be demonstrated in another way as in the known theory. 

The above result is the fixed point version of the recursion theorem. With help of the \emph{smn}-theorem we now obtain Kleene's version~\cite{kl52}.

\begin{corollary}\label{cor-reckl}
Let $A \subseteq \omega$ be an infinite c.e.\ set and $\theta^{(n)}$ be a quasi-G\"odel numbering of $\SSS^{(n)}_{A}$, for $n > 0$. Moreover,
let $r \in \RRR^{(n+1)}$. Then there is some index $c_{r} \in \omega$ with
\[
r(c_{r}, \vec x) = \theta^{(n)}_{c_{r}}(\vec x).
\]
\end{corollary}

Similarly, we obtain an effective version of this corollary which says the index $c_{r}$ can be computed from an index of $r$.

\begin{corollary}\label{cor-effkl}
Let $A \subseteq \omega$ be an infinite c.e.\ set and $\theta^{(n)}$ be a quasi-G\"odel numbering of $\SSS^{(n)}_{A}$, for $n > 0$. Then
there is a function $q \in \RRR^{(1)}$ such that
\[
\theta^{(n+1)}_{i}(q(i), \vec x) = \theta^{(n)}_{q(i)}(\vec x).
\]
\end{corollary}

Next, we will derive an effective version of the fixed point formulation above.  It says that the fixed point $e_{f}$ in Theorem~\ref{thm-fp} can be computed from an index of $f$.

\begin{theorem}\label{thm-effp}
Let $A \subseteq \omega$ be an infinite c.e.\ set and $\theta^{(n)}$ be a quasi-G\"odel numbering of $\SSS^{(n)}_{A}$, for $n > 0$. Then
there is a function $e \in \RRR^{(m+1)}$ such that all $i \in \omega$ with $\theta^{(m+1)}_{i}$ being total,
\[
\theta^{(n)}_{\theta^{(m+1)}_{i}(e(i, \vec y), \vec y)} = \theta^{(n)}_{e(i, \vec y)}.
\]
\end{theorem}
\begin{proof}
Let $g \in \RRR^{(m+1)}$ be as in Lemma~\ref{lem-helpfp}, and $i \in \omega$ such that $\theta^{(m+1)}_{i}$ is total. Since for total functions the modified substitution defined above agrees with usual substitution, it follows from Theorem~\ref{thm-sub} that there is a function $q \in R^{(1)}$ so that 
\[
\theta^{(m+1)}_{q(i)}(j, \vec y) = \theta^{(m+1)}_{i}(g(j, \vec y), \vec y).
\]
Set $e(j, \vec y) \Def g(q(j), \vec y)$. Then $e \in \RRR^{(m+1)}$ and 
\begin{align*}
\theta^{(n)}_{e(i, \vec y)} 
&= \theta^{(n)}_{\theta^{(m+1)}_{q(i)}(q(i), \vec y)} \\
&= \theta^{(n)}_{\theta^{(m+1)}_{i}(g(q(i), \vec y), \vec y)} \\
&= \theta^{(n)}_{\theta^{(m+1}_{i}(e(i, \vec y). \vec y)}. \tag*{\qedhere}
\end{align*}
\end{proof}

As a further consequence of the recursion theorem we obtain that the index functions used in this section can be chosen as one-to-one. The padding lemma we derive next will be needed here.

\begin{theorem}[Padding Lemma]\label{thm-pad}
Let $A \subseteq \omega$ be an infinite c.e.\ set and $\theta^{(n)}$ be a quasi-G\"odel numbering of $\SSS^{(n)}_{A}$, for $n > 0$. Then
$\theta^{(n)}$ has a one-to-one padding function, that is, a one-to-one function $p \in \RRR^{(2)}$ so that
\[
\theta^{(n)}_{p(i,j)} = \theta^{(n)}_{i}.
\]
\end{theorem}
\begin{proof}
Let $s \in \RRR^{(2)}$ be an \emph{smn}-function for the case $m = 2$ and $r \in \RRR^{(4)}$ defined by
\[
r(c,j,x,b) \Def \begin{cases}
			0 & \text{if there exists $a < j$ with $s(c,a) = s(c,j)$}, \\
			1 & \text{if for all $a < j$, $s(c,a) \neq s(c,j)$, $j < x$, and $s(c,j) = s(c,x)$}. \\
			b & \text{otherwise}.
		\end{cases}
\]
Since the modified substitution of functions in $\widehat{\SSS}_{A}$ into total functions agrees with the usual one, it follows with Theorem~\ref{thm-sub} that there is a function $g \in \RRR^{(1)}$ so that
\[
\theta^{(n+2)}_{g(i)}(c, j, \vec x) =r(c, j, x_{1}, \theta^{(n)}_{i}(\vec x)).
\]
It then follows with Corollay~\ref{cor-effkl} that there exists $k \in \RRR^{(1)}$ with
\[
\theta^{(n+2)}_{g(i)}(k(i), j, \vec x) = \theta^{(n+1)}_{k(i)}(j, \vec x).
\]
Now, assume that $\lambda j.\ s(k(i), j)$ is not one-to-one, let $a$ be minimal with the property that
\[
(\exists j > a)\, s(k(i), a) = s(k(i), j),
\]
and choose a $j$ with the property just stated, say $\hat{\jmath}$. Then,
\begin{align*}
0
&= \theta^{(n+2)}_{g(i)}(k(i), \hat{\jmath}, \hat{\jmath}^{(n)}) \\
&= \theta^{(n+1)}_{k(i)}(\hat{\jmath}, \hat{\jmath}^{(n)}) \\
&= \theta^{(n)}_{s(k(i), \hat{\jmath})}(\hat{\jmath}^{(n)}) \\
&=  \theta^{(n)}_{s(k(i), a)}(\hat{\jmath}^{(n)}) \\
&= \theta^{(n+1)}_{k(i)}(a, \hat{\jmath}^{(n)}) \\
&= \theta^{(n+2)}_{g(i)}(k(i), a, \hat{\jmath}^{(n)}) \\
&=1.
\end{align*}
Thus, $\lambda j.\ s(k(i), j)$ is one-to-one. Define
\[
p(i,j) \Def \begin{cases} 
		s(k(0),0) & \text{if $(i,j) = (0,0)$}, \\
		s(k(i), \mu z.\ (\forall (a,b) < (i,j))\, s(k(i), z) \neq p(a,b)) & \text{otherwise}.
	  \end{cases}
\]
Note that because $\lambda j.\ s(k(i), j)$ is one-to-one, if $(i, j) \neq (0,0)$, then there is always some $z$ with $s(k(i), z) \neq p(a,b)$, for all $a, b \in \omega$ with $(a, b) < (i, j)$. Therefore $p \in R^{(2)}$. Moreover, $p$ is one-to-one and $\theta^{(n)}_{p(i,j)} = \theta^{(n)}_{i}$.
\end{proof}

\begin{corollary}\label{cor-indinj}
Let $A \subseteq \omega$ be an infinite c.e.\ set and $\theta^{(n)}$ be a quasi-G\"odel numbering of $\SSS^{(n)}_{A}$, for $n > 0$. 
Moreover, let  $f \in \RRR^{(m)}$. Then there is a one-to-one function $\hat{f} \in \RRR^{(m)}$ so that $$\theta^{(m)}_{\hat{f}(\vec y)} = \theta^{(m)}_{f(\vec y)}.$$
\end{corollary}

It follows that that the function $v$ in Requirement~(QGN~II), the function $d$ in Condition~(QGN~E), the \emph{smn}-function, the function $g$ in Statement~(\ref{thm-qgnequiv-4}) of Theorem~\ref{thm-qgnequiv}, the function $\sub$ in Theorem~\ref{thm-sub}, the function $g$ in Lemma~\ref{lem-helpfp}, and hence the function $e_{f}$ in the recursion theorem, the function $q$ in Corollary~\ref{cor-effkl} and the function $e$ in the effective version of the recursion theorem can all be chosen as one-to-one. A consequence of the latter fact is that every recursive definition has infinitely many fixed points.

\begin{theorem}\label{thm-infp}
Let $A \subseteq \omega$ be an infinite c.e.\ set and $\theta^{(n)}$ be a quasi-G\"odel numbering of $\SSS^{(n)}_{A}$, for $n > 0$. Then
there is a one-to-one function $\fix \in \RRR^{(2)}$ so that for all $i \in \omega$ for which $\theta^{(n)}_{i}$ is total, and all $j \in \omega$,
\[
\theta^{(n)}_{\theta^{(1)}_{i}(\fix(i, j))} = \theta^{(n)}_{\fix(i, j)}.
\]
\end{theorem}
\begin{proof}
Let $i \in \omega$ be such that $\theta^{(n)}_{i}$ is total. By Theorem~\ref{thm-sub} there is a one-to-one function $k \in \RRR^{(1)}$ with $\theta^{(2)}_{k(i)}(c, j) = \theta^{(1)}_{i}(c)$. Moreover, as just seen, the function $e \in \RRR^{(2)}$ in Theorem~\ref{thm-effp} can be chosen as one-to-one. Thus, we have that
\[
\theta^{(n)}_{\theta^{(1)}_{i}(e(k(i), j))} = \theta^{(n)}_{\theta^{(2)}_{k(i)}(e(k(i), j), j)} = \theta^{(n)}_{e(k(i), j)}.
\]
It therefore suffices to define $\fix(i,j) \Def e(k(i), j)$.
\end{proof}

With help of the recursion theorem we can now show that it is not decidable whether a function is defined on an initial segment only, or is total, though $\SSS_{A}$ has a simple structure. This and several similar results will be consequences of Rice's theorem~\cite{ri53}. Let to this end, for $X \subseteq \SSS^{(n)}_{A}$,
\[
I_{\theta^{(n)}}(X) \Def \set{i \in \omega}{\theta^{(n)}_{i} \in X}.
\]

\begin{theorem}[Rice]\label{thm-rice}
Let $A \subseteq \omega$ be an infinite c.e.\ set and $\theta^{(n)}$ be a quasi-G\"odel numbering of $\SSS^{(n)}_{A}$, for $n > 0$. Moreover, 
let $C \subseteq \SSS^{(n)}_{A}$. Then $I_{\theta^{(n)}}(C)$ is computable if, and only if, $C = \emptyset$ or $C = \SSS^{(n)}_{A}$.
\end{theorem}
\begin{proof}
If $C = \emptyset$ or $C = \SSS^{(n)}_{A}$, respectively, then $I_{\theta^{(n)}}(C) = \emptyset$ or $I_{\theta^{(n)}}(C) = \omega$ and hence computable. For the converse implication assume that $\emptyset \neq C \neq \SSS^{(n)}_{A}$, but $I_{\theta^{(n)}}(C)$ is computable. Then $\emptyset \neq I_{\theta^{(n)}}(C) \neq \omega$. Let $i \in I_{\theta^{(n)}}(C)$ and $j \in \omega \setminus I_{\theta^{(n)}}(C)$. Set 
\[
f(x) = \begin{cases}
		j & \text{if $x \in I_{\theta^{(n)}}(C)$}. \\
		i & \text{otherwise},
	 \end{cases}
\]
then $f \in \RRR^{(1)}$. By the recursion theorem there is hence some index $a \in \omega$ so that $\theta^{(n)}_{a} = \theta^{(n)}_{f(a)}$. Then $\theta^{(n)}_{a} \in C$, exactly if $\theta^{(n)}_{a} \notin C$, a contradiction.
\end{proof}

As a consequence of this theorem we now obtain that the sets $\set{i \in\omega}{\theta^{(n)}_{i} \in \anf^{(n)}_{A}}$ and $\set{i \in \omega}{\theta^{(n)}_{i} \in \RRR^{(n)}}$ cannot be computable. Initial segment functions can be extended to  total computable functions. As we will see now, this cannot be done effectively.

\begin{theorem}\label{thm-anfext}
Let $A \subseteq \omega$ be an infinite c.e.\ set and $\theta^{(n)}$ be a quasi-G\"odel numbering of $\SSS^{(n)}_{A}$, for $n > 0$. Then 
there is no function $q \in \RRR^{(1)}$ so that, if $\theta^{(n)}_{i} \in \anf^{(n)}_{A}$ then $\theta^{(n)}_{q(i)} \in \RRR^{(n)}$ and 
$\graph(\theta^{(n)}_{i}) \subset \graph(\theta^{(n)}_{q(i)})$.
\end{theorem}
\begin{proof}
Assume that there is a function $q \in \RRR^{1}$ as stated, and let $j$ be a $\theta^{(1)}$-index of $q$. We show that there is a function $v \in \RRR^{(1)}$ so that $\theta^{(n)}_{v(i)} \in \anf^{(n)}_{A}$ and for all $\vec x \le \kappa(i)$, $\theta^{(n)}_{v(i)}(\vec x) = \theta^{(n)}_{i}(\vec x)$. This will help us to construct a function $g \in \RRR^{(1)}$ such that
\begin{equation}\label{eq-gcase}
\theta^{(n)}_{g(i)}(\vec x) = 
\begin{cases}
	\theta^{(n)}_{i}(\vec x) +1 & \text{if, in a simultaneous search, $\pair{\pair{i, \vec x}, \theta^{(n)}_{i}(\vec x) +1}$ is found} \\
	& \text{in $\egraph(\lambda a, \vec y.\ \theta^{(n)}_{a}(\vec y))$ not later than  $\pair{\pair{j,i},q(i)+1}$ is} \\
	&\text{found in $\egraph(\lambda a, y.\ \theta^{(1)}_{a}(y))$,} \\
	\theta^{(n)}_{q(v(i))}(\vec x) +1 & \text{if  $\pair{\pair{j,i},q(i)+1}$  is found earlier than $\pair{\pair{i, \vec x}, \theta^{(n)}_{i}(\vec x) +1}$}.
\end{cases}
\end{equation}
Since $q$ is total, always one of the cases holds. If the first case holds, $\theta^{(n)}_{i}(\vec x)$ is defined. As $\theta^{(n)}_{v(i)}$ is an initial segment function, we have that also $\theta^{(n)}_{q(v(i))}$ is total. Therefore, in the second case, $\theta^{(n)}_{q(v(i))}(\vec x)$ is defined as well. Thus, $\theta^{(n)}_{g(i)} \in\RRR^{(n)}$. By applying the recursion theorem we hence obtain an index $\hat{\imath}$ with $\theta^{(n)}_{\hat{\imath}} = \theta^{(n)}_{g(\hat{\imath})}$. Then $\theta^{(n)}_{\hat{\imath}} \in \RRR^{(n)}$.

Now, consider $\theta^{(n)}_{\hat{\imath}}(\kappa(\hat{\imath}))$, and suppose that in (\ref{eq-gcase}) the first case holds. Then
\[
\theta^{(n)}_{\hat{\imath}}(\kappa(\hat{\imath})) = \theta^{(n)}_{\hat{\imath}}(\kappa(\hat{\imath})) +1.
\]
Therefore, the second case must hold. By the properties of $v$, we have that for all $\vec x \le \kappa(\hat{\imath})$, $\theta^{(n)}_{v(\hat{\imath})}(\vec x) = \theta^{(n)}_{\hat{\imath}}(\vec x)$. Moreover, $\theta^{(n)}_{v(\hat{\imath})}$ is an initial segment function and hence $\theta^{(n)}_{q(v(\hat{\imath}))}$  is a total extension of $\theta^{(n)}_{v(\hat{\imath})}$. Because $\kappa(\hat{\imath}) \in \dom(\theta^{(n)}_{v(\hat{\imath})})$, we have that $\theta^{(n)}_{q(v(\hat({\imath}))}(\kappa(\hat{\imath})) = \theta^{(n)}_{v(\hat{\imath})}(\kappa(\hat{\imath}))$. Thus,
\[
\theta^{(n)}_{\hat{\imath}}(\kappa(\hat{\imath})) = \theta^{(n)}_{q(v(\hat{\imath}))}(\kappa(\hat{\imath}))+1 = \theta^{(n)}_{v(\hat{\imath})}(\kappa(\hat{\imath})) +1 = \theta^{(n)}_{\hat{\imath}}(\kappa(\hat{\imath})) +1.
\]
This shows that there cannot exist a function $q$ with the stated properties.

Next, we will consider the construction of the functions $v$ and $g$ used in the proof above. We will start with the construction of $v$. Let to this end $h \in \RRR^{(1)}$ be an enumeration of the extended graph of the universal function of $\theta^{(n)}$ and $\er(i)$ be the first $a \in A$ with respect to a fixed enumeration of $A$ so that $\kappa(i) < a^{(n)}$. Moreover, let 
\begin{gather*}
\widehat{Q}(i, b, z) \Leftrightarrow \pi^{(2)}_{2}(h(b)) > 0 \wedge \bigwedge_{\nu=1}^{n} \pi^{(n)}_{\nu}(\narg(z)) < \er(i) \wedge \pi^{(2)}_{1}(h(b)) = \pair{i} \ast \narg(z), \\
Q(i, t, z) \Leftrightarrow (\exists b \le t)\, \widehat{Q}(i, b, z),\\
f(i, t ,z) \Def \pi^{(2)}_{2}(h(\mu b \le t.\ \widehat{Q}(i, b, z)))
\end{gather*}
and $k \in \RRR^{(2)}$ be as in Lemma~\ref{lem-meth1}. Then $\lambda t.\ k(i, t)$ enumerates the extended graph of the function
\[
r(\vec x) = \begin{cases}
			\theta^{(n)}_{i}(\vec x) & \text{if $\vec x < (\er(i))^{(n)}$}, \\
			\text{undefined} & \text{otherwise}.
		\end{cases}
\]
Obviously, $r \in \anf^{(n)}_{A}$. Hence, by (QGN~II), there is a function $v \in \RRR^{(1)}$ with $r = \theta^{(n)}_{v(i)}$. As is easily seen, $v$ has the properties mentioned above.

For the construction of function $g$ let $h' \in \RRR^{(1)}$ enumerate the extended graph of the universal function of $\theta^{(1)}$. Moreover, define
\begin{gather*}
\widehat{Q}_{1}(i, b, z) \Leftrightarrow \pi^{(2)}_{2}(h(b)) > 0 \wedge \pi^{(2)}_{1}(h(b)) = \pair{i} \ast \narg(z), \\
\widehat{Q}_{2}(i, b) \Leftrightarrow \pi^{(2)}_{1}(h'(b)) = \pair{j, i} \wedge \pi^{(2)}_{2}(h'(b)) > 0, \\
\widehat{Q}_{3}(i. b, z) \Leftrightarrow \pi^{(2)}_{2}(h(b)) > 0 \wedge \pi^{(2)}_{1}(h(b)) = \pair{q(v(i))} \ast \narg(z), 
\end{gather*}
\begin{multline*}
\widehat{Q}_{4}(i, b, z) \Leftrightarrow \mbox{} \\
\hspace{3em} [\widehat{Q}_{1}(i, b, z) \wedge \neg(\exists c < b)\, \widehat{Q}_{2}(i, c)] \vee [\widehat{Q}_{3}(i, b, z) \wedge \mbox{} \\
  (\exists b' < b)\, [\widehat{Q}_{2}(i, b') \wedge \neg (\exists c \le b')\, \widehat{Q}_{1}(i, c, z)]], 
 \end{multline*}
 \vspace{-6ex}
 \begin{gather*}
Q'(i, t, z) \Leftrightarrow (\exists b \le t)\, \widehat{Q}_{4}(i, b, z), \\
f(i, t, z) \Def (\pi^{(2)}_{2}(h(\mu b \le t.\ \widehat{Q}_{4}(i, b, z)))) +1,
\end{gather*}
and construct function $k \in \RRR^{(2)}$ as in Lemma~\ref{lem-meth1}. Then $\lambda t.\ k(i, t)$ enumerates the extended graph of the function $p$ defined by
\[
p(\vec x) \Def
\begin{cases}
	\theta^{(n)}_{i}(\vec x) +1 & \text{if for some $b$, $\pi^{(2)}_{2}(h(b)) > 0$, $\pi^{(2)}_{1}(h(b)) = \pair{i, \vec x}$ and for no $c < b$,}  \\
						& \text{$h'(c) = \pair{\pair{j, i}, q(i) + 1}$}, \\
	\theta^{(n)}_{q(v(i))}(\vec x) +1 & \text{if for some $b$, $h'(b) = \pair{\pair{j, i}, q(i) +1}$ and for no $c \le b$,} \\
						& \text{$\pi^{(2)}_{2}(h(c)) > 0$ and $\pi^{(2)}_{1}(h(c)) = \pair{i, \vec x}$}.
\end{cases}
\]
As seen above, $p$ is total. Therefore $p \in \SSS^{(n)}_{A}$. Hence, by (QGN~II), there is a function $g \in \RRR^{(1)}$ with $p = \theta^{(n)}_{i}$ which by the definition of $p$ is as required above.
\end{proof}

As consequence of this result it follows that also the edge length of the domain of an initial segment function $\theta^{(n)}_{i}$ cannot be computed from given index $i$.
\begin{corollary}\label{cor-edgel}
Let $A \subseteq \omega$ be an infinite c.e.\ set and $\theta^{(n)}$ be a quasi-G\"odel numbering of $\SSS^{(n)}_{A}$, for $n > 0$. Then
there is no function $p \in \RRR^{(1)}$ such that $p(i)$ is the edge length of $\dom(\theta^{(n)}_{i})$, if $\theta^{(n)}_{i} \in \anf^{(n)}_{A}$.
\end{corollary}
\begin{proof}
Again we assume that there is such a function $p$. Let $j$ be a $\theta^{(1)}$-index of $p$. In what follows we construct a function $q \in \RRR^{(1)}$ so that
\[
\theta^{(n)}_{q(i)}(\vec x) = 
\begin{cases}
	\theta^{(n)}_{i}(\vec x) & \text{if, in a  simultaneous search, $\pair{\pair{i, \vec x}, \theta^{(n)}_{i}(\vec x) +1}$ is not found in} \\
	& \text{$\egraph(\lambda a, \vec y.\ \theta^{(n)}_{a}(\vec y))$ later than $\pair{\pair{j,i},p(i)+1}$ is found in} \\
	& \text{$\egraph(\lambda a, y.\ \theta^{(1)}_{a}(y))$, or $\vec x < (p(i))^{(n)}$,}\\
0 &  \text{if $\vec x \not< (p(i))^{(n)}$ and $\pair{\pair{j,i},p(i)+1}$ is found earlier.} 
\end{cases}
\]
Because $p$ is total, always one of the two cases holds. If in the first case we find $\pair{\pair{i, \vec x}, \theta^{(n)}_{i}(\vec x)}$ not later than the other tuple, it follows that $\theta^{(n)}_{i}(\vec x)$ is defined. The same is true, once we know that $\vec x < (p(i))^{(n)}$. Since then $\vec x \in \dom(\theta^{(n)}_{i})$, by the properties of $p$. It follows that $\theta^{(n)}_{q(i)}(\vec x) \in \RRR^{(n)}$. Moreover, $\theta^{(n)}_{q(i)}$ is an extension of $\theta^{(n)}_{i}$. This contradicts what we have shown in the previous theorem. Hence, there is no such function $p$.

It remains to show how the function $q$ can be constructed. Let to this end, $h, \hat{h}$, respectively, be computable enumerations of the extended graphs of the universal functions of $\theta^{(n)}$ and $\theta^{(1)}$ that exist by Condition~(QGN~I). Moreover, let the relations $\widehat{Q}_{1}$ and $\widehat{Q}_{2}$ be defined as in the proof of the previous theorem. In addition, let
\begin{multline*}
\widehat{Q}_{3}(i,b,z) \Leftrightarrow \mbox{} \\
\hspace{3em} [\widehat{Q}_{1}(i, b, z) \wedge \neg(\exists c < b)\, \widehat{Q}_{2}(i, c)] \vee [(\exists c \le b)\, \widehat{Q}_{2}(i, c) \wedge \mbox{} \\
  \bigwedge_{\nu+1}^{n}1 + \pi^{(n)}_{\nu}(\narg(z)) < \pi^{(2)}_{2}(h'(c))], 
\end{multline*}
\vspace{-5ex}
\begin{gather*}
\widehat{Q}_{4}(i,b,z) \Leftrightarrow \widehat{Q}_{2}(i, b) \wedge \bigvee_{\nu=1}^{n} 1 + \pi^{(n)}_{\nu}(\narg(z)) \ge \pi^{(2)}_{2}(h'(b)), \\
Q(i, t, z) \Leftrightarrow (\exists b \le t)\, [\widehat{Q}_{3}(i, b, z) \vee \widehat{Q}_{4}(i, b, z)], \\
f(i, t, z) \Def
\begin{cases}
	\pi^{(2)}_{2}(h(\mu b \le t.\ \widehat{Q}_{1}(i, b, z))) & \text{if for some $b \le t$, $\widehat{Q}_{3}(i, b, z)$}, \\
	1 & \text{if for some $b \le t$, $\widehat{Q}_{4}(i, b, z)$}, \\
	0 & \text{otherwise}.
\end{cases}
\end{gather*}

Note that, if for some $b \in \omega$, $\widehat{Q}_{4}(i, b, z)$ holds then $\widehat{Q}_{1}(i, c, z)$, and hence also $\widehat{Q}_{3}(i, c, z)$ will not apply, for all $c \in \omega$. The reason is that $\theta^{(n)}_{i}(\vec x)$ is only defined if $\vec x < (p(i))^{(n)}$. Now,  use the relations and the function $f$ just defined and apply Lemma~\ref{lem-meth1}. Then we obtain a function $k \in \RRR^{(2)}$ such that $\lambda t.\ k(i, t)$ enumerates the extended graph of function $r$ defined by
\[
r(\vec x) \Def
\begin{cases}
	\theta^{(n)}_{i}(\vec x) & \text{if for some $b \in \omega$, $\pi^{(2)}_{2}(h(b)) > 0$, $\pi^{(2)}_{1}(h(b)) = \pair{i, \vec x}$}  \\
					& \text{and for no $c < b$, $h'(c) = \pair{\pair{j,i}, p(i) + 1}$, or $\vec x < (p(i))^{(n)}$}, \\
	0 & \text{if for some $b \in \omega$, $h'(b) = \pair{\pair{j,i}, p(i) + 1}$ and $\vec x \not< (p(i))^{(n)}$}.
\end{cases}
\]
As seen above, $r \in \RRR^{(2)}$. Because of Condition~(QGN~II) there is thus some $q \in \RRR^{(1)}$ with $r = \theta^{(n)}_{q(i)}$. Then $q$ is as required.
\end{proof}

Since the effective version of the recursion theorem holds for $\theta^{(n)}$, it follows by a result of Ershov~\cite[Satz 9]{er73} that $\theta^{(n)}$ is precomplete which means that for every partial computable function $r \in \PPP^{(1)}$ there is a total computable function $g \in \RRR^{(1)}$ so that for all $i \in \dom(r)$,
\[
\theta^{(n)}_{r(i)} = \theta^{(n)}_{g(i)}.
\]
Consequently, we would not be able to obtain positive results in Theorem~\ref{thm-anfext} and Corollary~\ref{cor-edgel} by working in classical computability theory. In that case one has more algorithms at hand. However, because of the precompleteness of $\theta^{(n)}$ there also cannot exist any partial computable functions $q$ and $p$ with the properties stated in both results.

\section{Computably enumerable sets}\label{sec-ce}

As usual, we call a set $C \subseteq \omega$ \emph{computably enumerable (c.e.)} if there is a function $f \in \RRR^{(1)}$ with $C = \range(f)$ or $C$ is empty. Then, of course, it follows in the usual way that a set of natural numbers is computable if and only if the set and its complement are both c.e.\footnote{Note that this result holds only classically, not constructively, as we cannot decide whether a computable set is empty of not.} 

\begin{theorem}\label{thm-eq}
The following five statements are equivalent:
\begin{enumerate}

\item \label{pn-eq-1}
$C$ is c.e.

\item \label{pn-eq-2}
$C = \range(r)$, for some $r \in \SSS^{(1)}_{A}$.

\item \label{pn-eq-3}
For some $f \in \RRR^{(1)}$, $C = \set{a}{a+1 \in \range{f}}$.

\item \label{pn-eq-4}
For some computable $B \subseteq \omega$, $C = \set{a}{(\exists i)\, \pair{a, i} \in B}$.

\item \label{pn-eq-5}
For some c.e.\ $B \subseteq \omega$, $C = \set{a}{(\exists i)\, \pair{a, i} \in B}$.

\end{enumerate}
\end{theorem}
\begin{proof}
The proof of $(\ref{pn-eq-1}) \Rightarrow (\ref{pn-eq-2})$ is obvious: If $C$ is empty let $r$ be the nowhere defined function. In the other case there is a function $f \in \RRR^{(1)}$ with $C = \range(f)$. Then choose $r = f$. 

Next, we show that  $(\ref{pn-eq-2}) \Rightarrow (\ref{pn-eq-3})$. If $C$ is empty set $f \Def \lambda x.\ 0$. In case $C$ is not empty but finite, say $C = \{ a_{0}, \ldots, a_{m} \}$, define
\[
f(x) \Def \begin{cases}
		a_{x} +1 & \text{if $x < m$}, \\
		a_{m} +1 & \text{otherwise}.
	     \end{cases}
\]
If, finally, $C$ is infinite, then the function $r \in \SSS^{(1)}_{A}$ with $C = \range(r)$ needs be total. So, set $f \Def \lambda x.\ r(x) +1$.

To show $(\ref{pn-eq-3}) \Rightarrow  (\ref{pn-eq-4})$, set $B \Def \set{\pair{a, i}}{f(i) = a + 1}$. Since $f$ is total computable, it follows as ususally that $B$ is computable.

Since computable sets are c.e., we also have that $(\ref{pn-eq-4}) \Rightarrow (\ref{pn-eq-5})$. So, it remains to show that $(\ref{pn-eq-5}) \Rightarrow (\ref{pn-eq-1})$. If $B$ is empty, the same holds for $C$. Assume that $B$ is not empty. Then $B = \range(g)$, for some $g \in \RRR^{(1)}$. Set $f \Def \pi^{(2)}_{1} \circ g$. Then $f \in \RRR^{(1)}$ and $C = \range(f)$. Thus, $C$ is c.e.\footnote{Similarly, the implications $(\ref{pn-eq-1}) \Rightarrow (\ref{pn-eq-2})$, $(\ref{pn-eq-2}) \Rightarrow (\ref{pn-eq-3})$, and $(\ref{pn-eq-5}) \Rightarrow (\ref{pn-eq-1})$ in Theorem \ref{thm-eq} are not valid constructively.}
\end{proof}

The above result contains the well-known characterisations of c.e.\ sets. Only the characterisation via function domains is missing. The structure of the non-total functions is too restricted in our case. We only have that the domain of every function in $\SSS^{(1)}_{A}$ is c.e. As we will see, however, the central results about c.e.\ sets can still be derived.
Statements~(\ref{pn-eq-4}) and (\ref{pn-eq-5}) are also known as projection theorems. The next result provides the connection between the function studied in this work and the c.e.\ sets.

\begin{proposition}\label{pn-compce}
\begin{enumerate}
\item \label{pn-compce-1}
For $f \in \FFF^{(1)}$,
\[
f \in \RRR^{(1)} \Leftrightarrow \text{$\graph(f)$ is computable} \Leftrightarrow \text{$\graph(f)$ is c.e.}
\]
\item\label{pn-compce-2}
For $r \in \widehat{\SSS}^{(1)}_{A}$,
\[
r \in \SSS^{(1)}_{A} \Leftrightarrow \text{$\graph(r)$ is c.e.}
\]
\end{enumerate}
\end{proposition}
\begin{proof}
We only show (\ref{pn-compce-2}). The other statement is a well-known fact in classical computability theory. Because of statement (\ref{pn-compce-1}) we only have to consider the case that $r \in \anf^{(1)}_{A}$. Then $r \in \SSS^{(1)}_{A}$ and $\graph(r)$ is a finite set. Therefore the equivalence holds.
\end{proof}

Next, we will derive the well-known closure properties of the class of c.e.\ sets. In the literature some of them are usually shown by using the function domain characterisation. Here, we will have to give other proofs. For completeness reasons, we include proofs of all statements.

\begin{theorem}\label{thm-closce}
For $r \in \SSS^{(1)}_{A}$ and $B, C \subseteq \omega$, if $B$ and $C$ are c.e.\ then so are $B \cap C$, $B \cup C$, $\pair{B, C}$, $r^{{-1}}[B]$ and $r[B]$.
\end{theorem} 
\begin{proof}
By Theorem~\ref{thm-eq} there are functions $g, k \in \RRR^{(1)}$ with $B = \set{a}{a+1 \in \range(g)}$ and $C = \set{a}{a+1 \in \range(k)}$. set
\begin{gather*}
f_{1}(x) \Def \begin{cases}
			g(x) & \text{if for some $y \le x$, $g(x) = k(y)$}, \\
			k(x) & \text{if for some $y < x$, $k(x) = g(y)$}, \\
			0 & \text{otherwise}, 
		\end{cases}\\
f_{2}(x) \Def \begin{cases}
			g(x) & \text{if $x$ is even}, \\
			k(x) & \text{otherwise}, 
		\end{cases}\\
f_{3}(x) \Def \begin{cases}
			\pair{g(x) \prc 1, k(x) \prc 1}+1 & \text{if $g(x) > 0$ and $k(x) > 0$}, \\
			0 & \text{otherwise}.
		\end{cases}
\end{gather*}
Then $f_{1}, f_{2}, f_{3} \in \RRR^{(1)}$. Moreover, $B \cap C = \set{a}{a \in \range(f_{1})}$, $B \cup C = \set{a}{a+1 \in \range(f_{2})}$ and $\pair{B, C} = \set{a}{a+1 \in \range(f_{3})}$. Hence, these sets are c.e. The computable enumerability of $r^{-1}[B]$ follows with the Projection Theorem~\ref{thm-eq}(\ref{pn-eq-5}), since
\begin{align*}
r^{-1}[B]
&= \set{x \in \omega}{(\exists y)\, \pair{x, y} \in \graph(r) \wedge y \in B} \\
&= \set{x \in \omega}{(\exists y)\, \pair{x, y} \in \graph(r) \cap \pair{\omega, B}}.
\end{align*}
In the same way we obtain that $r[B]$ is c.e. 
\end{proof}

As we have seen in Theorem~\ref{thm-eq}, the c.e.\ sets are the ranges of functions in $\SSS^{(1)}_{A}$. This allows us to introduce a numbering $W^{A}$ of all c.e.\ sets. Let $\theta$ be a quasi-G\"odel numbering of $\SSS^{(1)}_{A}$ and define
\[
W^{A}_{i} \Def \range(\theta_{i}).
\]
Since $\theta$ is a quasi-G\"odel numbering, $W^{A}$ satisfies the subsequent normal form theorem.

\begin{theorem}[Enumeration Theorem] \label{thm-enumce}
There is a computable set $B \subseteq \omega$ such that for all $i \in \omega$,
\[
W^{A}_{i} = \set{x \in \omega}{(\exists t)\, \pair{i, x, t} \in B}.
\]
\end{theorem}
\begin{proof}
By Condition~(QGN~I) there is a computable enumeration $h \in \RRR^{(1)}$ of the extended graph of the universal function of $\theta$. The we have that
\[
x \in W^{A}_{i} \Leftrightarrow (\exists t)\, \pi^{(2)}_{1}(\pi^{(2)}_{1}(h(t))) = i \wedge \pi^{(2)}_{2}(h(t)) = x+1.
\]
Set $B \Def \set{\pair{i, x, t}}{\pi^{(2)}_{1}(\pi^{(2)}_{1}(h(t))) = i \wedge \pi^{(2)}_{2}(h(t)) = x+1}$. Then $B$ has the asserted properties.
\end{proof}

The result strengthens the Projection Theorem~\ref{thm-eq}(\ref{pn-eq-4}): the c.e.\ sets can be uniformly obtained from the recursive sets by applying a projection. As follows from the above proof we moreover have

\begin{corollary}\label{cor-unien}
The set $\set{\pair{i, x}}{x \in W^{A}_{i}}$ is c.e.
\end{corollary}

Our next aim is to show that for the closure operations in Theorem~\ref{thm-closce} there are corresponding computable index operations.  As in the last section, we have to construct functions that enumerate the extended graph of another function. These other functions are now enumerations of c.e.\ sets. The definition of the graph enumerations is again based on a scheme. which we indicate below. Let to this end $\val{\cdot}$ be an effective coding of all finite sequences of natural numbers that is  one-to-one and onto such that there are total computable functions $\lth$ and $(\cdot)_{\cdot}$ with
\[
\lth(\val{x_{1}, \ldots. x_{m}}) = m \quad\text{and}\quad (\val{x_{1}, \ldots. x_{m}})_{j} = x_{j} \quad\text{$(1 \le j \le m)$}.
\]
Then the \emph{course of values} of function $g \in \RRR^{(2)}$ up to $t$ is defined by
\[
\overline{g}(i, t) \Def \val{g(i, 0), \ldots, g(i, t)}.
\]

\begin{lemma}\label{lem-methce}
Let $f \in \RRR^{(3)}$ and $Q \subseteq \omega^{3}$ be a computable relation with  $f(i, t, z) > 0$ if $Q(i, t, z)$ holds. Moreover, let $k \in \PPP\FFF^{(2)}$ be defined by
\begin{align*}
&k(i, 2t) = \pair{t, 0} \\
&k(i, 2t+1) = g(i, t), \quad\text{where} \\
&g(i, 0) = \pair{0, 0} \\
&g(i, t+1) = 
\begin{cases}
	\pair{0, 0} & \text{if not $Q(i, t+1, \overline{g}(i, t))$ but}\\
	& \text{$\pi^{(2)}_{2}(g(i, t)) = 0$}, \\
	\pair{0, f(i, t+1, \overline{g}(i, t))} & \text{if $Q(i, t+1, \overline{g}(i, t))$ and} \\
	& \text{$\pi^{(2)}_{2}(g(i, t)) = 0$}, \\
	\pair{1+ \pi^{(2)}_{1}(g(i, t)), f(i, t+1, \overline{g}(i, t))} & \text{if $Q(i, t+1, \overline{g}(i, t))$ and} \\
	& \text{$\pi^{(2)}_{2}(g(i, t)) > 0$}, \\
	\pair{1+ \pi^{(2)}_{1}(g(i, t)), \pi^{(2)}_{2}(g(i, t))} & \text{otherwise}.
\end{cases}
\end{align*}
Then $k \in \RRR^{(2)}$ and $\lambda t.\ k(i, t)$ enumerates the extended graph of a unary computable function which is either nowhere defined or total.
\end{lemma}
\begin{proof}
As follows from the definition, $k \in \RRR^{(2)}$. The function $\lambda t.\ g(i, t)$ first lists the value $\pair{0,0}$ and repeats this until $Q(i, t, \overline{g}(i, t-1))$ holds for the first time, and then lists $\pair{0, f(i, t, \overline{g}(i, t-1))}$.  According to the assumption, $f(i, t, \overline{g}(i, t-1)) > 0$ in this case. Therefore, we have for all $t' \ge t$ that $\pi^{(2)}_{2}(g(i,t')) > 0$. Thus, $\range(\lambda t. k(i, t))$ is the extended graph of a total function or the nowhere defined function, depending on whether there is some $t > 0$ with $Q(i, t, \overline{g}(t-1))$, or not.
\end{proof}

\begin{theorem}\label{thm-indoc}
There are functions $\cut, \un, \fpair, \inv, \im \in \RRR^{(2)}$ so that
\begin{enumerate}
\item\label{thm-indoc-1}
$W^{A}_{\cut(i, j)} W^{A}_{i} \cap W^{A}_{j}$,

\item\label{thm-indoc-2}
$W^{A}_{\un(i, j)} = W^{A}_{i} \cup W^{A}_{j}$,

\item\label{thm-indoc-3}
$W^{A}_{\fpair(i, j)} = \pair{W^{A}_{i}, W^{A}_{j}}$,

\item \label{thm-indoc-4}
$W^{A}_{\inv(i, j)} = \theta^{-1}_{i}[W^{A}_{j}]$,

\item\label{thm-indoc-5}
$W^{A}_{\im(i, j)} = \theta_{i}[W^{A}_{j}]$.

\end{enumerate}
\end{theorem}
\begin{proof}
(\ref{thm-indoc-1}) By Condition~(QGN~I) there exists a computable enumeration of the universal function of numbering $\theta$, say $h \in \RRR^{(2)}$. Moreover define
\begin{multline*}
\widehat{Q}(\pair{i, j}, \pair{a, b}, z) \Leftrightarrow  \mbox{} \\
\hspace{2em} \pi^{(2)}_{1}(\pi^{(2)}_{2}(h(a))) = i \wedge \pi^{(2)}_{1}(\pi^{(2)}_{1}(h(b))) = j \wedge \pi^{(2)}_{2}(h(a)) = \pi^{(2)}_{2}(h(b))
\wedge \mbox{} \\
 \pi^{(2)}_{2}((z)_{1}) = 0 \wedge (\forall 1 \le c \le \lth(z))\ \pi^{(2)}_{2}(h(a)) \neq \pi^{(2)}_{2}((z)_{c}), 
\end{multline*} 
\vspace{-1cm}
\begin{gather*}
Q(\pair{i, j}, t, z) \Leftrightarrow (\exists x \le t)\, \widehat{Q}(\pair{i, j}, x, z), \\
f(\pair{i, j}, t, z) \Def \pi^{(2)}_{2}(h(\pi^{(2)}_{1}(\mu x \le t.\ \widehat{Q}(\pair{i, j}, x, z)))).
\end{gather*}
Then it follows for the function $k \in \RRR^{(2)}$ constructed according to Lemma~\ref{lem-methce} that $\lambda t.\ k(\pair{i, j}, t))$  enumerates the extended graph of a function $r \in \RRR^{(1)}$ that enumerates $W^{A}_{i} \cap W^{A}_{j}$. By Condition~(QGN~II) there is therefore a function $v \in \RRR^{(1)}$ with $r = \theta_{v(\pair_{i, j})}$. Thus, it suffices to set $\cut(i, j) \Def v(\pair{i, j})$.

The remaining statements follow in the same way. We only indicate how $Q$ and $f$ have to be chosen in each case.

(\ref{thm-indoc-2}) Set
\begin{multline*}
\widehat{Q}(\pair{i, j}, a, z) \Leftrightarrow [\pi^{(2)}_{1}(\pi^{(2)}_{1}(h(a))) = i \vee \pi^{(2)}_{1}(\pi^{(2)}_{1}(h(a))) = j] \wedge \mbox{} \\
 \pi^{(2)}_{2}((z)_{1}) = 0 \wedge (\forall 1 \le c \le \lth(z))\, \pi^{(2)}_{2}(h(a)) \neq \pi^{(2)}_{2}((z)_{c}),
 \end{multline*}
 \vspace{-1cm}
 \begin{gather*}
 Q(\pair{i, j}, t, z) \Leftrightarrow (\exists a \le t)\, \widehat{Q}(\pair{i, j}, a, z), \\
 f(\pair{i, j}, t, z) \Def \pi^{(2)}(h(\mu a \le t.\ \widehat{Q}(\pair{i, j}, a, z))).
 \end{gather*}

(\ref{thm-indoc-3}) Set
\begin{multline*}
\widehat{Q}(\pair{i, j}, \pair{a, b}, z) \Leftrightarrow \mbox{} \\
\hspace {2em} \pi^{(2)}_{2}(h(a)) > 0 \wedge \pi^{(2)}_{2}(h(b)) > 0 \wedge \pi^{(2)}_{1}(\pi^{(2)}_{1}(h(a))) = i \wedge \pi^{(2)}_{1}(\pi^{(2)}_{1}(h(b))) = j \wedge \mbox{} \\
 (\forall 1 \le c \le \lth(z))\, \pair{\pi^{(2)}_{2}(h(a)) \prc 1, \pi^{(2)}_{2}(h(b)) \prc 1} \neq 1+ \pi^{(2)}_{2}((z_{c})), \hspace{-2em}
\end{multline*}
\vspace{-1cm}
\begin{gather*}
Q(\pair{i, j}, t, z) \Leftrightarrow (\exists x \le t)\, \widehat{Q}(\pair{i, j}, x, z), \\
f(\pair{i, j}, t, z) \Def \pair{\pi^{(2)}_{2}(h(\pi^{(2)}_{1}(\mu x \le t.\ \widehat{Q}(\pair{i, j} x, z)) \prc 1)), \\
\hspace{6cm} \pi^{(2)}_{2}(h(\pi^{(2)}_{2}(\mu x \le t.\ \widehat{Q}(\pair{i, j} x, z)) \prc 1))} +1.
\end{gather*}

(\ref{thm-indoc-4}) Set
\begin{multline*}
\widehat{Q}(\pair{i, j}, \pair{a, b}, z) \Leftrightarrow \mbox{}\\
\hspace{2em}  \pi^{(2)}_{1}(\pi^{(2)}_{1}(h(a))) = i \wedge \pi^{(2)}_{1}(\pi^{(2)}_{1}(h(b))) = j \wedge \pi^{(2)}_{2}(h(a)) = \pi^{(2)}_{2}(h(b)) \wedge \mbox{} \\
 \pi^{(2)}_{2}(h(a)) > 0 \wedge (\forall 1 \le c \le \lth(z))\, 1+ \pi^{(2)}_{2}(\pi^{(2)}_{1}(h(a))) \neq \pi^{(2)}_{2}((z)_{c}),
\end{multline*}
\vspace{-1cm}
\begin{gather*}
Q(\pair{i, j}, t, z) \Leftrightarrow (\exists x \le t)\, \widehat{Q}(\pair{i, j}, x, z), \\
f(\pair{i, j}, t, z) \Def  1+ \pi^{(2)}_{2}(\pi^{(2)}_{1}(h(\pi^{(2)}_{1}(\mu x \le t.\ \widehat{Q}(\pair{i, j}, x, z))))).
\end{gather*}

(\ref{thm-indoc-5}) Set
\begin{multline*}
\widehat{Q}(\pair{i, j}, \pair{a, b}, z) \Leftrightarrow \mbox{} \\
\hspace{1em}  \pi^{(2)}_{1}(\pi^{(2)}_{1}(h(a))) = i \wedge \pi^{(2)}_{1}(\pi^{(2)}_{1}(h(a))) = j \wedge \pi^{(2)}_{2}(\pi^{(2)}_{1}(h(a))) +1 = \pi^{(2)}_{2}(h(b))  \wedge  \mbox{} \\
\pi^{(2)}_{2}((z)_{1}) = 0 \wedge (\forall 1 \le c \le \lth(z))\, \pi^{(2)}_{2}(h(a)) \neq \pi^{(2)}_{2}((z)_{c}),
\end{multline*}
\vspace{-1cm}
\begin{gather*}
Q(\pair{i, j}, t, z) \Leftrightarrow (\exists x \le t)\, \widehat{Q}(\pair{i, j}, x, z), \\
f(\pair{i, j}, t, z) \Def \pi^{(2)}_{1}(h(\pi^{(2)}_{1}(\mu x \le t.\ \widehat{Q}(\pair{i, j}, x, z)))). \tag*{\qedhere}
\end{gather*}
\end{proof}

Let $K^{A} \Def \set{i}{i \in W^{A}_{i}}$ be the \emph{self-reproducibility problem}. 

\begin{theorem}\label{thm-K}
$K^{A}$ is c.e., but not computable.
\end{theorem}
\begin{proof}
Let $h \in \RRR^{2}_{A}$ again enumerate the extended graph of the universal function of $\theta$. Then,
\[
i \in K^{A} \Leftrightarrow i \in \range(\theta_{i}) \Leftrightarrow (\exists a) (\exists t)\, h(t) = \pair{\pair{i, a}, i + 1}.
\]
With the projection theorem one obtains that $K^{A}$ is c.e. It remains to show that $\omega \setminus K^{A}$ is not c.e. Assume to the contrary that $\omega \setminus K^{A}$ is c.e. Then there is some $i \in \omega$ with $\omega \setminus K^{A} = W^{A}_{i}$. 
If $i \in W^{A}_{i}$ then $i \in K^{A} \cap (\omega \setminus K^{A})$, a contradiction. So, $i \notin W^{A}_{i}$, that is, $i \in K^{A}$. Hence, $i \in W^{A}_{i}$, again a contradiction. It follows that $K^{A}$ is not c.e.
\end{proof} 

\begin{corollary}\label{cor-unice}
The set $\set{\pair{i, x}}{x \in W^{A}_{i}}$ is c.e., but not computable.
\end{corollary}
\begin{proof}
As we have seen in Corollary~\ref{cor-unien}, this set is c.e. If it  would be computable, also $\set{\pair{i,i}}{i \in W^{A}_{i}}$ and hence $K^{A}$ would be computable, which is not the case.
\end{proof}

As consequence of Theorem~\ref{thm-K} we next obtain that in terms of $W^A$-indices, given a computable set, one cannot uniformly pass to its complement.

\begin{theorem}\label{thm-cc}
There is no function $\comp \in \RRR^{(1)}$ so that for all $i \in \omega$, if  $W^{A}_{i}$ is computable then $W^{A}_{\comp(i)}$ is its complement. 
\end{theorem}
\begin{proof}
Let $h \in \RRR^{(1)}$ enumerate the set $K^{A}$ and define
\begin{gather*}
Q(i, t, z) \Leftrightarrow (\exists a \le t)\, h(a) = i, \\
f(i, t, z) \Def 1+\pi^{(2)}_{1}((z)_{\lth(z)}).
\end{gather*}
Then the function $k \in \RRR^{(2)}$ constructed as in Lemma~\ref{lem-methce} is such that  for $i \in K^{A}$, $\lambda t.\ k(i, t)$ enumerates the extended graph of the identity function on $\omega$. For all other $i$, the extended graph of the nowhere defined function is enumerated. Let $v \in \RRR^{(1)}$ be as in Condition~(QGN~II) so that $\range(\lambda t.\ k(i, t)) = \egraph(\theta_{v(i)})$. Then, 
\[
W^{A}_{v(i)} = \begin{cases}
				\omega & \text{if $i \in K^{A}$}, \\
				\emptyset & \text{otherwise}.
		      \end{cases}
\]
Now, assume that there is such a function $\comp \in \RRR^{(1)}$. Then, 
\[
i \in \omega \setminus K^{A} \Leftrightarrow W^{A}_{\comp(v(i))} \neq \emptyset \Leftrightarrow (\exists x)\, x \in W^{A}_{\comp(v(i))}.
\]
With Corollary~\ref{cor-unice} and the projection theorem it follows that the complement of $K^{A}$ is c.e. Since $K^{A}$ is c.e. as well, $K^{A}$ is even computable, which is not the case.
\end{proof}

Next, we will derive the single-valuedness theorem (cf.\ \cite{ro67}). A set $C \subseteq \omega$ is called \emph{single-valued} if for every $x \in \omega$ there is at most one $y \in \omega$ so that $\pair{x, y} \in C$. Each single-valued set is thus the graph of a partial function. The single-valuedness theorem states the existence of an enumeration of all single-valued c.e. sets, which can also be considered as an enumeration of all partial functions with c.e. graph. We therefore derive a further version in which an enumeration of those single-valued c.e. sets is constructed that are graphs of the functions in $\SSS^{(1)}_{A}$. For $C \subseteq \omega$ let
\[
\dom(C) = \set{x}{(\exists y)\, \pair{x, y} \in C}.
\]

\begin{theorem}[Single-valuedness Theorem I]\label{thm-sv}
There is a function $\sv \in \RRR^{(1)}$ such that for $i \in \omega$,
\begin{enumerate}

\item \label{thm-sv-1}
$W^{A}_{\sv(i)}$ is single-valued.

\item  \label{thm-sv-2}
$W^{A}_{\sv(i)} \subseteq W^{A}_{i}$.

\item  \label{thm-sv-3}
$\dom(W^{A}_{\sv(i)}) = \dom(W^{A}_{i})$.

\item  \label{thm-sv-4}
If $W^{A}_{i}$ is single-valued then $W^{A}_{sv(i)} = W^{A}_{i}$.

\end{enumerate}
\end{theorem}
\begin{proof}
Let again $h \in \RRR^{(1)}$ enumerate the extended graph of the universal function of $\theta$, and define
\begin{multline*}
\widehat{Q}(i, a, z) \Leftrightarrow \mbox{} \\
\hspace{2em} \pi^{(2)}_{1}(\pi^{(2)}_{1}(h(a))) = i \wedge \pi^{(2)}_{2}(h(a)) > 0 \wedge (\forall 1 \le c \le \lth(z))\, 
[\pi^{(2)}_{2}((z)_{c}) > 0 \Rightarrow \mbox{} \\
 \pi^{(2)}_{1}(\pi^{(2)}_{2}(h(a)) \prc 1) \neq \pi^{(2)}_{1}(\pi^{(2)}_{2}((z)_{c}) \prc 1)], 
\end{multline*}
\vspace{-1cm}
\begin{gather*}
Q(i, t, z) \Leftrightarrow (\exists a \le t)\, \widehat{Q}(i, a, z), \\
f(i, t, z) \Def \pi^{(2)}_{2}(h(\mu a \le t.\ \widehat{Q}(i, a, z))).
\end{gather*}
Let $k  \in \RRR^{(2)}$ be the function constructed as in Lemma~\ref{lem-methce} with respect to relation $Q$ and function $f$. By Condition~(QGN~II) there is then a function $\sv \in \RRR^{(1)}$ with $\range(\lambda t.\ k(i, t)) = \egraph(\theta_{\sv(i)})$. Moreover, it follows from its construction that for all $y < x$, if $\theta_{\sv(i)}(y) \neq \theta_{\sv(i)}(x)$ then also $\pi^{(2)}_{1}(\theta_{\sv(i)}(y)) \neq \pi^{(2)}_{1}(\theta_{\sv(i)}(x))$. Thus, $W^{A}_{\sv(i)}$ is single-valued. In addition, $W^{A}_{\sv(i)} \subseteq W^{A}_{i}$ and $\dom(W^{A}_{\sv(i)}) = \dom(W^{A}_{i})$. If $W^{A}_{i}$ is single-valued then $\theta_{\sv(i)}$ is an enumeration of $W^{A}_{i}$. Hence, $W^{A}_{\sv(i)} = W^{A}_{i}$ in this case.
\end{proof}

\begin{theorem}[Single-valuedness Theorem II]\label{thm-sg}
There is a function $\sg \in \RRR^{(1)}$ such that for $i \in \omega$,
\begin{enumerate}

\item \label{thm-sg-1}
$W^{A}_{\sg(i)}$ is single-valued.

\item  \label{thm-sg-2}
$W^{A}_{\sg(i)} \subseteq W^{A}_{i}$.

\item  \label{thm-sv-g}
$\dom(W^{A}_{\sg(i)})$ is either empty, equal to $\omega$, or an initial segment of $\omega$ with length in $A$.

\item  \label{thm-sg-4}
If $W^{A}_{i}$ is single-valued and $\dom(W^{A}_{i})$ either empty, equal to $\omega$, or an initial segment of $\omega$ with length in $A$, then $W^{A}_{sg(i)} = W^{A}_{i}$.

\end{enumerate}
\end{theorem}
\begin{proof}
Let $h \in \RRR^{(1)}$ enumerate the extended graph of the universal function of $\theta$, and define
\begin{align*}
\widehat{Q}(i, a, t, z)  & \mbox{} \Leftrightarrow \mbox{} \\
&\hspace{-3em}\pi^{(2)}_{1}(\pi^{(2)}_{1}(h(a))) = i \wedge \pi^{(2)}_{i} > 0 \wedge [[\pi^{(2)}_{2}((z)_{\lth(z)}) = 0 \wedge 
\pi^{(2)}_{1}(\pi^{(2)}_{2}(h(a)) \prc 1) = 0] \vee \mbox{} \\
& [\pi^{(2)}_{2}((z)_{\lth(z)}) = 0 \wedge  \pi^{(2)}_{1}(\pi^{(2)}_{2}(h(a)) \prc 1) = 1+\pi^{(2)}_{1}(\pi^{(2)}_{2}((z)_{\lth(z))}) \prc 1)]] \wedge \mbox{} \\
& (\exists c \in A_{t})\, [\pi^{(2)}_{1}(\pi^{(2)}_{2}(h(a)) \prc 1) < c \wedge (\forall x < c) (\exists b \le t)\, \pi^{(2)}_{1}(\pi^{(2)}_{1}(h(b))) = i \wedge \mbox{} \\
&\pi^{(2)}_{2}(h(b)) > 0 \wedge \pi^{2}_{1}(\pi^{(2)}_{2}(h(b)) \prc 1) = x],
\end{align*}
\vspace{-1cm}
\begin{gather*}
Q(i, t, z) \Leftrightarrow (\exists a \le t)\, \widehat{Q}(i, a, t ,z), \\
f(i, t, z) \Def \pi^{(2)}_{2}(h(\mu a \le t.\ \widehat{Q}(i, a, t, z))).
\end{gather*}
Furthermore, let $k \in \RRR^{(2)}$ be constructed according to Lemma~\ref{lem-methce} and $\sg \in \RRR^{(1)}$ the function existing by Condition~(QGN~II) so that $\lambda t.\ k(i, t)$ enumerates the extended graph of $\theta_{\sg(i)}$. If $\dom(W^{A}_{i})$ contains an initial segment of $\omega$ with length in $A$, then it follows from the construction of $k$ that $\pi^{(2)}_{1}(\theta_{\sg(i)}(0)) = 0$. Moreover, if $\theta_{\sg(i)}(x+1) \neq \theta_{\sg(i)}(x)$, then $\pi^{(2)}_{1}(\theta_{\sg(i)}(x+1)) = 1+\pi^{(2)}_{1}(\theta_{\sg(i)}(x))$. In addition, $\pi^{(2)}_{1}(\theta_{\sg(i)}(x))$ is smaller than the maximal length of the initial segments of $\omega$ that have a length in $A$ and are contained in $\dom(W^{A}_{i})$, if such a length exists at all. Therefore, $W^{A}_{\sg(i)}$ is single-valued and $\dom(W^{A}_{\sg(i)})$ either empty, equal to $\omega$, or an initial segment of $\omega$ with length in $A$. Furthermore, $W^{A}_{\sg(i)} \subseteq W^{A}_{i}$. Hence, if $W^{A}_{i}$ is single-valued so that $\dom(W^{A}_{i})$ is either empty, equal to $\omega$, or an initial segment of $\omega$ with length in $A$, then $W^{A}_{\sg(i)} = W^{A}_{i}$.
\end{proof}

With the help of the first single-valuedness theorem we are now able to derive the reduction principle.

\begin{theorem}[Reduction Principle] \label{thm-redp}
Let $B, C \subseteq \omega$ be c.e. Then there are disjoint c.e.\ subsets $B', C'$ of $B$ and $C$, respectively, so that $B' \cup C' = B \cup C$.
\end{theorem}
\begin{proof}
Let $X = \pair{B, \{ 0 \}} \cup \pair{C, \{ 1 \}}$. Then $X$ is c.e., say $X = W^{A}_{i}$. Let $X' \Def W^{A}_{\sv(i)}$ and  set $B' \Def \set{a}{\pair{a, 0} \in X'}$ as well as $C' \Def \set{b}{\pair{b, 1} \in X'}$. Then $B'$ and $C$ are as wanted.
\end{proof}

In what follows. let $\le_{m}$, $\le_{1}$, $\equiv_{m}$, $\equiv_{1}$ and $\equiv$, respectivly, denote $m$-reducibility, 1-reducibility, $m$-equivalence, 1-equivalence and computable isomorphism of sets and numberings, as usual. Since only total computable functions are involved in the corresponding definitions, these carry over unchanged to the theory under development here. The same holds for their well-known properties as well as the definition of $m$- and 1-completeness. We don't want to go into detail about this. Let 
\begin{align*}
&K^{A}_{0} \Def \set{\pair{i, x}}{x \in W^{A}_{i}}, \\
&K^{A}_{1} \Def \set{\pair{i, x}}{x \in \dom(\theta_{i})}, \\
&K^{A}_{2} \Def \set{i}{i \in \dom(\theta_{i})}.
\end{align*}

\begin{theorem}\label{thm-K-cpl}
$K^{A}$, $K^{A}_{0}$, $K^{A}_{1}$ and  $K^{A}_{2}$ are 1-complete.
\end{theorem}
\begin{proof}
The proof of completeness proceeds as usual. By Corollary~\ref{cor-unice}, $K^{A}_{0}$ is c.e. Let $B$ be a c.e.\ set, say $B = W^{A}_{j}$. Then $\lambda x.\ \pair{j, x}$ 1-reduces $B$ to $K^{A}_{0}$,

Since $K^{A}$ is c.e., it suffices to show that $K^{A}_{0} \le_{1} K^{A}$. Let $h \in \RRR^{(1)}$ enumerate the set $K^{A}_{0}$ and 
\begin{gather*}
Q(i, t, z) \Leftrightarrow (\exists a \le t)\, h(a) = 1, \\
f(i, t, z) \Def \pi^{(2)}_{1}(z) + 1.
\end{gather*}
Then it follows for the function $k \in \RRR^{(2)}$ constructed according to Lemma~\ref{lem-meth1} that $\lambda t.\ k(i,t)$ enumerates the extended graph of the identity on $\omega$, in case that $i \in K^{A}_{0}$. Otherwise, it enumerates the extended graph of the nowhere defined function. By applying Condition~(QGN~II) we now obtain a function $v \in \RRR^{(1)}$ so that $\range(\lambda t.\ k(i,t)) = \egraph(\theta_{v(i)})$. As we have already seen, we can assume $v$ to be one-to-one. Observe that
\[
\pair{j, x} \in K^{A}_{0} \Leftrightarrow (\forall y)\, \theta_{v(\pair{j, x})}(y) = y \Leftrightarrow v(\pair{j, x}) \in W^{A}_{v(j, x)} \Leftrightarrow v(\pair{j, x}) \in K^{A}.
\]
It follows that $K^{A}_{0} \le_{1} K^{A}$.

Because of Condition~(QGN~I) there is an enumeration $\hat{h} \in \RRR^{(1)}$ of the extended graph of the universal function of numbering $\theta$. Then 
\[
\pair{i, x} \in K^{A}_{i} \Leftrightarrow (\exists a)\, \pi^{(2)}_{1}(\pi^{(2)}_{1}(\hat{h}(a))) = i \wedge \pi^{(2)}_{2}(\pi^{(2)}_{1}(\hat{h}(a))) = x \wedge \pi^{(2)}_{2}(\hat{a}) > 0, 
\]
from which it follows that $K^{A}_{1}$ is c.e. We show that $K^{A}_{0} \le_{1} K^{A}_{1}$. Let to this end the relation $Q$ be as defined above and set $f(i, t, z) \Def 1$. Moreover, construct the function $k \in \RRR^{(2)}$ according Lemma~\ref{lem-meth1} on this basis. Then $\lambda t.\ k(i, t)$ enumerates the extended graph of the function $\lambda x.\ 0$, in case that $i \in K^{A}_{0}$. Otherwise, it enumerates the extended graph of the nowhere defined function. Let $v \in \RRR^{(1)}$ as in Condition~(QGN~II). Then $\range(\lambda t.\ k(i, t)) = \egraph(\theta_{v(i)})$. Moreover, we have that
\[
i \in K^{A}_{0} \Leftrightarrow (\forall y)\, \theta_{v(i)} = 0 \Leftrightarrow v(i) \in \dom(\theta_{v(i)}) \Leftrightarrow \pair{v(i), v(i)} \in K^{A}_{1}.
\]
Hence, $K^{A}_{0} \le_{1} K^{A}_{1}$.

Since $K^{A}_{1}$ is c.e., the same holds for $K^{A}_{2}$. In addition, we have
\[
i \in K^{A} \Leftrightarrow \pair{i, i} \in K^{A}_{0} \Leftrightarrow \pair{v(i), v(i)} \in K^{A}_{1} \Leftrightarrow v(\pair{i, i}) \in K^{A}_{2}.
\]
Thus, $K^{A} \le_{1} K^{A}_{2}$, which shows that also $K^{A}_{2}$ is 1-complete.
\end{proof}

In the classical theory of c.e.\ sets,  the sets mentioned in the theorem above $K^{A}$, $K^{A}_{0}$, $K^{A}_{1}$ and $K^{A}_{2}$ are also shown to be 1-complete. Since it is the aim of the present paper to show that this theory can as well be developed on the basis of the functions in $\SSS^{(1)}_{A}$ and quasi-G\"odel numberings, the 1-completeness of $K^{A}$ and $K^{A}_{0}$ is what was expected. The 1-completeness of $K^{A}_{1}$ and $K^{A}_{2}$, however, is less obvious because of the special form of the domains of the functions in $\SSS^{(1)}_{A}$.

\begin{corollary} \label{cor-K1eq}
$K^{A} \equiv_{1} K^{A}_{0}  \equiv_{1} K^{A}_{1} \equiv_{1} K^{A}_{2}$.
\end{corollary}

The notion of productive set is now introduced as usual. $C \subseteq\omega$ is \emph{$A$-productive}, if there is some $p \in \SSS^{(1)}_{A}$ such that for all $i \in \omega$, if $W^{A}_{i} \subseteq C$ then $p(i)\conv \in C \setminus W^{A}_{i}$. Since $\emptyset \subseteq C$ and $\set{i}{W^{A}_{i} = \emptyset}$ is infinite by the padding lemma, we have that $\dom(p)$ is infinite as well. Therefore $p$ cannot be an initial segment function. That is, $p \in \RRR^{(1)}$.

\begin{proposition}
$C \subseteq \omega$ is $A$-productive if, and only if, there is a total function $p \in \RRR^{(1)}$  so that $p(i) \in C \setminus W^{A}_{i}$, for $i \in \omega$ such that $W^{A}_{i} \subseteq C$.
\end{proposition}

As usual it can moreover be shown that $p$ can even be chosen as one-to-one and onto, and that $A$-productiveness is inherited upwards under $m$-reducibility. 

\begin{theorem}
Every $A$-productive set has an infinite c.e.\ subset.
\end{theorem}
\begin{proof}
Define $k \in \RRR^{(2)}$ by k$(i, 2t) \Def \pair{t, 0}$ and $k(i, 2t+1) \Def \pair{t, i+1}$. Then $k$ enumerates the extended graph of $\lambda x.\ i$. Now, let $v \in 
\RRR^{(1)}$ be as in Condition~(QGN~II) so that $\theta_{v(i)} = \lambda x.\ i$ and hence $W^{A}_{v(i)} = \{ i \}$. Assume that $C \subseteq \omega$ is $A$-productive with productive function $p \in \RRR^{(1)}$ and let $\un \in \RRR^{(2)}$ be as in Theorem~\ref{thm-indoc} with $W^{A}_{\un(i,j)} W^{A}_{i} \cup W^{A}_{j}$. Moreover, let $i_{0}$ be a $W^{A}$-index of the empty set.  Set
\begin{align*}
&g(0) = i_{0}, \\
&g(a+1) = \un(v(p(g(a))), g(a)).
\end{align*}
Then $g \in \RRR^{(1)}$. In addition, 
\[
W^{A}_{g(a+1)} = \{ p(g(a)) \} \cup W^{A}_{g(a)} = \{ p(g(a)), \ldots, p(g(0)) \}
\]
and $p(g(a)) \in C \setminus W_{g(a)}$. Thus, by defining $g' = p \circ g$ we obtain a one-to-one total computable function and consequently, $\range(g')$ is an infinite c.e.\ subset of $C$.
\end{proof}

 \begin{theorem}\label{thm-prodK}
 \begin{enumerate}
 
 \item \label{thm-prodK-1}
$\omega \setminus K^{A}$ is $A$-productive.

\item \label{thm-prodK-2}
$\text{$C$ $A$-productive} \Leftrightarrow (\omega \setminus K) \le_{1} C \Leftrightarrow (\omega \setminus K^{A}) \le_{m} C$.

\end{enumerate}
\end{theorem}
\begin{proof}
(\ref{thm-prodK-1}) follows by choosing $\lambda x. i$ as productive function. 

(\ref{thm-prodK-2}) Since $A$-productiveness is inherited upwards under $m$-reducibility, we have that, if $(\omega \setminus K^{A}) \le_{m} C$ then $C$ is $A$-productive. We will now show that for every $A$-productive set $C$, $(\omega \setminus K) \le_{1} C$.

Let $p \in \RRR^{(1)}$ be a one-to-one productive function of $C$ and $g \in \RRR^{(1)}$ enumerate set $K^{A}$. Moreover, let $f(\pair{i,j}, t, z) \Def p(j) + 1$ and
\[
Q(\pair{i, j}, t, z) \Leftrightarrow (\exists a \le t)\, g(a) = i.
\]
For the function $k \in \RRR^{(2)}$ constructed according to Lemma~\ref{lem-meth1}, it then holds that in case $i \in K^{A}$, $\lambda t.\ k(\pair{i, j}, t)$ enumerates the extended graph of the function $\lambda x.\ p(j)$ and otherwise the extended graph of the nowhere defined function. If $v$ is the function existing for this $k$ according to (QGN~II), then we have for $i \in K$ that $\theta_{v(\pair{i, j})}(x) = p(j)$. Otherwise, $\theta_{v(\pair{i, j})}(x)$ is undefined. By the recursion theorem there is now a function $g\in \RRR(1)$ with $\theta_{g(i)} = \theta_{v(\pair{i, g(i)})}$, and as we have seen there is even a one-to-one function $g$ with this property. It follows that
\[
W^{A}_{g(i)} = \begin{cases}
				\{ p(g(i)) \} & \text{if $i \in K^{A}$}, \\
				\emptyset & \text{otherwise}.
			\end{cases}
\]
We therefore obtain
\begin{align*}
i \in K^{A} 
&\Rightarrow W^{A}_{g(i)} = \{ p(g(i)) \} \\
&\Rightarrow W^{A}_{g(i)} \not\subseteq C \quad \text{(as $p$ is a productive function of $C$)} \\
&\Rightarrow p(g(i)) \notin C \\
\intertext{and}
i \notin K^{A} &\Rightarrow W^{A}_{g(i)} = \emptyset \Rightarrow W^{A}_{g(i)} \subseteq C \Rightarrow p(g(i)) \in C.
\end{align*}
Since $p \circ g$ is one-to-one, total and computable, this proves that $(\omega \setminus K) \le_{1} C$.
\end{proof}

If $C \subseteq \omega$ is a c.e.\ set the complement of which is $A$-productive, $C$ is called \emph{$A$-creative}. From the above results it follows that $K^{A}$ is $A$-creative. Note that Myhill's theorem on the coincidence of the notions of 1-equivalence and computable isomorphism also holds in the approach to the theory of c.e.\ sets presented here: the proof in \cite{ro67} uses only arguments which are admissible in our approach as well. Therefore, we obtain the following characterisation of the $A$-creative sets.

\begin{theorem}\label{thm-crea}
Let $C \subseteq \omega$. Then the following  four statements are equivalent:
\begin{enumerate}

\item $C$ is 1-complete.

\item $C$ is $m$-complete.

\item $C$ is $A$-creative.

\item $C \equiv K^{A}$.

\end{enumerate}
\end{theorem}

Rogers~\cite{ro58} shows that a set is creative exactly if it is the self-reproducibility problem of a G\"odel numbering. With Theorem~\ref{thm-crea} we obtain a corresponding result for quasi-G\"odel numberings.

\begin{theorem}\label{thm-godrea}
A set $C \subseteq \omega$ is $A$-creative if, and only if, there is a quasi-G\"odel numbering $\chi$ of $\SSS^{(1)}_{A}$ such that $C = \set{i}{i \in \range(\chi_{i})}$.
\end{theorem}
\begin{proof}
We have already seen that for each quasi-G\"odel numbering $\chi$ the set $\set{i}{i \in \range(\chi_{i})}$ is $A$-creative. Assume conversely that $C$ is $A$-creative. Then $C \equiv K^{A}$ by the preceding theorem. Therefore, there is a one-to-one and onto function $f \in \RRR^{(1)}$ so that $i \in C$ exactly if $f(i) \in \range(\theta_{f(i)})$. Since $f$ is total, the modified substitution $\msubst^{(1,1)}_{A}(f; \theta_{i})$ coincides with the usual composition $f \circ \theta_{i}$. By Theorem~\ref{thm-sub} there is thus a function $g \in \RRR^{(1)}$ with $\theta_{g(i)} = f \circ \theta_{i}$. Moreover, there is a function $r \in \RRR^{(1)}$ so that $\theta_{r(i)} = f^{-1} \circ \theta_{i}$. Let $p = r \circ f$ and $q = f^{-1} \circ g$. Then the numbering $\chi$ we looking for is defined by $\chi \Def \theta \circ p$. Obviously, we then have that $\theta = \chi \circ q$. As is readily verified, $\chi$ is a quasi-G\"odel numbering. In addition, we have that
\begin{align*}
i \in \range(\chi_{i}) 
\Leftrightarrow i \in \range(\theta_{p(i)}) 
&\Leftrightarrow i \in \range(\theta_{r(f(i))})  \\
&\Leftrightarrow i \in  \range(f^{-1} \circ \theta_{f(i)})
\Leftrightarrow f(i) \in \range(\theta_{f(i)})
\Leftrightarrow i \in C. \tag*{\qedhere}
\end{align*}
\end{proof}

We hope that with the results in this section we have presented a convincing selection of theorems showing that the theory of c.e.\ sets can also be constructed on the basis of the theory of functions from $\SSS^{(1)}_{A}$. In particular, all results apply here that are usually derived without referring to the domain characterisation of the c.e.\ sets such as Myhill's theorem mentioned above. In the other cases, the above proofs show how one can replace constructions in which the domain characterisation is commonly used by other constructions that are admissible in the theory developed here. We now want to show that the numbering $W^A$ defined here via a quasi-G\"odel numbering, which we have shown has many properties of the commonly used numbering $W$ defined at the beginning of Section~\ref{sec-isf}, is not essentially different from $W$.

\begin{theorem}\label{thm-W}
$W^{A} \equiv W$.
\end{theorem}
\begin{proof}
By the number-theoretic analogue of Myhill's theorem (cf.~\cite{er73}) it suffices to prove that $W^{A} \le_{1} W$ and $W \le_{1} W^{A}$. We first show that $W^{A} \le_{1} W$. Since $\theta$ satisfies Condition~(QGN~I), $\lambda i, x.\,  \theta_{i}(x) \in \PPP^{(2)}$. Therefore, there is a function $g \in \RRR^{(1)}$ with $\theta_{i} = \varphi_{g(i)}$. Moreover, there is a function $f \in \RRR^{(1)}$ so that $\dom(\varphi_{f(i)}) = \range(\varphi_{i})$. Since $\varphi$ has a padding function, there also  such functions that are one-to-one.. Thus, $W^{A}_{i} = W_{f(g(i))}$, that is, $W^{A} \le_{1} W$.

Next, we show that also $W \le_{1} W^{A}$. Let to this end $p \in \RRR^{(1)}$ such that $\range(\varphi_{p(i)}) = \dom(\varphi_{i})$ and $\varphi_{p(i)}$ is total, if $\dom(\varphi_{i})$ is not empty, and $\varphi_{i}$ is nowhere defined, otherwise. Moreover, let $k \in \RRR^{(2)}$ enumerate the extended graph of $\varphi_{p(i)
}$. Since $\theta$ satisfies Condition~(QGN~II) and has a one-to-one padding function, there is a one-to-one function $v \in \RRR^{(1)}$ with $\theta_{v(i)} = \varphi_{p(i)}$. Therefore, $\range(\theta_{v(i)}) = \range(\varphi_{p(i)}) = \dom(\varphi_{i})$, that is, $W \le_{1} W^{A}$.
\end{proof}

As follows from this result, the $A$-productive and $A$-creative sets, respectively, coincide with the productive and the creative ones. Note that the above theorem does not obviate the results on $W^{A}$, such as Theorem \ref{thm-indoc}. Theorem~\ref{thm-W} is a metatheorem derived within classical computability theory, whereas the results in this section are results of the theory presented here, derived in this same theory.

In Section~\ref{sec-comp} we have seen that the sets $\set{i}{\theta_{i} \in \anf^{(1)}_{A}}$ and $\set{i}{\theta_{i} \in \RRR^{(1)}}$ are not computable. At the end of this section we want to determine the exact position of these sets in the arithmetic hierarchy.

\begin{theorem}\label{thm-arith}
\begin{enumerate}

\item \label{thm-arith-1}
$\set{i}{\theta_{i} \in \anf^{(1)}_{A}}$ is $\Sigma_{2}$-complete.

\item \label{thm-arith-2}
$\set{i}{\theta_{i} \in \RRR^{(1)}}$ is $\Pi_{2}$-complete.

\end{enumerate}
\end{theorem}
\begin{proof}
It suffices to prove (\ref{thm-arith-1}). Let to this end $h \in \RRR^{(1)}$ enumerate the extended graph of the universal function of numbering $\theta$. Then we have that
\begin{multline*}
\theta_{i} \in \anf^{(1)}_{A} \Leftrightarrow \mbox{} \\
\hspace{2em} (\exists t) (\exists a)\, [a \in A_{t} \wedge (\forall x <a) (\exists c)\, [\pi^{(2)}_{1}(h(c)) = \pair{i, x} \wedge \pi^{(2)}
_{2}(h(c)) > 0] \wedge \mbox{} \\
[(\forall y \ge a) (\forall b)\, [\pi^{(2)}_{1}(h(b)) = \pair{i, y} \Rightarrow \pi^{(2)}_{2}(h(b)) = 0]].
\end{multline*}
It follows that $\set{i}{\theta_{i} \in \anf^{(1)}_{A}} \in \Sigma_{2}$. It remains to show that for $C \in \Sigma_{2}$, $C \le_{1} \set{i}{\theta_{i} \in \anf^{(1)}_{A}}$. Let $C \in \Sigma_{2}$. Then there is a ternary computable predicate $Z$ such that
\[
i \in C \Leftrightarrow (\exists x) (\forall y)\, Z(i, x, y).
\]
Set $f(i, t, z) \Def 1$ and 
\[
Q(i, t, z) \Leftrightarrow (\exists a \in A_{t})\, [\narg(z) < a \wedge (\forall x < a) (\exists y \le t)\, \neg Z(i, x, y)], 
\]
and let $k \in \RRR^{(2)}$ be the function constructed from this according to Lemma~\ref{lem-meth1}. As can be seen from the construction, for the one-to-one function $v \in \RRR^{(1)}$ that exists according to (QGN~II) and Corollary~\ref{cor-indinj} one then has that 
\[
\theta_{v(i)}(x) = \begin{cases}
				0 & \text{if there is some $a \in A$ so that $x < a$ and for all $x' < x$}\\
				   & \text{there is some $y$ so that $\neg Z(i, x', y)$}, \\
				\text{undefined} & \text{otherwise.}
			\end{cases}
\]
If $i \in C$ then $\set{x}{(\forall y)\, Z(i, x, y)}$ is not empty. Let $\hat{x}$ be the smallest element of this set and $\hat{a} = \max \set{a \le \hat{x}}{a \in A \vee a = 0}$. Then $\theta_{v(i)}(x)$ is undefined, for all $x \ge \hat{a}$, and $\theta_{v(i)}(x) = 0$, for all $x < \hat{a}$. Thus, $\theta_{v(i)} \in \anf^{(1)}_{A}$.

If, on the other hand, $i \notin C$, then there is some $y \in \omega$ with $\neg Z(i, x, y)$, for all $x \in \omega$. It follows in this case that for all $x \in \omega$, $\theta_{v(i)} = 0$. That is, $\theta_{v(i)} \in \RRR^{(1)}$. So, we have
\[
i \in C \Leftrightarrow \theta_{v(i)} \in \anf^{(1)}_{A}.
\]
That is, $C \le_{1} \set{i}{\theta_{i }\in \anf^{(1)}_{A}}$.
\end{proof}

\section{Enumerability of subsets of $\SSS_{A}^{(n)}$}\label{sec-compII}

In this and the next section we consider the computability of type-two objects over the function classes $\SSS_{A}^{(n)}$. We start with the enumerability of sets of functions in $\SSS_{A}^{(n)}$. The usual definition in this case is intensional and uses quasi-G\"odel numbers of these functions. As main result a theorem of Rice-Shapiro type is derived, which gives an extensional index-free characterisation of such sets. 

\begin{defin}
Let $\theta^{(n)}$ be a quasi-G\"odel numbering of $\SSS_{A}^{(n)}$. Then a set $X \subseteq \SSS_{A}^{(n)}$ is called  \emph{completely c.e.} if the index set $I_{\theta^{(n)}}(X)$ of $X$ with respect to $\theta^{(n)}$ is c.e.
\end{defin}

For stating the main result of this section, we need a canonical indexing of $\anf_{A}^{(n)}$.
Let to this end $\fun{\kappa}{\omega}{\omega^{n}}$ again be an effective and one-to-one map that enumerates $\omega^{n}$ initial segment by initial segment. Then $\kappa(c) < a^{(n)}$, exactly if $c < a^{n}$. Moreover, let $f_{A} \in \RRR^{(1)}$ enumerate $A$, and define
\[
\widehat{\alpha}_{\pair{a,b}}(\vec x) \Def \begin{cases}
								(a)_{\kappa^{-1}(\vec x)} & \text{if $\kappa^{-1}(\vec x) < \min \{ \lth(a)+1, f_{A}(b)^{n} \}$}, \\
								0 & \text{if $\lth(a) < \kappa^{-1}(\vec x) < f_{A}(b)^{n}$}, \\
								\text{undefinded} & \text{otherwise}, \\
							\end{cases}
\]
and
\begin{align*}
&\alpha^{(n)}_{0} \Def \lambda \vec x.\ \text{undefined}, \\
&\alpha^{(n)}_{i+1} \Def \widehat{\alpha}_{\mu j.\ \widehat{\alpha}^{(n)}_{j} \notin \{ \alpha^{(n)}_{0}, \ldots, \alpha^{(n)}_{i} \}}.
\end{align*}

\begin{lemma}\label{lem-alph}
Let $\theta^{(n)}$ be a quasi-G\"odel numbering of $\SSS^{(n)}_{A}$. Then the following five statements hold:
\begin{enumerate}

\item \label{lem-alph-1}
$\alpha^{(n)}$ is a one-to-one numbering of $\anf^{(n)}_{A}$.

\item \label{lem-alph-2}
$\egraph(\lambda i, \vec x.\ \alpha^{(n)}_{i}(\vec x))$ is computable.

\item \label{lem-alph-3}
Let $\nlg(i)$ be the edge length of $\dom(\alpha^{(n)}_{i})$. Then $\nlg \in \RRR^{(1)}$.

\item \label{lem-alph-4}
$\alpha^{(n)} \le_{1} \theta^{(n)}$.

\item \label{lem-alph-5}
$\set{\pair{i, j}}{\graph(\alpha^{(n)}_{i}) \subseteq \graph(\theta^{(n)}_{j})}$ is c.e.

\end{enumerate}
\end{lemma}
\begin{proof}
(\ref{lem-alph-1}) follows by the construction. We next show (\ref{lem-alph-2}) and (\ref{lem-alph-3}). As is easily seen, the sequence number
\[
\val{\widehat{\alpha}^{(n)}_{i}} \Def \val{\pair{\kappa(0), \widehat{\alpha}^{(n)}_{i}(\kappa(0))}, \ldots, \pair{\kappa(m), \widehat{\alpha}^{(n)}_{i}(\kappa(m))}}
\]
where $m \Def (f_{A}(\pi^{(2)}_{2}(i)))^{n}-1$, is computable from $i$.  Thus, if $g \in \RRR^{(1)}$ with
\[
g(i) = \mu j.\ \val{\widehat{\alpha}^{(n)}_{j}} \notin \set{\val{\text{`empty sequence'}}, \ldots, \val{\widehat{\alpha}^{(n)}_{g(\nu)}}}{1 \le \nu \le i},
\]
for $i > 0$, then $\alpha^{(n)}_{i} = \widehat{\alpha}^{(n)}_{g(i)}$, for $i > 0$. As is also readily seen, $\egraph(\lambda i, \vec x.\ \widehat{\alpha}^{(n)}_{i}(\vec x))$ is computable. Thus the same holds for $\egraph(\lambda i, \vec x.\ \alpha^{(n)}_{i}(\vec x))$. Since $\nlg(0) = 0$ and for $i > 0$,  $\nlg(i) = f_{A}(\pi^{(2)}_{2}(g(i)))$ we further obtain that $\nlg \in \RRR^{(1)}$.

(\ref{lem-alph-4}) Let $h \in \RRR^{(1)}$ enumerate $\egraph(\lambda i, \vec x.\ \alpha^{(n)}_{i}(\vec x))$. Moreover, define
\begin{gather*}
\widehat{Q}(i, c, z) \Leftrightarrow \pi^{(2)}_{2}(h(c)) > 0 \wedge \pi^{(2)}_{1}(h(c)) = \pair{i} \ast \arg(z), \\
Q(i, t, z) \Leftrightarrow (\exists c \le t)\, \widehat{Q}(i, c, z),\\
f(i, t, z) \Def \pi^{(2)}_{2}(h(\mu c \le t.\, \widehat{Q}(i, c, z))),
\end{gather*}
and construct $k \in \RRR^{(2)}$ as in Lemma~\ref{lem-meth1}. Then $\lambda t.\ k(i, t)$ enumerates the extended graph of $\alpha^{(n)}_{i}$. According to (QGN~II)  and Corollary~\ref{cor-indinj} there is then a one-to-one function $v \in \RRR^{(1)}$ such that $\alpha^{(n)}_{i} = \theta^{(n)}_{v(i)}$,

(\ref{lem-alph-5}) Let, in addition, $h' \in \RRR^{(1)}$ be an enumeration of $\egraph(\lambda i, \vec x.\ \theta^{(n)}_{i}(\vec x))$. Then we have that
\begin{align*}
\graph(\alpha^{(n)}_{i}) \subseteq \graph(\theta^{(n)}_{j}) \\
&\hspace{-7em} \Leftrightarrow\mbox{}  (\forall \vec x < (\nlg(i))^{(n)})\, \alpha^{(n)}_{i}(\vec x) = \theta^{(n)}_{j}(\vec x) \\
&\hspace{-7em}  \Leftrightarrow\mbox{}  (\forall \vec x < (\nlg(i))^{(n)}) (\exists t) (\exists t')\, \pi^{(2)}_{1}(h(t)) = \pair{i, \vec x} \wedge \pi^{(2)}_{1}(h'(t')) = \pair{j, \vec x} \wedge \mbox{} \\
 \mbox{}&\hspace{-4em}\pi^{(2)}_{2}(h(t)) > 0 \wedge \pi^{(2)}_{2}(h(t)) = \pi^{(2)}_{2}(h'(t')) \\
&\hspace{-7em}  \Leftrightarrow\mbox{} (\exists t) (\exists t') (\forall c < \nlg(i))\, \pi^{(2)}_{1}(h((t)_{c})) = \pair{i, \kappa(c)} \wedge \pi^{(2)}_{1}(h'((t')_{c})) = \pair{j, \kappa(c)} \wedge \mbox{} \\
 & \mbox{}\hspace{-4em} \pi^{(2)}_{2}(h((t)_{c})) > 0 \wedge \pi^{(2)}_{2}(h((t)_{c})) = \pi^{(2)}_{2}(h'((t')_{c})).
\end{align*}
The statement is now a consequence of the projection lemma.
\end{proof}
Note that in the proof of Property~(\ref{lem-alph-5}) only Condition~(QGN~I) was used.

\begin{theorem}[Rice, Shapiro]\label{thm-rs}
Let $\theta^{(n)}$ be a quasi-G\"odel numbering of $\SSS^{(n)}_{A}$. Then, 
$X \subseteq \SSS^{(n)}_{A}$ is completely c.e.\ if, and only if, there is a c.e.\ set $C \subseteq \omega$ so that
\begin{equation}\label{eq-scop}
X = \set{r \in \SSS^{(n)}_{A}}{(\exists i \in C)\, \graph(\alpha^{(n)}_{i}) \subseteq \graph(r)}.
\end{equation}
\end{theorem}
\begin{proof}
Let us assume first that $X$ has the special form. Then
\[
j \in I_{\theta^{(n)}}(X) \Leftrightarrow (\exists i \in C)\, \graph(\alpha^{(n)}_{i}) \subseteq \graph(\theta^{(n)}_{j}).
\]
With Lemma~\ref{lem-alph}(\ref{lem-alph-5}) and the projection lemma it follows that $I_{\theta^{(n)}}(X)$ is c.e.

Conversely, suppose that $I_{\theta^{(n)}}(X)$ is c.e. The proof now proceeds in three steps.

\begin{claim}\label{cl-1}
$(\forall r \in X) (\exists s \in X \cap\anf^{(n)}_{A})\, \graph(s) \subseteq \graph(r)$.
\end{claim}
\noindent
Without restriction assume that $r$ is total. Moreover, suppose that there is no $s \in X \cap \anf^{(n)}_{A}$ with $\graph(s) \subseteq \graph(r)$. To derive a contradiction, it suffices to show that $\omega \setminus K^{A} \le_{m} I_{\theta^{(n)}}(X)$ in this case. By our assumption it would follow that $\omega \setminus K^{A}$ is c.e., which is not the case as we have already seen. 

Let $h \in \RRR^{(1)}$ enumerate $K^{A}$. We will show that there is a function $q \in \RRR^{(1)}$ with
\[
\theta^{(n)}_{q(i)}(\vec x) = \begin{cases}
						r(\vec x) & \text{if there are $a \in A$ and $c \in \omega$ so that $\vec x < a^{(n)}$, $a \le c$ and $h(c') \neq i$,} \\
						& \text{for all $c' \le c$,} \\
						\text{undefined} & \text{otherwise}.
					\end{cases}
\]
Then we have that $i \in \omega \setminus K^{A}$, exactly if $q(i) \in I_{\theta^{(n)}}(X)$. To see this, note that if $i \in \omega \setminus K^{A}$ then $h(c) \neq i$, for all $c \in \omega$. Hence, $\theta^{(n)}_{q(i)} = r$ in this case and thus $q(i) \in I_{\theta^{(n)}}(X)$, as $r \in X$. If $ i \in K^{A}$, there is some $c \in \omega$ with $h(c) = 1$. Let $\hat{c}$ be the smallest such $c$ and $\hat{a} \Def \max\set{a \le \hat{c}}{a \in A \vee a=0}$. Then we have for all $\vec x \le \hat{a}^{(n)}$ that $\theta^{(n)}_{q(i)}(\vec x) = r(\vec x)$. For all other $\vec x \in \omega^{n}$, $\theta^{(n)}(\vec x)$ is undefined. Consequently, $\theta^{(n)}_{q(i)} \in \anf^{(n)}_{A}$ and $\graph(\theta^{(n)}_{q(i)}) \subseteq graph(r)$. By our assumption this means that $\theta^{(n)}_{q(i)} \notin X$, that is $q(i) \notin I_{\theta^{(n)}}(X)$.

For the construction of $q$ we again apply Lemma~\ref{lem-meth1} and Condition~(QGN~II), as we did already many times. We only state the predicate $Q$ and the function $f$ needed in the construction:
\begin{gather*}
Q(i, \pair{b, c}, z) \Leftrightarrow f_{A}(b) \le c \wedge (\forall c' \le c)\, h(c') \neq i \wedge \bigwedge_{\nu=1}^{n} \pi^{(n)}_{\nu}(\arg(z)) < f_{A}(b), \\
f(i, \pair{b, c}, z) \Def r(\arg(z))+1.
\end{gather*}

\begin{claim}\label{cl-2}
$(\forall s, r \in \SSS^{(n)}_{A})\, [[s \in X \wedge \graph(s) \subseteq \graph(r)] \Rightarrow r \in X]$.
\end{claim} 
\noindent
Again we assume to the contrary that there are $s, r \in \SSS^{(n)}_{A}$ so that $s \in X$, $\graph(s) \subseteq \graph(r)$, but $r \notin X$. Let $j$ be a $\theta^{(n)}$-index of $s$. Then we construct a function $p \in \RRR^{(1)}$ so that
\[
\theta^{(n)}_{p(i)}(\vec x) = \begin{cases}
						s(\vec x) & \text{if in a simultaneous search in the extended graph of $\theta^{(n)}$ and $K^{A}$,} \\
							      & \text{respectively, $\pair{\pair{j, \vec x}, s(\vec x) +1}$ will not be found later than $i$}, \\
						r(\vec x) & \text{if $i$ will be found earlier}, \\
						\text{undefined} & \text{otherwise}.
					\end{cases}
\]
Note that since $\graph(s) \subseteq \graph(r)$, in case that $i \in K^{A}$ we always have that $\theta^{(n)}_{p(i)} = r$, independently when $\pair{\pair{j, \vec x}, s(\vec x) +1}$ will be found. By our assumption $r \notin X$. Thus it follows for $i \in K^{A}$ that $p(i) \notin I_{\theta^{(n)}}(X)$. On the other hand, if $i \in \omega \setminus K^{A}$, then $\theta^{(n)}_{p(i)} = s$. Because $s \in X$, we have in this case that  $p(i) \in I_{\theta^{(n)}}(X)$. This shows that $\omega \setminus K^{A}  \le_{m}  I_{\theta^{(n)}}(X)$, which is impossible as we have seen.

It remains to construct the function $p$. Let to this end $j'$ be a $\theta^{(n)}$-index of $r$ and $h' \in \RRR^{(1)}$ be an enumeration of the extended graph of the universal function of numbering $\theta^{(n)}$. The existence of $h'$ follows from Condition~(QGN~I). Define 
\begin{gather*}
\widehat{Q}_{1}(a, b, z) \Leftrightarrow \pi^{(2)}_{2}(h'(b)) > 0 \wedge \pi^{(2)}_{1}(h'(b)) = \pair{a} \ast \arg(z), \\
\widehat{Q}_{2}(i, b, z) \Leftrightarrow [\widehat{Q}_{1}(j, b ,z) \wedge (\forall c < b)\, h(c) \neq i] \vee \mbox{} \\
\hspace{4cm}[\widehat{Q}_{1}(j', b, z) \wedge (\exists b' \le b)\, [h(b') = i \wedge \neg (\exists c \le b')\, \widehat{Q}_{1}(j, c, z)]], \\
Q(i, t, z) \Leftrightarrow (\exists b \le t)\, \widehat{Q}_{2}(i, b, z), \\
f(i, t, z) \Def \pi^{(2)}_{2}(h'(\mu b \le t.\ \widehat{Q}_{2}(i, b, z))),
\end{gather*}
and construct function $k \in \RRR^{(2)}$ as in Lemma~\ref{lem-meth1}. Then $\lambda t.\ k(i, t) $ enumerates the extended graph of function $d$ with
\[
d(\vec x) = \begin{cases}
			s(\vec x) & \text{$(\exists b)\, \pi^{(2)}_{2}(h'(b)) > 0 \wedge \pi^{(2)}_{1}(h'(b)) = \pair{j, \vec x} \wedge (\forall c < b)\, h(c) \neq i]$}, \\
			r(\vec x) & \text{$(\exists b)\, h(b) = i \wedge \neg (\exists c \le b)\, [\pi^{(2)}_{2}(h'(c)) > 0 \wedge \pi^{(2)}_{1}(h'(c)) = \pair{j, \vec x}]]$}, \\
			\text{undefined} & \text{otherwise}.
		\end{cases}
\]
As follows from the definition, $d \in \SSS^{(n)}_{A}$. By Condition~(QGN~II) there is then a $p \in \RRR^{(1)}$ with $\theta^{(n)}_{p(i)} = d$. Consequently, $p$ has the properties mentioned above.

\begin{claim}
There is a c.e.\ set $C \subseteq \omega$ so that
$X = \set{r \in \SSS^{(n)}_{A}}{(\exists i \in C)\, \graph(\alpha^{(n)}_{i}) \subseteq \graph(r)}$.
\end{claim}
\noindent
Let $g \in \RRR^{(1)}$ with $\alpha^{(n)} = \theta^{(n)} \circ g$ and $C \Def g^{-1}[I_{\theta^{(n)}}(X)]$. Then $C$ is c.e. By Claim~\ref{cl-1} we have for $r \in \SSS^{(n)}_{A}$ that
\begin{align*}
r \in X 
&\Rightarrow (\exists s \in X \cap \anf^{(n)}_{A})\, \graph(s) \subseteq \graph(r) \\
&\Rightarrow (\exists i)\, [\alpha^{(n)}_{i} \in X \wedge \graph(\alpha^{(n)}_{i}) \subseteq \graph(r)] \\
&\Rightarrow (\exists i \in C)\, \graph(\alpha^{(n)}_{i}) \subseteq \graph(r).
\end{align*}

Conversely, we have with Claim~\ref{cl-2} that
\begin{align*}
 (\exists i \in C)\, \graph(\alpha^{(n)}_{i}) \subseteq \graph(r) 
 &\Rightarrow (\exists s \in X)\, \graph(s) \subseteq \graph(r) \\
 &\Rightarrow r \in X. \tag*{\qedhere}
 \end{align*}
\end{proof}	

As is well known~\cite{od92}, the completely c.e.\ subsets of a numbered set generate a topology, called the \emph{Ershov topology}. Sets as in Equation~(\ref{eq-scop}), on the other hand, generate a topology on $\SSS_{A}^{(n)}$, called the \emph{Scott topology}. From the above result it follows that both topologies are equivalent on $\SSS_{A}^{(n)}$.

\section{Effective operations on $\SSS_{A}^{(n)}$}\label{sec-effop}	
			
Various kinds of effective operations have been studied on the partial computable functions. Here, we will consider two kinds: computable operators which use the graph of the argument function as oracle and are hence defined for all functions in  $\widehat{\SSS}^{(n)}_{A}$, and Markov computable functions which are only defined for the computable functions in  $\widehat{\SSS}^{(n)}_{A}$. They use quasi-G\"odel numbers of the argument function in their computation, that is, algorithms computing the argument function. The main result of the section will be a theorem of Myhill-Shepherson type relating the two kinds of operations.

\begin{defin}[cf.\ \cite{ro67}]\label{def-compop}
An operator $\fun{\widehat{G}}{\widehat{\SSS}^{(n)}_{A}}{\widehat{\SSS}^{(m)}_{A}}$ is \emph{computable}, if there is a c.e.\ set $C \subseteq \omega$ so that for $r \in \widehat{\SSS}^{(n)}_{A}$,
\[
\graph(\widehat{G}(r)) = \set{\pair{\pair{\vec y}, z}}{(\exists a)\, [\pair{\pair{\pair{\vec y}, z}, a} \in C \wedge (\forall c \le \lth(a))\, (a)_{c} \in \graph(r)]}.
\]
We say that $C$ \emph{defines} the operator $\widehat{G}$.
\end{defin}

As follows from the definition, $\graph(\widehat{G}(r))$ is c.e., if $\graph(r)$ is c.e. Thus, for computable operators $\fun{\widehat{G}}{\widehat{\SSS}^{(n)}_{A}}{\widehat{\SSS}^{(m)}_{A}}$ we have  that $\widehat{G}[\SSS^{(n)}_{A}] \subseteq \SSS^{(m)}_{A}$. If one assigns  a code number $\pair{\hspace{.5em}}$ to the zero-ary tuple, the above definition and the subsequent results for $m=0$ also include the case of the recursive functionals. However, it should be noted here that $\widehat{\SSS}^{(0)}_{A}$ and $\SSS^{(0)}_{A}$ contain the zero-digit constant functions, each of which we identify with the respective constant, as well as the zero-ary nowhere defined function. 	
\begin{theorem}\label{thm-coprop}
Let $\fun{\widehat{G}}{\widehat{\SSS}^{(n)}_{A}}{\widehat{\SSS}^{(m)}_{A}}$ and $G$ be its restriction to $\SSS^{(n)}_{A}$. Then $\widehat{G}$ is computable if, and only if, the following three conditions are satisfied:
\begin{enumerate}

\item \label{thm-coprop-1}
For all $s, s' \in \anf^{(n)}_{A}$, 
\[
\graph(s) \subseteq \graph(s') \Rightarrow \graph(G(s)) \subseteq \graph(G(s')).
\]

\item \label{thm-coprop-2}
For all $r \in \widehat{\SSS}^{(n)}_{A}$,
\[
\graph(\widehat{G}(r)) = \bigcup\set{\graph(G(s))}{s \in \anf^{(n)}_{A} \wedge \graph(s) \subseteq \graph(r)}.
\]

\item \label{thm-coprop-3}
$\set{\pair{i, j}}{\graph(\alpha^{(m)}_{j}) \subseteq \graph(G(\alpha^{(n)}_{i}))}$ is c.e.

\end{enumerate}
\end{theorem}
\begin{proof}
Assume that $\widehat{G}$ is computable. Then it follows from the definition that for $s \in \anf^{(n)}_{A}$ and $r \in \widehat{G}^{(n)}_{A}$ with $\graph(s) \subseteq \graph(r)$, $\graph(G(s)) \subseteq \graph(\widehat{G}(r))$. Thus, (\ref{thm-coprop-1}) holds. Moreover, 
\[
\bigcup\set{\graph(G(s))}{s \in \anf^{(n)}_{A} \wedge \graph(s) \subseteq \graph(r)} \subseteq \graph(\widehat{G}(r)).
\]
For the converse inclusion note that $\graph(\widehat{G}(r))$ is the union of all $\graph(q)$, where $q \in \anf^{(m)}_{A}$ with $\graph(q) \subseteq \graph(\widehat{G}(r))$. 
Therefore, let $q$ be such an initial segment function. In the enumeration of $\graph(\widehat{G}(r))$ each element of $\graph(q)$ corresponds to a finite number of questions to $\graph(r)$. Since $\graph(q)$ is finite, there is thus some $s \in \anf^{(n)}_A$ with $\graph(s) \subseteq \graph(r)$ and $\graph(q) \subseteq \graph(G(s))$, which shows that also Condition~\ref{thm-coprop-2} holds.

For Condition~(\ref{thm-coprop-3}) observe that
\begin{multline*}
\graph(\alpha^{(m)}_{j}) \subseteq \graph(G(\alpha^{(n)}_{i})) \Leftrightarrow \mbox{} \\
\hspace{.3em}  (\forall \vec y < (\lg(j))^{(m)}) (\exists z) (\exists a)\, \pair{\pair{j, \vec y}, z+1} \in \egraph(\lambda b, \vec x.\ \alpha^{(m)}_{b}(\vec x)) \wedge \mbox{} \\
 \pair{\pair{\vec y, z}, a} \in C \wedge (\forall c < \lth(a))\, \pair{\pair{i} \ast \pi^{(2)}_{1}((a)_{c}), 1+\pi^{(2)}_{2}((a)_{c})} \in \egraph(\lambda b, \vec x.\ \alpha^{(n)}_{b}(\vec x)). \hspace{-.8em}
\end{multline*}
Recall that the unbounded existential quantifiers in the right-hand side can be brought in front of the expression by using an effective sequence encoding. Moreover, the extended graphs in the expression are c.e. With the projection lemma we hence obtain (\ref{thm-coprop-3}).

Now, conversely, suppose that  $\fun{\widehat{G}}{\widehat{\SSS}^{(n)}_{A}}{\widehat{\SSS}^{(m)}_{A}}$ satisfies Conditions~(\ref{thm-coprop-1})-(\ref{thm-coprop-3}). With (\ref{thm-coprop-2}) we have for $r \in \widehat{\SSS}^{(n)}_{A}$ that
\begin{align*}
&\graph(\widehat{G}(r)) \\
&\qquad= \bigcup\set{\graph(G(s))}{s \in \anf^{(n)}_{A} \wedge \graph(s) \subseteq \graph(r)} \\
&\qquad= \bigcup \{\, \graph(q) \mid q \in \anf^{(m)}_{A} \wedge (\exists s \in \anf^{(n)}_{A})\, \graph(s) \subseteq \graph(r) \wedge \mbox{} \\
& \hspace{9.5cm}  \graph(q) \subseteq \graph(G(s)) \,\} \\
&\qquad= \bigcup\set{\graph(\alpha^{(m)}_{j})}{(\exists i)\, \graph(\alpha^{(n)}_{i}) \subseteq \graph(r) \wedge \graph(\alpha^{(m)}_{j}) \subseteq \graph(G(\alpha^{(n)}_{i}))}.
\end{align*}
Let the graph encdoding $\val{\alpha^{(n)}_{b}}$ be as in the proof of Lemma~\ref{lem-alph}. Then it follows
\begin{align*}
\pair{\pair{\vec y}, \mbox{}&z}  \in \graph(\widehat{G}(r)) \\
\Leftrightarrow \mbox{} &(\exists j) (\exists i)\, \pair{\pair{\vec y}, z} \in \graph(\alpha^{(m)}_{j}) \wedge \graph(\alpha^{(m)}_{j}) \subseteq \graph(G(\alpha^{(n)}_{i})) \wedge   \mbox{} \\
& \hspace{9cm} \graph(\alpha^{(n)}_{i}) \subseteq \graph(r) \\
\Leftrightarrow \mbox{} &(\exists a)\, [(\exists j) (\exists i)\, \pair{\pair{j, \vec y}, z+1} \in \egraph(\lambda b, \vec x.\ \alpha^{(m)}_{b}(\vec x))  
\wedge  a = \val{\alpha^{(n)}_{i}} \wedge \mbox{} \\
& \hspace{3cm}\graph(\alpha^{(m)}_{j}) \subseteq \graph(G(\alpha^{(n)}_{i}))] \wedge (\forall c < \lth(a))\, (a)_{c} \in \graph(r).
\end{align*}
Set
\begin{multline*}
C \Def \{\, \pair{\pair{\pair{\vec y}, z}, a} \mid (\exists j) (\exists i)\, \pair{\pair{j, \vec y}, z+1} \in \egraph(\lambda b, \vec x.\ \alpha^{(m)}_{b}(\vec x)) \wedge \mbox{} \\\graph(\alpha^{(m)}_{j}) \subseteq \graph(G(\alpha^{(n)}_{i})) \wedge  a = \val{\alpha^{(n)}_{i}}\, \}.
\end{multline*}
Then $C$ defines $\widehat{G}$. Since $\val{\alpha^{(n)}_{i}}$ can be computed form $i$ and because of Lemma~\ref{lem-alph}(\ref{lem-alph-2}) and Condition~(\ref{thm-coprop-3}), we further have that $C$ is c.e. Thus, $\widehat{G}$ is computable.
\end{proof}

The second type of operator we are going to consider is at the basis of the Russian school of constructive mathematics. Contrary to computable operators these operators are only defined for functions in $\SSS^{(n)}_{A}$.

\begin{defin}
For $i > 0$, let $\theta^{(i)}$ be a quasi-G\"odel numbering of $\SSS^{(i)}_{A}$.
An operator $\fun{G}{\SSS^{(n)}_{A}}{\SSS^{(m)}_{A}}$ is called \emph{Markov-computable} if there is a function $g \in \RRR^{(1)}$ so that
\[
G(\theta^{(n)}_{i}) = \theta^{(m)}_{g(i)}.
\]
We say that $g$ \emph{realises} the operator $G$.
\end{defin}

Our aim is to derive an analogue of the Myhill-Shepherdson theorem~\cite{ms55} for the function classes considered here. We break the proof down in several steps.

\begin{lemma}\label{lem-ms1}
The restriction of a computable operator $\fun{\widehat{G}}{\widehat{\SSS}^{(n)}_{A}}{\widehat{\SSS}^{(m)}_{A}}$ to $\SSS^{(n)}_{A}$ is Markov-computable.
\end{lemma}
\begin{proof}
We use Lemma~\ref{lem-meth1} to construct a function $k' \in \RRR^{(2)}$ such that $\lambda t.\ k'(i, t)$ enumerates the extended graph of the function $G(\theta^{(n)}_{i})$, where $G$ is the restriction of $\widehat{G}$ to $\SSS^{(n)}_{A}$. The existence of a function realising $G$ is then a consequence of Condition~(QGN~II). Let  $k \in \RRR^{(2)}$ be as in Theorem~\ref{thm-qgn1}. Then $\lambda t.\ k(i, t)$ enumerates the graph of $\theta^{(n)}_{i}$. Moreover, let the c.e.\ set $C$ define $\widehat{G}$. Set
\begin{gather*}
\widehat{Q}(i, \pair{\pair{\pair{\vec y}, x}, a}, t, z)
\Leftrightarrow \vec{y} = \arg(z) \wedge (\forall c \le \lth(a)) (\exists t' \le t)\, k(i, t') = \\
\hspace{9cm} \pair{\pi^{(2)}_{1}((a)_{c}), 1+\pi^{(2)}_{2}((a)_{c})}, \\
Q(i, t, z) \Leftrightarrow (\exists \pair{\pair{\pair{\vec y}, x},a} \in C_{t})\, \widehat{Q}(i, \pair{\pair{\pair{\vec y}, x}, a}, t, z), \\
f(i, t, z) \Def \pi^{(2)}_{2}(\pi^{(2)}_{1}(\mu \pair{\pair{\pair{\vec y}, x}, a} \in C_{t}.\ Q(i, \pair{\pair{\pair{\vec y}, x}, a},t, z))). 
\end{gather*}
By comparing this definition with the condition that $C$ defines $\widehat{G}$ one sees that the function $k' \in \RRR^{(2)}$ constructed as in Lemma~\ref{lem-meth1} has the required property. 
\end{proof}

\begin{lemma}\label{lem-mon}
Every Markov-computable operator  $\fun{G}{\SSS^{(n)}_{A}}{\SSS^{(m)}_{A}}$ is monotone. That is, for $r, s \in \SSS^{(n)}_{A}$,
\[
\graph(s) \subseteq \graph(r) \Rightarrow \graph(G(s)) \subseteq \graph(G(r)).
\]
\end{lemma}
\begin{proof}
Assume to the contrary that there are $r, s \in \SSS^{(n)}_{A}$ so that $\graph(s) \subseteq \graph(r)$, but $\graph(G(s)) \not\subseteq \graph(G(r))$. Since for every $q \in \SSS^{(n)}_{A}$,
\[
\graph(q) = \bigcup\set{\graph(p)}{p \in \anf^{(n)}_{A} \wedge \graph(p) \subseteq \graph(q)},
\]
there is some $p \in \anf^{(n)}_{A}$ with
\[
\graph(p) \subseteq \graph(G(s)) \quad\text{and}\quad  \graph(p) \not\subseteq \graph(G(r)).
\]
Let us suppose for the moment that a function $\hat{k} \in \RRR^{(2)}$ can be constructed so that $\lambda t.\ \hat{k}(i, t)$ enumerates the extended graph of $s$, if $i \in \omega \setminus K_{A}$, and the extended graph of $r$, otherwise. Then it follows for the function $v \in \RRR^{(1)}$ existing for this $k$ by (QGN~II) that
\[
\theta^{(n)}_{v(i)} = \begin{cases}
					s & \text{if $i \in \omega \setminus K^{A}$}, \\
					r & \text{otherwise.}
				\end{cases}
\]
If $G$ is realised by $g \in \RRR^{(1)}$, we obtain
\[
\graph(p) \subseteq \graph(\theta^{(n)}_{g(v(i))}) \Leftrightarrow \theta^{(n)}_{v(i)} = s \Leftrightarrow i \in \omega \setminus K^{A}.
\]
By Lemma~\ref{lem-alph}(\ref{lem-alph-5}) the set $\set{i}{\graph(p) \subseteq \graph(\theta^{(n)}_{g(v(i))})}$ is c.e. and hence $\omega \setminus K^{A}$ is c.e.\ as well, which is false. Thus we have a contradiction.

It remains to show that a function $\hat{k}$ as above can indeed be constructed. Let again $k \in \RRR^{(1)}$ be as in Lemma~\ref{thm-qgn1}. Then $\lambda t.\ k(i, t)$ enumerates the extended graph of $\theta^{(n)}_{i}$. Moreover, let $j, j'$, respectively, be $\theta^{(n)}$-indices of $s$ and $r$, and let $h \in \RRR^{(1)}$ enumerate $K^{A}$. Then define $\hat{k}$ by
\[
\hat{k}(i, t) \Def \begin{cases}
				k(j, t) & \text{if $h(t') \neq i$, for all $t' \le t$}, \\
				k(j', t) & \text{otherwise.}
			\end{cases}
\]
Then $\hat{k}$ is as wanted.
\end{proof}
	
\begin{lemma}\label{lem-msct}	
Let $\fun{G}{\SSS^{(n)}_{A}}{\SSS^{(m)}_{A}}$ be Markov-computable. Then, for every $r \in \SSS^{(n)}_{A}$,
\[
\graph(G(r)) = \bigcup\set{\graph(G(s))}{s \in \anf^{(n)}_{A} \wedge \graph(s) \subseteq \graph(r)}.
\]
\end{lemma}
\begin{proof}
By the previous lemma the set $\set{\graph(G(s))}{s \in \anf^{(n)}_{A} \wedge \graph(s) \subseteq \graph(r)}$ is a chain with respect to inclusion. For $r \in \anf^{(n)}_{A}$, $\graph(r)$ is contained in this set. Hence, the statement holds trivially. It remains to consider the case that $r \in \RRR^{(n)}$. Because of Lemma~\ref{lem-mon} it suffices to show that
\[
\graph(G(r)) \subseteq \bigcup\set{\graph(G(s))}{s \in \anf^{(n)}_{A} \wedge \graph(s) \subseteq \graph(r)}.
\]
Assume to the contrary that this inclusion is wrong for some $r \in \RRR^{(n)}$. Then there exists $q \in \anf^{(m)}_{A}$ so that  $\graph(q) \subseteq \graph(G(r))$, but 
\begin{equation}\label{eq-cont}
\graph(q) \not\subseteq \bigcup\set{\graph(G(s))}{s \in \anf^{(n)}_{A} \wedge \graph(s) \subseteq \graph(r)}.
\end{equation}
Statement (\ref{eq-cont}) holds exactly if for all $s \in \anf^{(n)}_{A}$ with $\graph(s) \subseteq \graph(r)$, $\graph(q) \not\subseteq \graph(G(s))$.

Assume for the moment that there is some $v \in \RRR^{(1)}$ with $\theta^{(n)}_{v(i)} \in \anf^{(n)}_{A}$ and $\graph(\theta^{(n)}_{v(i)}) \subseteq \graph(r)$, if $ i \in K^{A}$, and $\theta^{(n)}_{v(i)} = r$, if $i \notin K^{A}$. Let $g \in \RRR^{(1)}$ realise $G$. Then we obtain
\[
\graph(q) \subseteq \graph(\theta^{(m)}_{g(v(i))}) \Leftrightarrow \theta^{(n)}_{v(i)} = r \Leftrightarrow i \in \omega \setminus K^{A}.
\]
As seen in the last proof, this is impossible. So, the assumption made at the beginning is wrong.

We will now consider the construction of function $v$. Let to this end $h, f_{A} \in \RRR^{(1)}$, respectively, be enumerations of $K^{A}$ and $A$. Moreover, let $\hat{k}, k \in \RRR^{(2)}$ be defined by 
\begin{align*}
&\hat{k}(i, t) \Def \begin{cases}
				\hat{k}(i, t-1) & \text{if for $t' < t$ and $t'' \le t$, $h(t') = i$ and $t = f_{A}(t'')^{n}$}, \\
				\pair{\pair{\kappa(t)}, r(\kappa(t))+1}  &  \text{otherwise}, 
			\end{cases}\\[1.5ex]
&k(i, 2t) \Def \pair{t, 0}, \\
&k(i, 2t+1) \Def \hat{k}(i, t).
\end{align*}
If $i \notin K^{A}$, $\lambda t.\ k(i, t)$ enumerates the extended graph of $r$; otherwise,  $\lambda t.\ \hat{k}(i, t)$ step by step lists 
\[
\pair{\pair{\kappa(0)}, r(\kappa(0))+1}, \pair{\pair{\kappa(1)}, r(\kappa(1))+1}, \ldots
\]
until $i$ has been found in $K^{A}$. If so, it continues this enumeration until an initial segment of $\omega^{n}$ with an edge length in $A$ has been enumerated by $\kappa(0), \ldots, \kappa(t)$. Therefore, in this case, $\lambda t.\ k(i, t)$ enumerates the extended graph of a function $s \in \anf^{(n)}_{A}$ with $\graph(s) \subseteq \graph(r)$. It follows that the function $v \in \RRR^{(1)}$ existing for this $k$ by Condition~(QGN~II) has the appropriate properties. 
\end{proof}
	
If $\alpha^{(n)} = \theta^{(n)} \circ f$, for some $f \in \RRR^{(n)}$ and the Markov-computable operator $\fun{G}{\SSS^{(n)}_{A}}{\SSS^{(m)}_{A}}$ is realised by $g \in \RRR^{(n)}$, then $G(\alpha^{(n)}_{i}) = \theta^{(n)}_{g(f(i))}$. With Lemma~{\ref{lem-alph}(\ref{lem-alph-5}) we also obtain that $\{\, \pair{i, j} \mid \graph(\alpha^{(m)}_{i}) \subseteq \graph(G(\alpha^{(n)}_{j})) \,\}$ is c.e. We have now completed all the necessary steps to derive the result we were aiming for.

\begin{theorem}[Myhill, Shepherdson]\label{thm-mysh}
\begin{enumerate}

\item	\label{thm-mysh-1} \sloppy
The restriction of every computable operator $\fun{\widehat{G}}{\widehat{\SSS}^{(n)}_{A}}{\widehat{\SSS}^{(m)}_{A}}$ to $\SSS^{(n)}_{A}$ is Markov-computable.

\item	\label{thm-mysh-2}
Every Markov-computable operator $\fun{G}{\SSS^{(n)}_{A}}{\SSS^{(m)}_{A}}$ is the restriction to $\SSS^{(n)}_{A}$ of a computable operator $\fun{\widehat{G}}{\widehat{\SSS}^{(n)}_{A}}{\widehat{\SSS}^{(m)}_{A}}$.

\end{enumerate}
\end{theorem}
\begin{proof}
(\ref{thm-mysh-1}) has been shown in Lemma~\ref{lem-ms1}. 

(\ref{thm-mysh-2}) By Lemma~\ref{lem-mon} the set $\set{\graph(G(s))}{s \in \anf^{(n)}_{A} \wedge \graph(s) \subseteq \graph(r)}$ with $r \in \SSS^{(n)}_{A}$ is a chain with respect to set inclusion. Thus, the union over this chain is single-valued, that is, the graph of an $m$-ary function. Let $\widehat{G}(r)$ be this function. Then we have
\[
\graph(\widehat{G}(r)) = \bigcup\set{\graph(G(s))}{s \in \anf^{(n)}_{A} \wedge \graph(s) \subseteq \graph(r)}.
\]
The chain is either finite or infinite. In the first case $\graph(\widehat{G}(r))$ is its maximal element. Since for all $s \in \anf^{(n)}_{A}$, $G(s) \in \SSS^{(m)}_{A}$, it follows that also $\widehat{G}(r) \in \SSS^{(m)}_{A}$. In the other case, every element of the chain is the graph of a function in $\SSS^{(m)}_{A}$. The domains of these functions in the chain therefore form an increasing chain of subsets of $\omega^{m}$. it follows that $\widehat{G}(r)$ is total in this case. Hence, in both cases, $\widehat{G}(r) \in \widehat{\SSS}^{(m)}_{A}$. With Lemma~\ref{lem-mon} and the remark preceding this theorem it follows that $\widehat{G}$ satisfies Conditions~(\ref{thm-coprop-1})-(\ref{thm-coprop-3}) in Theorem~\ref{thm-coprop}. Hence, $\widehat{G}$ is computable. In case that $r \in \SSS^{(n)}_{A}$ it moreover follows with Lemma~\ref{lem-msct} that $\widehat{G}(r) = G(r)$. That is, $G$ is the restriction of $\widehat{G}$ to $\SSS^{(n)}_{A}$.
\end{proof}

\section{Quasi-G\"odel numberings and the Rogers semi-lattice of computable numberings} \label{sec-semlat}

The aim of this section is to look at the quasi-G\"odel numbered sets $\SSS^{(1)}_{A}$ from a numbering-theoretic perspective. We limit ourselves to unary functions, since every $n$-ary function $f \in \SSS^{(n)}_A$ converts into the unary function $r \circ \kappa \in \SSS^{(1)}_{A'}$ with $A' = \set{a^n}{a \in A}$ by means of a one-to-one and initial segment-wise enumeration $\kappa$ of $\omega^n$. 

As will be seen, any two quasi-G\"odel numberings of $\SSS^{(1)}_{A}$ are computably isomorphic. Furthermore, $(\SSS_{A}^{(1)}, \theta^{A})$ is a retract of $(\PPP^{(1)}, \varphi)$, where $\theta^{A}$ and $\varphi$, respectively, are a quasi-G\"odel and a G\"odel numbering. If $A, B$ are infinite c.e.\ subsets of $\omega$ with $A \subseteq B$, then $(\SSS_{A}^{(1)}, \theta^{A})$ is a retract from $(\SSS_{B}^{(1)}, \theta^{B})$. Thus,  there are infinitely descending chains $((\SSS_{A_{\nu}}^{(1)}, \theta^{A_{\nu}}))_{\nu \in \omega}$ with  $(\SSS_{A_{\nu+1}}^{(1)}, \theta^{A_{\nu+1}})$ being a retract of $(\SSS_{A_{\nu}}^{(1)}, \theta^{A_{\nu}})$ and $(\SSS_{A_{0}}^{(1)}, \theta^{A_{0}})$ being a retract of $(\PPP^{(1)}, \varphi)$. It could still be that there are subsets $T \subseteq \anf_{\omega}^{(1)}$, which are not necessarily of the form $\anf_{A}^{(1)}$,
so that $\RRR^{(1)} \cup T$ is quasi-G\"odel numberable. Another central result is that the set of all such extensions of $\RRR^{(1)}$ has no $\subset^{*}$-minimal elements.

We will also study the Rogers semi-lattice of computable numberings of $\SSS^{(1)}_{A}$. A central result is every countable partially ordered set can isomorphically be embedded into the Rogers semi-lattice of computable numberings of $\SSS^{(1)}_{A}$. In addition, the semi-lattice contains infinitely many Friedberg numberings and positive numberings to which no Friedberg numbering can be reduced.

\begin{proposition}\label{pn-qgnchar}
Let $\eta, \gamma$ be numberings of $\SSS^{(1)}_{A}$. Then the following three statements hold:
\begin{enumerate}

\item \label{pn-qgnchar-1}
If $\gamma$ is a quasi-G\"odel numbering and $\eta \le_{m} \gamma$, then $\eta$ satisfies Condition~(QGN~I).

\item \label{pn-qgnchar-2}
If $\eta$ is a quasi-G\"odel numbering and $\eta \le_{m} \gamma$, then $\gamma$ satisfies Condition~(QGN~II).

\item \label{pn-qgnchar-3}
If $\eta, \gamma$, respectively, satisfy (QGN~I) and (QGN~II), then $\eta \le_{m} \gamma$.

\end{enumerate}
\end{proposition}
\begin{proof}
(\ref{pn-qgnchar-1}) Since $\eta \le_{m}\gamma$, there is some $g \in \RRR^{(1)}$ with $\eta = \gamma \circ g$. By assumption, 
\[
\egraph(\lambda i, x.\ \gamma_{i}(x))
\]
 is c.e. Hence this is also true for $\egraph(\lambda i, x.\ \gamma_{g(i)}(x))$.

(\ref{pn-qgnchar-2}) Let  $g \in \RRR^{(1)}$ with $\eta = \gamma \circ g$. Moreover, let $k \in \RRR^{(2)}$ and $r \in \SSS^{(1)}_{A}$ so that $\lambda t.\ k(i, t)$ enumerates the extended graph of $r$. Since $\eta$ satisfies (QGN~II), there is some $v \in \RRR^{(1)}$ with $r = \eta_{v(i)}$. Set $v' \Def g \circ v$. Then $r = \gamma_{v'(i)}$. Whence also $\gamma$ satisfies (QGN~II).

(\ref{pn-qgnchar-3}) Let $k \in \RRR^{(2)}$ be as in Theorem~\ref{thm-qgn1}. Then $\lambda t.\ k(i, t)$ enumerates the extended graph of function $\eta_{i}$. Since $\gamma$ satisfies (QGN~II) there is some $v \in \RRR^{(1)}$ with $\eta_{i} = \gamma_{v(i)}$. Thus, $\eta \le_{m} \gamma$.
\end{proof}

It follows that, if $\gamma$ is a quasi G\"odel numbering of $\SSS^{(1)}_{A}$, then $\eta$ is also a quasi-G\"odel numbering of $\SSS^{(1)}_{A}$, exactly if $\eta \equiv_{m} \gamma$. In particular, all quasi-G\"odel numberings of $\SSS^{(1)}_{A}$ are $m$-equivalent. As we will see next, they are maximal among all numberings of $\SSS^{(1)}_{A}$ satisfying (QGN~I), and minimal among those satisfying (QGN~II).

\begin{theorem}\label{thm-maxmin}
Let $\eta$ be a numbering of $\SSS^{(1)}_{A}$. Then the following three statements are equivalent:
\begin{enumerate}

\item\label{thm-maxmin-1}
$\eta$ is a quasi-G\"odel numbering.

\item\label{thm-maxmin-2}
$\eta$ satifies (QGN~I) and for all numberings $\gamma$ of $\SSS^{(1)}_{A}$, if $\gamma$ satisfies (QGN~I) then $\gamma \le_{m} \eta$.

\item\label{thm-maxmin-3}
$\eta$ satifies (QGN~II) and for all numberings $\gamma$ of $\SSS^{(1)}_{A}$, if $\gamma$ satisfies (QGN~II) then $\eta \le_{m} \gamma$.

\end{enumerate}
\end{theorem}
\begin{proof}
Assume (\ref{thm-maxmin-1}). Then (\ref{thm-maxmin-2}) follows with Proposition~\ref{pn-qgnchar}(\ref{pn-qgnchar-3}). Conversely, suppose (\ref{thm-maxmin-2}) and let $\leftidx{^{A}}\psi$ be a quasi-G\"odel numbering as in Theorem~\ref{thm-qgn2}. Then $\leftidx{^{A}}\psi \le_{m}\eta$, by our assumption. With Proposition~\ref{pn-qgnchar}(\ref{pn-qgnchar-2}) it follows that $\eta$ satisfies (QGN~II). Since it also satisfies (QGN~I), $\eta$ is a quasi-G\"odel numbering.

In a similar way it follows that (\ref{thm-maxmin-1}) and (\ref{thm-maxmin-3}) are equivalent.
 \end{proof}

At the end of Section~\ref{sec-comp} we pointed out that quasi-G\"odel numberings are pre-complete. This will be strengthened next.

\begin{theorem}\label{thm-compl}
Let $\theta$ be a quasi-G\"odel numbering of $\SSS^{(1)}_{A}$. Then the following two statements hold:
\begin{enumerate}

\item \label{thm-compl-1}
$\theta$ is complete with the nowhere defined function as distinguished element. This means that for any $p \in \PPP^{(1)}$ there exists $g \in \RRR^{(1)}$ such that
\[
\theta_{g(i)}(x) = \begin{cases}
				\theta_{p(i)}(x) & \text{if $i \in \dom(p)$}, \\
				\text{undefined} & \text{otherwise}.
			\end{cases}
\]

\item \label{thm-compl-2}
$\theta$ is cylindrical, that is, $\theta \equiv \theta^{c}$, where $\theta^{c}_{\pair{i, j}} = \theta_{j}$.

\end{enumerate}
\end{theorem}
\begin{proof}
(\ref{thm-compl-1})
Let $k \in \RRR^{(2)}$ be as in Theorem~\ref{thm-qgn1}. Then $\lambda t.\ k(i, t)$ enumerates the extended graph of $\theta_{i}$. Moreover, let $k' \in \RRR^{(2)}$ be defined by
\begin{align*}
&k'(i, 2t) \Def \pair{t, 0}, \\
&k'(i, 2t+1) \Def \begin{cases}
				k(i, t \prc \mu t'.\ p(i)\conv_{t'}) & \text{if for some $t' < t$, $p(i)\conv_{t'}$}, \\
				\pair{t, 0} & \text{otherwise}.
			\end{cases}
\end{align*}
Then $\lambda t.\ k'(i, t)$ enumerates the extended graph of $\theta_{i}$, if $i \in \dom(p)$. Otherwise, $\lambda t.\ k'(i, t)$ enumerates the extended graph of the nowhere defined function. Let $g \in \RRR^{(1)}$ be the function existing by (QGN~II). Then we have for $i \in \dom(p)$ that $\theta_{g(i)} = \theta_{p(i)}$. In the other case, $\theta_{g(i)}$ is nowhere defined.

(\ref{thm-compl-2}) Let $\theta^{c}$ with $\theta^{c}_{\pair{i, j}} = \theta_{j}$ be the cylinder of $\theta$. Then $\theta \le_{1}\theta^{c}$. Since by the Ershov-Myhill isomorphism theorem~\cite[p.~295]{er73} 1-equivalence coincides with computable isomorphism, it remains to show that also $\theta^{c} \le_{1} \theta$. Let to this end $p \in \RRR^{(2)}$ be the padding function existing by Theorem~\ref{thm-pad}. Then
\[
\theta^{c}_{\pair{i, j}} = \theta_{j} = \theta_{p(i, j)}.
\]
Thus, $\theta^{c}= \theta \circ g$ with $g(\pair{i, j}) = p(j, i)$. That is, $\theta^{c} \le_{1} \theta$.
\end{proof}

As we have seen above, all quasi-G\"odel numberings are $m$-equivalent. Since such numberings possess padding functions, they are even 1-equivalent. With the Ershov-Myhill isomorphism theorem we thus obtain an extension of Rogers' Isomorphism Theorem~\cite{ro58} for G\"odel numberings to quasi-G\"odel numberings.

\begin{theorem}[Isomorphism Theorem]\label{thm-iso}
All quasi-G\"odel numberings are computably isomorphic. 
\end{theorem}

It follows that a numbering of $\SSS^{(1)}_{A}$ is a quasi-G\"odel numbering, if it is computably isomorphic to a quasi-G\"odel numbering.

In Section~\ref{sec-isf} we fixed a G\"odel numbering $\varphi$ to construct a special $\varphi$-standard numbering of $\SSS^{(1)}_{A}$ which then turned out to be a quasi-G\"odel numbering. Thus, every quasi-G\"odel numbering of $\SSS^{(1)}_{A}$ is computable isomorphic to a special $\varphi$-standard numbering of $\SSS^{(1)}_{A}$. This result can be strengthened again.

\begin{theorem}\label{thm-qgnstand}
For every quasi-G\"odel numbering $\theta$ there is a G\"odel numbering $\xi$ so that $\theta$ is a special $\xi$-standard numbering.
\end{theorem}
\begin{proof}
Let $\leftidx{^{A}}\varphi$ be the special standard numbering von $\SSS^{(1)}_{A}$ used in Section~\ref{sec-isf}. Then $\leftidx{^{A}}\varphi = \varphi \circ f$, for some $f \in \RRR^{(1)}$. Moreover, there is a one-to-one and onto function $g \in \RRR^{(1)}$ with $\theta = \leftidx{^{A}}\varphi \circ g$. Set $\xi \Def \varphi \circ g$. Then $\xi$ is a G\"odel numbering, and as is readily verified, $\xi \circ (g^{-1} \circ f \circ g)$ is a special $\xi$-standard numbering of $\SSS^{(1)}_{A}$. Moreover, $\theta = \xi \circ (g^{-1} \circ f \circ g)$.
\end{proof}

It follows that the construction of a quasi-G\"odel numbering of $\SSS^{(1)}_{A}$ in Section~\ref{sec-isf} has universal character: every quasi-G\"odel numbering is obtained in this way.

In a next step we will examine which other properties the quasi-G\"odel numbered set $(\SSS^{(1)}_{A}, \theta)$ has as a subobject of $(\PPP^{(1)}, \varphi)$. Here, a subset $Z \subseteq \PPP^{(1)}$ with a numbering $\zeta$ is called \emph{subobject} of $(\PPP^{(1)}, \varphi)$ if $\zeta \le_{m} \varphi$ (cf.\ \cite{er73,er75}).

$(Z, \zeta)$ is an \emph{sn-subobject} of $(\PPP^{(1)}, \varphi)$ if, in addition, there is some $g \in \RRR^{(1)}$ so that $\zeta \circ g$ is a special $\varphi$-standard n numbering of $Z$.

$(Z, \zeta)$ is called \emph{r-subobject} or \emph{retract} of $(\PPP^{(1)}, \varphi)$ if, in addition, there is an idempotent onto function $\fun{H}{\PPP^{(1)}}{Z}$ and some $h \in \RRR^{(1)}$ such that $H(\varphi_{i}) = \zeta_{h(i)}$.

$(Z, \zeta)$ is an \emph{n-subobject} of $(\PPP^{(1)}, \varphi)$ if, in addition, there exist $f \in \RRR^{(1)}$ such that for all $i \in I_{\varphi}(Z)$, $\varphi_{i} = \zeta_{f(i)}$. In this case the numbering $zeta$ is called \emph{$\varphi$-normal} (cf.\ \cite{ma71}).

$(Z, \zeta)$ is called \emph{principal subobject} of $(\PPP^{(1)}, \varphi)$ and $\zeta$ \emph{$\varphi$-principal numbering}, if for every numbering $\eta$ of $Z$ with $\eta \le_{m} \varphi$ one has that $\eta \le_{m} \zeta$.

As seen above, the quasi-G\"odel numbering $\theta$ is computably isomorphic to a special $\varphi$-standard numbering of $\SSS^{(1)}_{A}$. This shows

\begin{corollary}\label{cor-snsub}
$(\SSS^{(1)}_{A}, \theta)$ is an sn-subobject of $(\PPP^{(1)}, \varphi)$,
\end{corollary}

In Section~\ref{sec-syn} we have considered a Turing program $M^{(1)}_{A}$ that calls arbitrary Turing programs $P$ so that $M^{(1)}_{A}(P)$ is a program that computes the functions in $\SSS^{(1)}_{A}$. By considering the semantics of the programs we obtain an idempotent onto map $\fun{M_{A}}{\PPP^{(1)}}{\SSS^{(1)}_{A}}$ such that $M_{A}(\varphi_{i}) = \leftidx{^{A}}\varphi_{i}$. As we have seen above, $\theta$ is computably isomorphic to $\leftidx{^{A}}\varphi_{i}$.

\begin{theorem}\label{thm-ret}
$(\SSS^{(1)}_{A}, \theta)$ is a retract of $(\PPP^{(1)}, \varphi)$.
\end{theorem}

As is shown by Ershov~\cite[$\S 4$]{er73}, every retract of $(\PPP^{(1)}, \varphi)$ is an n-subobject and each n-subobject is a principal subobject of $(\PPP^{(1)}, \varphi)$. Thus, every quasi-G\"odel numbering $\theta$ of $\SSS^{(1)}_{A}$  is $\varphi$-normal and a $\varphi$-principal numbering. Since  every numbering of $\SSS^{(1)}_{A}$ that satisfies (QGN~I) is $m$-reducible to $\varphi$, it follows with Theorem~\ref{thm-maxmin} that each $\varphi$-principal numbering of $\SSS^{(1)}_{A}$ is a quasi-G\"odel numbering is. As a consequence, every $\varphi$-normal numbering of $\SSS^{(1)}_{A}$ is a quasi-G\"odel numbering.

\begin{theorem}\label{thm-coinc}
Let $\hpt^{\varphi}_{A}$, $\nor^{\varphi}_{A}$ and $\qgn^{\varphi}_{A}$, respectively, be the set of $\varphi$-principal, $\varphi$-normal and quasi-G\"odel numberings of $\SSS^{(1)}_{A}$. Then
\[
\hpt^{\varphi}_{A} = \nor^{\varphi}_{A} = \qgn^{\varphi}_{A}.
\]
\end{theorem}

If $B$ is a c.e.\ superset of $A$ and $\gamma$ is a quasi-G\"odel numbering of $\SSS^{(1)}_{B}$, then the above results remain true if $\varphi$ is replaced by $\gamma$ and $(\PPP^{(1)}, \varphi)$ by $(\SSS^{(1)}_{B}, \gamma)$. In particular we obtain an analogue of Theorem~\ref{thm-ret}.

\begin{theorem}\label{thm-ret2}
$(\SSS^{(1)}_{A}, \theta)$ is a retract of $(\SSS^{(1)}_{B}, \gamma)$.
\end{theorem}

So far in this work it was shown that the function classes $\SSS^{(n)}_A$ have a quasi-G"odel numbering $\theta^{(n)}$ for every infinite c.e.\ set $A$ and that for this a sufficiently rich computability theory can be developed in which the same results apply as in the computability theory for all partial computable functions. If there are classes of functions that include the total computable functions but do not contain all partial computable functions and for which a satisfactory computability theory can be developed, one naturally wonders whether there are $\subset^{*}$-minimal classes of this kind. Here, $X \subset^{*} Y$, if $X \subset Y$ and $Y \setminus X$ is infinite. We do not want to treat this problem in full generality, but want to examine whether the family $\bb$ of sets contains $\subset^{*}$-minimal elements, where
\[
\bb \Def \set{\RRR^{(1)} \cup T}{\text{$T \subseteq \anf^{(1)}_{\omega}$ and $\RRR^{(1)} \cup T$ has a quasi-G\"odel numbering}}
\]
Set $h(0, j) \Def j$ and $h(i+1, j) \Def 2^{h(i, j)}$. Then each set $A_{i} \Def \set{h(i, j)}{j \in \omega}$ with $i \in \omega$ is computable and infinite. Since $A_{i+1} \subset^{*} A_{i}$, we obtain that $\bb$ contains infinite descending chains with respect to $\subset^{*}$.

If every function class in $\bb$ were of the form $\RRR^{(1)} \cup \anf^{(1)}_A$ with an infinite c.e.\ set $A$, we would of course know that $\bb$ contains no $\subset^{*}$-minimal elements, since from any infinite c.e.\ set one can remove infinitely many elements without destroying the property of being both infinite and c.e. However, we do not know whether there are other sets $T \subseteq \anf^{(1)}_\omega$ for which $\RRR^{(1)} \cup T$ has a quasi-G\"odel numbering.  Menzel and
Sperschneider~\cite[p.~76]{mes84} show that the family of sets 
\[
\set{\RRR^{(1)} \cup T}{\text{$T \subseteq \anf^{(1)}_{\omega}$ and $\RRR^{(1)} \cup T$ enumerable in $\varphi$}}
\]
does not contain $\subset^{*}$-minimal elements. Here, for a set $X$ with numbering $\delta$, $Y \subseteq X$ is \emph{enumerable in $\delta$} if $Y = \delta[E]$, for some c.e.\ set $E \subseteq \omega$.
If $\RRR^{(1)} \cup T$ now has a quasi-G\"odel numbering, then this is $m$-reducible to $\varphi$. So $\RRR^{(1)} \cup T$ is also enumerable in $\varphi$. This leads us to the following result:

\begin{theorem}\label{thm-noex}
$\bb$ has no $\subset^{*}$-minimal elements.
\end{theorem}.

Let $\bn_{A}$ be the set of numberings of $\SSS^{(1)}_A$ that satisfy Condition~(QGN~I). 
As was shown at the beginning of this section, the quasi-G\"odel numberings of $\SSS^{(1)}_A$ are just those numberings in $\bn_{A}$ that are maximal with respect to $\le_m$. The numberings in $\bn_{A}$ are exactly those numberings of $\SSS^{(1)}_A$ that have a computable universal function (which is however in $\PPP^{(2)} \setminus \SSS^{(2)}_{A}$). Therefore, these numberings are called \emph{computable}.

For two numberings $\eta$ and $\gamma$ of $\SSS^{(1)}_A$, $\eta \oplus \gamma$ defined by
\[
(\eta \oplus \gamma)(2i) \Def \eta(i) \quad\text{and}\quad (\eta \oplus \gamma)(2i+1) \Def \gamma(i)
\]
is the direct sum of $\eta$ and $\gamma$. Let $[\eta]$ be the equivalence class of all numberings of $\SSS^{(1)}_A$ that are $m$-equivalent with $\eta$. As is well known, the $m$-reducibility relation can be lifted to the equivalence classes yielding a partial order on the set of these classes denoted by $\le$. As shown in~\cite{er73}, the collection of all equivalence classes is an upper semi-lattice, usually called \emph{Rogers semi-lattice}, with respect to this order with $[\eta\oplus\gamma]$ as least upper bound of $[\eta]$ and $[\gamma]$. Because the direct sum of two computable numberings is computable again, we obtain for the set $[\bn_{A}] \Def \bn_{A}/{\le_{m}}$:

\begin{proposition}\label{pn-cnum}
\begin{enumerate}

\item \label{pn-cnum-1}
$([\bn_{A}], \le)$ is a Rogers sub-semi-lattice of the Rogers semi-lattice of all numberings of $\SSS^{(1)}_{A}$.

\item \label{pn-cnum-2}
$\qgn_{A}$ is the greatest element in $([\bn_{A}], \le)$.

\end{enumerate}
\end{proposition}

As we will see later in this section, $([\bn_{A}], \le)$ is not a lattice. The next result is an analogue of a theorem by Goetze~\cite{goe74,goe76} for the Rogers semi-lattice of the computable numberings of $\PPP^{(1)}$.

\begin{theorem}\label{thm-goe}
Every countable partially ordered set can be isomorphically embedded into the Rogers semi-lattice of computable numberings of $\SSS^{(1)}_{A}$.
\end{theorem}

We derive this result in several steps. Let to this end $\theta$ be a quasi-G\"odel numbering of $\SSS^{(1)}_{A}$, $a$ be the least positive element of $A$ and $\cs$ be the set of all computable subsets of $\omega$. Then $(\cs, \subseteq)$ is a lattice with the empty set as least element. For every $C \in \cs$ let the numbering $\theta^{C}$ be defined by
\[
\theta^{C}_{i}(x) \Def \begin{cases}
					\pi^{(a+1)}_{1}(i) \prc 1  & \text{if $x=0$ and $\pi^{(a+1)}_{1}(i) \in C+1$}, \\
					\vdots & \vdots\\
					\pi^{(a+1)}_{a}(i) \prc 1  & \text{if $x=a-1$ and $\pi^{(a+1)}_{1}(i) \in C+1$}, \\
					\theta_{\pi^{(a+1)}_{a+1}(i)}(x)  & \text{if $x \ge a$ and $\pi^{(a+1)}_{1}(i) \in C+1$, or $\pi^{(a+1)}_{1}(i) \notin C+1$ and} \\
					& \text{$\theta_{\pi^{(a+1)}_{a+1}(i)}(0) \notin C$,} \\
					\text{undefined}             & \text{otherwise}.
				\end{cases}
\]
Here, $C+1 \Def \set{b+1}{b \in C}$. Note that by the choice of $a$, for all $j$ with $0 \in \dom(\theta_{j})$, $\theta_{j}$ is defined at least on an initial segment of $\omega$ of length $a$. Therefore, $\theta^{C}_{i} \in \SSS^{(1)}_{A}$. 

\begin{lemma}\label{lem-goe1}
For all $C \in \cs$, $\theta^{C} \in \bn_{A}$.
\end{lemma}
\begin{proof}
Since $\theta$ satisfies (QGN~I) there is an enumeration $h \in \RRR^{(1)}$ of the extended graph of the universal function of $\theta$. Then we have
\begin{align*}
&\pair{\pair{i, x}, z} \in \egraph(\lambda i, x.\ \theta^{C}_{i}(x)) \Leftrightarrow \mbox{} \\
 &\hspace{2em}z = 0 \vee \bigvee_{\nu = 1}^{n} [x = \nu-1 \wedge \pi^{(a+1)}_{1}(i) \in C+1 \wedge z = \pi^{(a+1)}_{\nu}(i) ]\vee [(\exists t)\, h(t) = \mbox{} \\
&\hspace{2em} \pair{\pair{\pi^{(a+1)}_{a+1}(i), x}, z} \wedge [[x \ge a \wedge \pi^{(a+1)}_{1}(i) \in C+1] \vee [\pi^{(a+1)}_{1}(i) \notin C+1 \wedge  \mbox{} \\  
&\hspace{2em} (\exists t')\, [\pi^{(2)}_{1}(h(t')) = \pair{\pi^{(a+1)}_{a+1}(i), 0} \wedge \pi^{(2)}_{2}(h(t')) > 0 \wedge \pi^{(2)}_{2}(h(t')) \notin C+1]]]].
\end{align*}
It follows that $\egraph(\lambda i, x.\ \theta^{C}_{i}(x))$ is c.e. So, it remains to show that $\range(\theta^{C}) = \SSS^{(1)}_{A}$. Let to this end $r \in \SSS^{(1)}_{A}$ with $r = \theta_{j}$. Then
\[
i \Def \begin{cases}
		\pair{r(0)+1, \ldots, r(a)+1, j} & \text{if $\dom(r)$ is not empty}, \\
		\pair{0, \ldots, 0, j}		   & \text{otherwise}
	\end{cases}	
\]
is a $\theta^{C}$-index of $r$.
\end{proof}

Note that by Theorem~\ref{thm-maxmin}(\ref{thm-maxmin-3}), $\theta^{\emptyset} \le_{m} \theta$. On the other hand, $\theta_{i}  = \theta^{\emptyset}_{\pair{0, \ldots, 0, i}}$. Thus, $\theta \equiv_{m} \theta^{\emptyset}$.

\begin{lemma}\label{lem-goe2}
$(\set{[\theta^{C}]}{C \in \cs}, \le)$ is a lattice that is dually isomorphic zu $(\cs, \subseteq)$. The mapping $J(C) \Def [\theta^{C}]$ is a dual isomorphism.
\end{lemma}
\begin{proof}
Let $C, E \in \cs$. First, we show that $\theta^{E} \le_{m} \theta^{C}$, if $E \subseteq C$. Consider the following algorithm:\\

\noindent
\emph{Input:} $i$ 
\begin{enumerate}

\item\label{newstep1}
If $\pi^{(a+1)}_{1}(i) \in C+1$ then set $\hat{i} := i$ and go to Step~(\ref{newstep4})

\item\label{newstep2}
If $\pi^{(a+1)}_{1}(i) \in E+1 \setminus  C+1$ then find $j$ so that
\[
\theta_{j}(x) = \begin{cases}
			\pi^{(a+1)}_{1}(i) \prc 1 & \text{if $x = 0$}, \\
			\vdots & \vdots \\
			\pi^{(a+1)}_{a}(i) \prc 1 & \text{if $x = a-1$}, \\
			\theta_{\pi^{(a+1)}_{a+1}(i)}(x) & \text{if $x \ge a$,}
		\end{cases}
\]
set $\hat{i} := \pair{\pi^{(a+1)}_{1}(i), \ldots, \pi^{(a+1)}_{a}(i), j}$ and go to Step~(\ref{newstep4}).

\item\label{newstep3}
(After the procedure has finished with this step, we know that $\pi^{(a+1)}_{1}(i) \notin E+1$.)\\
Find $j'$ with
\[
\theta_{j'}(x) = \begin{cases}
				\theta_{\pi^{(a+1)}_{a+1}(i)}(x)  & \text{if $\theta_{\pi^{(a+1}_{a+1}(i)}(0) \notin E$}, \\
				\text{undefined}   & \text{otherwise}
			\end{cases}
\]
and set $\hat{i} := \pair{\pi^{(a+1)}_{1}(i), \ldots, \pi^{(a+1)}_{a}(i), j'}$.

\item\label{newstep4}
STOP

\end{enumerate}
\emph{Output:} $\hat{i}$.\\

\noindent
As we will see in the next step, the indices $j$ and $j'$ in Steps~(\ref{newstep2}) and (\ref{newstep3}), respectively, can be computed from $i$. Therefore, the above is really an algorithm. Let $q \in \RRR^{(1)}$ be the function it computes. Then $\theta^{E} = \theta^{C} \circ q$, that is $\theta^{E} \le_{m} \theta^{C}$.

By Theorem~\ref{thm-qgn1} there is function $k \in \RRR^{(2)}$ such that for every $i \in \omega$, $\lambda i, t.\ k(i, t)$ enumerates the extended graph of $\theta_{i}$. Let 
\begin{gather*}
k'(i, t) \Def \begin{cases}
 			\pair{t, \pi^{(a+1)}_{1}(i)}	& \text{if $t=0$}, \\
			\vdots & \vdots \\
			\pair{t, \pi^{(a+1)}_{a}(i)}	& \text{if $t = a-1$}, \\
			k(\pi^{(a+1)}_{a+1}(i), t \prc a) & \text{if $t \ge a$, and $\pi^{(2)}_{2}(k(\pi^{(a+1)}_{a+1}(i), t \prc a)) = 0$ or} \\
			& \text{$\pi^{(2)}_{1}(k(\pi^{(a+1)}_{a+1}(i), t \prc a)) \ge a$},\\
			\pair{0, 0} & \text{otherwise}, 
		\end{cases}\\
Q(i, t) \Leftrightarrow \pi^{(2)}_{1}(k(\pi^{(a+1)}_{a+1}(i), t)) = 0 \wedge \pi^{(2)}_{2}(k(\pi^{(a+1}_{a+1}(i), t)) \notin E+1, \\
k''(i, t) \Def \begin{cases}
			k(\pi^{(a+1)}_{a+1}(i), t \prc \mu t' \le t.\ Q(i, t')) & \text{if for some $t' \le t$, $Q(i, t')$}, \\
			\pair{t, 0} 	& \text{otherwise}.
		\end{cases}
\end{gather*}
Then $\lambda t.\ k'(i, t)$ enumerates the extended graph of the function $\theta_{j}$ defined in Step~(\ref{newstep2}) of the above algorithm and $\lambda i, t.\ k''(i, t)$ the extended graph of the function $\theta_{j'}$ defined in Step~(\ref{newstep3}). Let $v', v'' \in \RRR^{(1)}$ be the functions now existing by Condition~(QGN~II). Then $\theta_{j} = \theta_{v'(i)}$ and $\theta_{j'} = \theta_{v''(i)}$. Thus, we can effectively find indices $j$ and $j'$ from given $i$ with the required properties. 

Next, we show that also conversely, $\theta^{E} \le_{m} \theta^{C}$ implies $C \subseteq E$. To this end we need the following result:

\begin{claim}\label{eq-zr}
Let $B$ be a computable set. Then the set $\set{i}{\theta^{B}(0) = y}$ is computable, exactly if $y \in B$.
\end{claim}
which we will prove in a later step. 

Assume that $\theta^{E} \le_{m} \theta^{C}$ and suppose that $y \in C$. Then we need to show that $y \in E$. Since $y \in C$, it follows with Claim~\ref{eq-zr} that the set $\set{i}{\theta^{C}(0) = y}$ is computable. Because $\theta^{E} \le_{m} \theta^{C}$, it follows that $\set{i}{\theta^{E}(0) = y}$ is computable as well. Hence $y \in E$, again by Claim~\ref{eq-zr}.

Thus, we have shown that $J$ is a dual isomorphism. Since $(\cs, \subseteq)$ is a lattice, the same holds for $(\set{[\theta^{C}]}{C \in \cs}, \le)$.

Finally, we prove Claim~\ref{eq-zr}. If $y \in B$ then $\theta^{B}_{i}(0) = y$, exactly if $\pi^{(a+1)}_{1}(i) = y+1$. Therefore, the set $\set{i}{\theta^{B}_{i}(0) = y}$ is computable. If $y \notin B$, then we have that
\[
\theta^{B}_{i}(0) = y \Leftrightarrow \pi^{(a+1)}_{1}(i) \notin B+1 \wedge \theta_{\pi^{(a+1)}_{a+1}(i)}(0) = y.
\]
By Rice's theorem the set $\set{j}{\theta_{j}(0) = y}$ is not computable, for every $y \in \omega$. As $B$ is computable it follows that the set $\set{i}{\theta^{B}_{i}(0) = y}$ cannot be computable.
\end{proof}

Goetze~\cite{goe76} has shown that every countable partially ordered set can be isomorphically embedded in the lattice of all computable set. By Lemma~\ref{lem-goe2} it therefore follows that every such set can be isomorphically embedded in the Rogers semi-lattice of computable numberings of $\SSS^{(1)}_{A}$. This concludes the proof of Theorem~\ref{thm-goe}

In the next section we will show that $\SSS^{(1)}$ is an effectively given domain with properties as considered in \cite[Theorem~3.1]{sp90}. As a consequence we obtain that the Rogers semi-lattice of computable numberings of $\SSS^{(1)}_{A}$ contains a Friedberg numbering, that is a one-to-one computable numbering. As in Khutoretski\u{\i}~\cite[Corollary~2]{kh69} one even obtains that there are infinitely many Friedberg numberings of $\SSS^{(1)}_{A}$ which are pairwise incomparable with respect to $m$-reducibiility. Pour-El~\cite{pe64} shows that the equivalence class $[\gamma]$ generated by a Friedberg numbering of $\SSS^{(1)}_{A}$ is minimal in $(\bn_{A}, \le)$.

\begin{theorem}\label{thm-nonlat}
$(\bn_{A}, \le)$ is not a lattice.
\end{theorem}

$(\bn_{A}, \le)$ also contains minimal elements that are not generated by Friedberg numberings, A numbering $\eta$ of $\SSS^{(1)}_{A}$ is called \emph{positive}, if $\set{\pair{i, j}}{\eta_{i} = \eta_{j}}$ is c.e. Ershov~\cite[p.~303]{er73} shows that the equivalence class  generated by a positive numbering is minimal in the Rogers semi-lattice of numberings of a given set.

\begin{theorem}\label{lem-minnum}
$\SSS^{(1)}_{A}$ has a positive computable numbering to which no Friedberg numbering can be reduced.
\end{theorem}
\begin{proof}
The general construction is given by Khutoretski\u{\i}~\cite[Example~1]{kh69}. Here we verify the assumptions made. Let $f \in \RRR^{(1)}$. The we have to show that the sets
\[
\set{g \in \SSS^{(1)}_{A}}{\graph(g) \subseteq \graph(f)} \quad\text{and}\quad \set{g \in \SSS^{(1)}_{A}}{\graph(g) \not\subseteq \graph(f)}
\]
can be numbered in a one-to-one way so that the universal functions of these numberings have a c.e.\ graph. This, however, is a consequence of Mal'cev~\cite[Theorem~5]{ma70}, if the classes can be enumerated in a quasi-G\"odel numbering $\theta$ of $\SSS^{(1)}_{A}$ and the sets
\[
\set{i}{\graph(\alpha^{(1)}_{i}) \subset \graph(f)} \quad\text{and}\quad \set{i}{\graph(\alpha^{(1)}_{i}) \not\subset \graph(f)}
\]
are c.e. Note that
\[
\graph(\alpha^{(1)}_{i}) \subset \graph(f) \Leftrightarrow (\forall x < \lg(i))\, \alpha^{(1)}_{i}(x) = f(x).
\]
Since $f \in \RRR^{(1)}$ it follows that both sets are even computable.  

Let $C \Def \set{i}{\graph(\alpha^{(1)}_{i}) \subset \graph(f)}$ and $g \in \RRR^{(1)}$ with $\alpha^{(1)} =  \theta \circ g$. Moreover, let $j$ be a $\theta$-index of $f$. Then
\[
\set{g \in \SSS^{(1)}_{A}}{\graph(g) \subseteq \graph(f)} = \theta[\{ j \} \cup g[C]].
\]
Thus, $\set{g \in \SSS^{(1)}_{A}}{\graph(g) \subseteq \graph(f)}$ is enumerable in $\theta$. 

Because of Condition~(QGN~I) there is some $h \in \RRR^{(1)}$ that enumerates the extended graph of the universal function of $\theta$. Then we have that
\[
\graph(\theta_{i}) \not\subseteq \graph(f) \Leftrightarrow  (\exists x) (\exists t)\, \pi^{(2)}_{1}(h(t)) = \pair{i, x} \wedge \pi^{(2)}_{2}(h(t)) > 0 \wedge \pi^{(2)}_{2}(h(t)) \neq f(x) +1.
\]
Thus, $\set{i}{\graph(\theta_{i}) \not\subseteq \graph(f)}$ is c.e. and hence $\set{g \in \SSS^{(1)}_{A}}{\graph(g) \not\subseteq \graph(f)}$ enumerable in $\theta$.
\end{proof}

\section{$\widehat{\SSS}^{(1)}_{A}$ as effectively given domain}\label{sec-dom}

In this section the connection with domain theory is investigated. As will be seen, $\widehat{\SSS}^{(1)}_{A}$ is an effectively given algebraic domain with the finite functions as its compact elements. Furthermore, the domain-theoretic computability notions coincide with those developed for $\SSS^{(1)}_{A}$. The main result says that $\widehat{\SSS}^{(1)}_{A}$  can be mapped onto every other effectively given algebraic domain $D$ by an effectively continuous operator. Via this operator each quasi-G\"odel numbering of $\SSS^{(1)}_{A}$ defines an admissible numbering of the computable elements of $D$. Moreover, every such numbering can be obtained in this way. More generally, the operator induces a homomorphism of the Rogers semi-lattice of computable numberings of $\SSS^{(1)}_{A}$ onto the Rogers semi-lattice of computable numberings of the computable elements of $D$.

Let $(D, \sqsubseteq)$ be a poset. $D$ is \emph{pointed} if it contains a least element $\bot$. A subset $L$ of $D$ is \emph{directed}, if it is non-empty and every pair of elements in $L$ has an upper bound in $L$. $D$ is a \emph{directed-complete partial order (dcpo)}, if every directed subset $L$ of $D$ has a least upper bound  $\bigsqcup L$ in $D$. 

Let $(D, \sqsubseteq)$ and $(D', \sqsubseteq')$ be posets. Then a map $\fun{G}{D}{D'}$ is \emph{Scott-continuous}, if it is monotone and for any directed subset $L$ of $D$ with existing least upper bound, $G(\bigsqcup L) = \bigsqcup' G[L]$.

Assume that $x, y$ are elements of a poset $D$. Then $x$ is said to \emph{approximate} $y$, written $x \ll y$, if for any directed subset $L$ of $D$ the least upper bound of which exists in $D$, the relation $y \sqsubseteq \bigsqcup L$ always implies the existence of some  $u \in L$ with $x \sqsubseteq u$. Moreover, $x$ is \emph{compact} if $x \ll x$. A subset $B$ of $D$ is a \emph{basis} of $D$, if for each $x \in D$ the set $B_x = \set{u\in B}{u \ll x}$ contains a directed subset with least upper bound $x$. Note that the set of all compact elements of $D$ is included in every basis of $D$. A directed-complete pointed partial order $D$ is said to be \emph{continuous} (or a \emph{domain}) if it has a basis and it is called \emph{algebraic} (or an \emph{algebraic domain}) if its compact elements form a basis. Note we here assume that a domain is always pointed which is not usually the case in the literature. Standard references for domain theory and its applications are~\cite{we87,ds90,aj94,slg94,ac98,gh03}.

\begin{lemma} \label{lem-propd}
Let $D$ and $D'$ be domains. Then the following statements hold:
\begin{enumerate}

\item\label{lem-propd-1}
The approximation relation $\ll$ is transitive.

\item\label{lem-propd-2}
$x \ll y \Rightarrow x \sqsubseteq y$.

\item\label{lem-propd-3}
$u \sqsubseteq x \ll y \sqsubseteq z \Rightarrow x \ll z$.

\item\label{lem-propd-4}
$\bot \ll x$.

\item\label{lem-propd-5}
$B_{x}$ is directed with respect to $\ll$.


\item\label{lem-propd-7}
$\fun{G}{D}{D'}$ is Scott-continuous, exactly if for all $x \in D$, $G(x) = \bigsqcup' G[B_{x}]$.

\end{enumerate}
\end{lemma}

Note that by Properties~(\ref{lem-propd-2}), (\ref{lem-propd-3}) we have that $x \ll y$ exactly if $x \sqsubseteq y$, in case that $x$ or $y$ is compact.

\begin{defin}\label{def-dombb}
Let $D$ be a continuous domain with countable basis $B$ and a numbering $\fun{\beta}{\omega}{B}$. Then $D$ is said to be \emph{effectively given} if the set $\set{\pair{i, j}}{\beta_{i} \ll \beta_{j}}$ is c.e. 
\end{defin}

If $D$ is effectively given, then an element $x$ is \emph{computable} if $\beta^{-1}[B_{x}]$ is c.e. Let $D_{c}$ be the set of all computable elements of $D$. If $D'$ is another domain, say with basis $B'$ and numbering $\beta'$ so that $D'$ is effectively given, then a map $\fun{G}{D}{D'}$ is \emph{computable} if $G$ is Scott-continuous and $\set{\pair{i, j}}{\beta'_{j} \ll G(\beta_{i})}$ is c.e. Note that for computable maps $G$, $G[D_{c}] \subseteq D'_{c}$.

\begin{defin}\label{def-adm}
Let $D$ be an effectively given domain. Then a numbering $\fun{\eta}{\omega}{D_{c}}$ of the computable elements of $D$ is called \emph{admissible} if it satisfies the following two requirements:
\begin{description}

\item[\rm(A~I)] $\set{\pair{i, j}}{\beta_{i} \ll \eta_{j}}$ is c.e.

\item[\rm(A~II)] There is a function $d \in \RRR^{(1)}$ such that for all $ i \in \omega$, if $\beta[W_{i}]$ is directed then $\eta_{d(i)} = \bigsqcup \beta[W_{i}]$.

\end{description}
\end{defin}
Weihrauch and Deil~\cite{wd80} have shown that for every effectively given domain an admissible numbering can be constructed. 

\begin{proposition}[Weihrauch, Deil, 1980]\label{pn-advol}
Let $D$ be an effectively given domain and  let $\eta$ be an admissible  and $\gamma$  an arbitrary numbering of its computable elements. Then the following statements hold:
\begin{enumerate}


\item\label{pn-advol-2}
$\eta$ is complete with special element $\bot$.

\item\label{pn-advol-3}
$\gamma \text{ satisfies (A~I)} \Leftrightarrow \gamma \le_{m} \eta$.

 \item\label{pn-advol-4}
$\gamma \text{ satisfies (A~II)} \Leftrightarrow \eta \le_{m} \gamma$.

\item\label{pn-advol-5}
$\gamma \text{ is admissible} \Leftrightarrow \eta \equiv_{m} \gamma \Leftrightarrow \eta \equiv \gamma$. 

\end{enumerate}
\end{proposition}
Since $\eta$ is complete, it ensues with a result of Ershov~\cite[p.~332]{er73} that $\eta$ is cylindrical. As follows from the definition of an effectively given domain, all basic elements are computable. Moreover, if $\gamma$ is any numbering of $D_{c}$ satisfying (A~II) then $\beta \le_{m} \gamma$.

\begin{lemma}\label{lem-effapp}
Let $D$ be an effectively given domain with basis $B$ and numbering $\beta$ of the basis elements. Moreover, let $\eta$ be a numbering of $D_{c}$ that satisfies (A~I). Then there is a function $r \in \RRR^{(2)}$ such that for all $i, j \in \omega$ with $\beta_{i} \ll \eta_{j}$ the following statements hold:
\begin{enumerate}

\item\label{lem-effapp-1}
$\varphi_{r(i, j)} \in \RRR^{(1)}$,

\item\label{lem-effapp-2}
$\beta(i) \ll \beta(\varphi_{r(i, j)}(0))$,

\item\label{lem-effapp-3}
$\beta(\varphi_{r(i, j)}(a)) \ll \beta(\varphi_{r(i, j)}(a+1))$ \quad $(a \in \omega)$,

\item\label{lem-effapp-4}
$\eta(j) = \bigsqcup_{a} \beta(\varphi_{r(i, j)}(a))$.

\end{enumerate}
\end{lemma}
\begin{proof}
Let $E_{i j} \Def \set{a}{\beta(i) \ll \beta(a) \ll \eta(j)}$. Then $E_{i j}$ is c.e. Thus, there is some function $q \in \RRR^{(2)}$ so that $E_{i j} = W_{q(i, j)}$. Since $B_{\eta(j)}$ is directed with respect to $\ll$, the same holds for $\beta[E_{i j}]$. It follows that $\bigsqcup \beta[E_{i j}]$ exists. Moreover, $\bigsqcup \beta[E_{i j}] \sqsubseteq \eta(j)$. Let $i, j$ be such that $\beta(i) \ll \eta(j)$. Since $B_{\eta(j)}$ is directed with respect to $\ll$, we have that for any $u \in B_{\eta(j)}$ there is some $u' \in \beta[E_{i j}]$ with $u \ll u'$. Thus $E_{i j}$ is non-empty and in addition, $\eta(j) \sqsubseteq \bigsqcup \beta[E_{i j}]$. Hence, $\eta(j) = \bigsqcup \beta[E_{i j}]$.

In the sequel let $s \in \RRR^{(1)}$ such that $\varphi_{s(a)}$ is a total function enumerating $W_{a}$, if $W_{a}$ is non-empty.  Moreover, let $k \in \RRR^{(2)}$ with $W_{k(a, c)} = W_{a} \cap W_{c}$, Then define $g \in \RRR^{(3)}$ by
\begin{align*}
&g(i, j, 0) \Def \varphi_{s(q(i, j))}(0), \\
&g(i, j, a+1) \Def \varphi_{s(k(q(g(i, j, a), j), q(\varphi_{s(q(i, j))}(a+1), j)))}(0).
\end{align*}
Then $g(i, j, a+1) \in E_{g(i, j ,a) j} \cap E_{\bar{\imath} j}$, where $\bar{\imath}$ is the $(a+1)$-st element of $E_{i j}$ in the enumeration $\varphi_{s(q(i, j))}$. Because $\beta[E_{i j}]$ is directed with respect to $\ll$, we have that $E_{g(i, j, a) j} \cap E_{\bar{\imath} j}$ is non-empty. Therefore, $g(i, j, a+1)$ is defined. Furthermore, for all $a$, $\beta(i) \ll \beta(g(i, j, a+1))$, from which we obtain that $\bigsqcup \beta[E_{i j}] \sqsubseteq \bigsqcup_{a} \beta(g(i, j, a))$. Conversely, since for all $a$, $g(i, j, a) \in E_{i j}$, we also have that $\bigsqcup_{a} \beta(g(i, j, a)) \sqsubseteq \bigsqcup \beta[E_{i j}]$. Thus, $\bigsqcup_{a} \beta(g(i, j, a))  = \bigsqcup \beta[E_{i j}] = \eta(j)$. Now, let $r \in \RRR^{(2)}$ with $\varphi_{r(i, j)}(a) = g(i, j, a)$. Then $r$ is as required.
\end{proof}

For the next consequence choose $i$ such that $\beta(i) = \bot$.

\begin{corollary}\label{cor-ibot}
For every $x \in D_{c}$, there is a function $p \in \RRR^{(1)}$ so that
\begin{enumerate}

\item \label{cor-ibot-1}
$\beta(p(a)) \ll \beta(p(a+1))  \quad (a \in \omega$),

\item \label{cor-ibot-2}
$x = \bigsqcup_{a} \beta(p(a))$.

\end{enumerate}
\end{corollary}

After these more technical results which we will need later, we now start investigating the relationship of the function classes $\widehat{\SSS}^{(n)}_{A}$ with domains. Again we will restrict ourselves
to considering only the classes $\widehat{\SSS}^{(1)}_{A}$.

\begin{theorem}
Let $A \subseteq \omega$ be a c.e.\ infinite set and for $f, g \in \widehat{\SSS}^{(1)}_{A}$ set 
\[
f \sqsubseteq g \Leftrightarrow \graph(f) \subseteq \graph(g).
\]
Then $(\widehat{\SSS}^{(1)}_{A}, \sqsubseteq)$ is an effectively given algebraic domain such that:
\begin{enumerate} 

\item \label{thm-fctdom-1}
The nowhere defined function is the least element.

\item \label{thm-fctdom-2}
The initial segment functions in $\anf^{(1)}_{A}$ are exactly the compact elements.

\item \label{thm-fctdom-3}
The functions in $\SSS^{(1)}_{A}$ are the computable elements.

\item \label{thm-fctdom-4}
For numberings of $\SSS^{(1)}_{A}$,  Conditions~(QGN~I) and (A~I) as well as (QGN~II) and (A~II) are equivalent. In particular,  the quasi-G\"odel numberings are exactly the admissible numberings.

\item \label{thm-fctdom-5}
The computability notions for operators $\fun{G}{\widehat{\SSS}^{(1)}_{A}}{\widehat{\SSS}^{(1)}_{A}}$ in Definition~\ref{def-compop} and in this section coincide.

\end{enumerate}
\end{theorem}
\begin{proof}
If $L \subseteq \widehat{\SSS}^{(1)}_{A}$ is directed with respect to $\sqsubseteq$, then $L$ has to be a chain. Thus the union of the graphs of the functions in $L$ is again the graph of a function. This function must be in $\widehat{\SSS}^{(1)}_{A}$. It is the least upper bound of $L$ with respect to  $\sqsubseteq$. Thus, $\widehat{\SSS}^{(1)}_{A}$ is directed-complete. Obviously the nowhere defined function is the least element. Every function in $\widehat{\SSS}^{(1)}_{A}$ is the least upper bound of its restrictions to initial segments of $\omega$ with length in $A$. So, the functions in $\anf^{(1)}_{A}$ form a basis. 

(\ref{thm-fctdom-2}) All functions in  $\anf^{(1)}_{A}$ are compact. To see this let $p \in \anf^{(1)}_{A}$ and $L$ be a directed subset of $\widehat{\SSS}^{(1)}_{A}$ with $p \sqsubseteq \bigsqcup L$. Then $\graph(p) \subseteq \bigcup \set{\graph(q)}{q \in L}$. Since $\graph(p)$ is finite, it is covered by the union of the graphs of finitely many functions in $L$. Because this union is again contained in the graph of some function $r \in L$, as $L$ is directed, we have that $p \sqsubseteq r$. 
Conversely, if $f \in \widehat{\SSS}^{(1)}_{A}$ is compact, then consider the directed set $L$ of all functions $s \in \anf^{(1)}_{A}$ with $s \sqsubseteq f$. Then $f \sqsubseteq \bigsqcup L$. By compactness there is some $s \in L$ with $f \sqsubseteq s$. Since $s \in \anf^{(1)}_{A}$, the same must hold for $f$. Thus, the domain $\widehat{\SSS}^{(1)}_{A}$ is algebraic.

With Lemma~\ref{lem-alph}(\ref{lem-alph-4}) and Statement~(\ref{thm-fctdom-4}) we obtain that $\set{\pair{i, j}}{\alpha^{(1)}_{i} \sqsubseteq \alpha^{(1)}_{j}}$ is c.e. Thus, the domain $\widehat{\SSS}^{(1)}_{A}$ is effectively given with respect to the numbering $\alpha^{(1)}$ of the basis.

(\ref{thm-fctdom-3}) Let $f \in \SSS^{(1)}_{A}$. Then it follows with Lemma~\ref{lem-alph}(\ref{lem-alph-5}) that $\set{i}{\alpha^{(1)}_{i} \sqsubseteq f}$ is c.e. Hence, $f \in (\widehat{\SSS}^{(1)}_{A})_{c}$. Conversely, if $f \in (\widehat{\SSS}^{(1)}_{A})_{c}$ then $\set{i}{\alpha^{(1)}_{i} \sqsubseteq f}$ is c.e. Since $\graph(f) = \bigcup\{\, \graph(s) \mid s \in \anf^{(1)}_{A} \wedge s \sqsubseteq f \,\}$ we have,
\[
\pair{x, y} \in \graph(f) \Leftrightarrow (\exists i)\, \alpha^{(1)}_{i} \sqsubseteq f \wedge \lg(i) > x \wedge \pair{\pair{i, x}, y+1} \in \egraph(\lambda a, z.\ \alpha^{(1)}_{a}(z)).
\]
Because the extended graph of the universal function of $\alpha^{(1)}$ is computable, by  Lemma~\ref{lem-alph}(\ref{lem-alph-2}), it follows that $\graph(f)$ is c.e. With Lemma~\ref{pn-compce}(\ref{pn-compce-2}) we therefore have that $f \in \SSS^{(1)}_{A}$.

(\ref{thm-fctdom-4}) As in the proof of Lemma~\ref{lem-alph}(\ref{lem-alph-5}) it is only required that $\theta$ satisfies Condition~(QGN~I), it follows that every numbering of $\SSS^{(1)}_{A}$ with Property~(QGN~I) satisfies Condition~(A~I). Conversely, assume that $\theta$ is a numbering of $\SSS^{(1)}_{A}$ satisfying Condition~(A~I). Then it follows as in the proof of Statement~(\ref{thm-fctdom-3}) that $\egraph(\lambda j, x.\  \theta_{j}(x))$ is c.e. Thus, $\theta$ meets Requirement~(QGN~I).

Next assume that $\theta$ has Property~(QGN~II) and let $\alpha^{(1)}[W_{i}]$ be directed. Then there exists some $r \in \widehat{\SSS}^{(1)}_{A}$ with $r = \bigsqcup \alpha^{(1)}[W_{i}]$. Hence,
\[
\graph(r) = \set{\pair{x, y}}{(\exists j)\, j \in W_{i} \wedge \pair{\pair{j, x}, y+1} \in \egraph(\lambda a, z.\ \alpha^{(1)}_{a}(z))}, 
\]
which implies that $r \in \SSS^{(1)}_{A}$. As we moreover see, $\egraph(r)$ can be uniformly enumerated in $i$. Therefore, by Condition~(QGN~II), there is a function $v \in \RRR^{(1)}$ so that $\theta_{v(i)} = r = \bigsqcup \alpha^{(1)}[W_{i}]$. Thus, $\theta$ fulfils Requirement~(A~II). Since, for any $k \in \RRR^{(2)}$, the set $\set{j}{\egraph(\alpha^{(1)}_{j}) \subseteq \range(\lambda t.\ k(i, t))}$ is c.e., uniformly in $i$, it conversely follows that $\theta$ has Property~(QGN~II), once it satisfies Requirement~(A~II).

(\ref{thm-fctdom-5}) Let $C \subseteq \omega$ be a c.e.\ set that defines $G$. Then
\begin{align*}
\graph(\alpha^{(1)}_{i}) \subseteq \mbox{} &\graph(G(\alpha^{(1)}_{j}))  \\ \Leftrightarrow \mbox{} 
 &\hspace{-3em}(\forall x < \lg(i)) (\exists a)\, \pair{\pair{\pair{x}, \alpha^{(1)}_{i}(x)}, a} \in C \wedge (\forall c \le \lth(a))\, (a)_{c} \in \graph(\alpha^{(1)}_{j}) \\ \Leftrightarrow \mbox{} 
 &\hspace{-3em}(\forall x < \lg(i)) (\exists y) (\exists a)\, \pair{\pair{i, x}, y+1} \in \egraph(\lambda b, z.\ \alpha^{(1)}_{b}(z)) \wedge \pair{\pair{\pair{x}, y}, a} \in C \wedge \mbox{} \\
&\hspace{-3em}(\forall c \le \lth(a))\, \pair{\pair{j, \pi^{(2)}_{1}((a)_{c})}, \pi^{(2)}_{2}((a)_{c})+1} \in \egraph(\lambda b, z.\ \alpha^{(1)}_{b}(z)).
\end{align*}
By Lemma~\ref{lem-alph}(\ref{lem-alph-2}) $\egraph(\lambda a, z.\ \alpha^{(1)}_{a}(z))$ is c.e. With the Tarski-Kuratowski algorithm~\cite{ro67} we can now bring this expression in a $\Sigma_{1}$-form from which we see that $\set{\pair{i, j}}{\graph(\alpha^{(1)}_{i}) \subseteq \graph(G(\alpha^{(1)}_{j}))}$ is c.e.

Conversely assume that $V \Def \set{\pair{i, j}}{\graph(\alpha^{(1)}_{i}) \subseteq \graph(G(\alpha^{(1)}_{j}))}$ is c.e. Moreover, Let $r \in \widehat{\SSS}^{(1)}_{A}$.  By the continuity of $G$  have that
\begin{align*}
\pair{\pair{x}, z} \mbox{} & \in \graph(G(r)) \\
\Leftrightarrow \mbox{}  &(\exists i) (\exists a)\, \alpha^{(1)}_{i} \sqsubseteq G(\alpha^{(1)}_{a}) \wedge \pair{x, z} \in \graph(\alpha^{(1)}_{i}) \wedge \alpha^{(1)}_{a} \sqsubseteq r \\
\Leftrightarrow \mbox{}  &(\exists \bar{a}) (\exists a) (\exists i) \pair{\pair{i, x}, z+1} \in \egraph(\lambda b, u.\ \alpha^{(1)}_{b}(u)) \wedge \pair{i, a} \in V \wedge \lth(\bar{a}) = \lg(a) \wedge \mbox{} \\
& (\forall c \le \lg(a))\, \pi^{(2)}_{1}((\bar{a})_{c}) = c \wedge \pair{\pair{a, c}, \pi^{(2)}_{2}((\bar{a})_{c})+1} \in  \egraph(\lambda b, u.\ \alpha^{(1)}_{b}(u)) \wedge \mbox{} \\
& (\bar{a})_{c} \in \graph(r).
\end{align*}
Therefore by setting
\begin{multline*}
C \Def \set{\pair{\pair{\pair{x}, z}, \bar{a}}}{(\exists a) (\exists i)\, \pair{\pair{i, x}, z} \in  \egraph(\lambda b, u.\ \alpha^{(1)}_{b}(u)) \wedge \pair{i, a} \in V \wedge \mbox{}\\
 \lth(\bar{a}) = \lg(a) \wedge 
\pi^{(2)}_{1}((\bar{a})_{c}) = c \wedge \pair{\pair{a, c}, \pi^{(2)}_{2}((\bar{a})_{c})+1} \in  \egraph(\lambda b, u.\ \alpha^{(1)}_{b}(u))},
\end{multline*}
we obtain that $C$ is c.e. and defines $G$.
\end{proof}

The next result shows that each of the  algebraic domains $\widehat{\SSS}^{(1)}_{A}$ can be computably mapped onto any other effectively given domain. For the proof we need an extension of a result by Weihrauch and Sch\"afer~\cite{ws83}.

\begin{proposition}\label{pn-wstech}
Let $D$ be an effectively given domain with basis $B$ and numbering $\beta$ of the base elements. Then there is a computable operator $\fun{G}{\widehat{\SSS}^{(1)}_{\omega}}{\widehat{\SSS}^{(1)}_{\omega}}$ such that the following statements hold for $f, g \in \widehat{\SSS}^{(1)}_{\omega}$,

\begin{enumerate}

\item\label{pn-wstech-1}
If $f \in \RRR^{(1)}$ then also $G(f) \in \RRR^{(1)}$.

\item\label{pn-wstech-2}
If $f \in \anf^{(1)}_{\omega}$ then also $G(f) \in \anf^{(1)}_{\omega}$.

\item\label{pn-wstech-3}
For all $a \in \dom(G(f))$, $\beta(G(f)(a)) \ll \beta(G(f)(a+1))$.

\item\label{pn-wstech-4}
If $\beta[\range(f)]$ is directed then $\bigsqcup_{a}\beta(G(f)(a)) \sqsubseteq \bigsqcup \beta[\range(f)]$.

\item\label{pn-wstech-5}
If $f \in \RRR^{(1)}$ and $\beta[\range(f)]$ is directed then  $\bigsqcup_{a}\beta(G(f)(a)) = \bigsqcup \beta[\range(f)]$.

\item\label{pn-wstech-6}
If $f \sqsubseteq g$ then $\bigsqcup_{a} \beta(G(f)(a)) \sqsubseteq  \bigsqcup_{a} \beta(G(g)(a))$.

\item\label{pn-wstech-7}
For $f \in \anf^{(1)}_{\omega}$, $\bigsqcup_{a} \beta(G(f)(a)) \in B$.

\end{enumerate}
\end{proposition}
\begin{proof}
Since $\set{\pair{i, j}}{\beta_{i} \ll \beta_{j}}$ is c.e., there are functions $k \in \RRR^{(2)}$ and $h \in \RRR^{(1)}$ so that $\lambda t.\ k(t, j)$ enumerates $\set{i}{\beta_{i} \ll \beta_{j}}$ and $h$ enumerates $\set{\pair{i, j}}{\beta_{i} \ll \beta_{j}}$. Now, we define functions $q, p, r \in \RRR^{(1)}$:

If $\lth(a) \le 1$, we set
\[
p(a) \Def q(a) \Def i_{\bot} \quad\text{and}\quad r(a) \Def 1,
\]
where, $i_{\bot}$ is a $\beta$-index of $\bot$.

In case that $a = \val{a_{0}, \ldots, a_{m+1}}$ and $q(a')$ as well as $p(a')$ and $r(a')$ are already defined for $a' = \val{a_{0}, \ldots, a_{m}}$, say $r(a') = \pair{t, c}$, then we define $q(a)$, $p(a)$ and $r(a)$ as follows: if there is $\pair{i, j} \le m+1$ with
\[
\{ \pair{k(c, a_{t}), k(j, a_{i})}, \pair{p(a'), k(j, a_{i})} \} \subset \{ h(0), \ldots, h(m+1) \},
\]
then set
\begin{align*}
&q(a) \Def p(a'), \\
&\text{$p(a) \Def k(j, a_{i})$, for the smallest number $\pair{i, j}$ with this property,} \\
&r(a) \Def r(a') + 1.
\end{align*}

If there is no such number $\pair{i, j}$, then set
\begin{align*}
&q(a) \Def \pi^{(2)}_{1}(\mu \pair{b, a}.\ \{ \pair{q(a'), b}, \pair{b, p(a')} \} \subseteq \{ h(0), \ldots, h(n) \}), \\
&p(a) \Def p(a'), \\
&r(a) \Def r(a').
\end{align*}

As follows form the definition, $\beta_{q(a')} \ll \beta_{p(a')}$. By Lemma~\ref{lem-propd}(\ref{lem-propd-5}) there is thus a number $\pair{b, n}$ with the above property. Hence, $q(a)$ is defined also in this case.

In addition to the above let $v \in \RRR^{(1)}$ with
\begin{align*}
&v(\val{\text{`empty sequence'}}) \Def \val{\text{`empty sequence'}},\\
&v(\val{a_{0}, \ldots, a_{m}}) \Def \val{\pair{0, a_{0}}, \ldots, \pair{m, a_{m}}}.
\end{align*}
Moreover, let
\[
C \Def \set{\pair{\pair{\lth(a), q(a)}, v(a)}}{a \in \omega}.
\]
Then $C$ is c.e. As is readily seen, $C$ defines a computable operator $G$ on $\widehat{\SSS}^{(1)}_{\omega}$. For $f \in \widehat{\SSS}^{(1)}_{\omega}$ with non-empty domain it follows that $\dom(f) = \dom(G(f))$. $G$ maps the nowhere defined function onto the function that for $0$ has value $i_{\bot}$, and is undefined, otherwise. This implies  Statements~(\ref{pn-wstech-1}) and (\ref{pn-wstech-2}).

Let $\overline{f}(m) \Def \val{f(0), \ldots, f(m)}$. Then we obtain  from the definition of $G$ that for $a' = \overline{f}(m)$ and $a = \overline{f}(m+1)$, $G(f)(m) = q(a')$ and $G(f)(m+1) = q(a)$. Since $\beta_{q(a')} \ll \beta_{q(a)}$, this proves Statement~(\ref{pn-wstech-3}). In particular, we have that $\bigsqcup_{a} \beta(G(f)(a))$ exists. Moreover, Statements~(\ref{pn-wstech-6}) and (\ref{pn-wstech-7}) follow, because $G$ is monotone by Theorem~\ref{thm-coprop} and by Statement~(\ref{pn-wstech-2}), $\dom(G(f))$ is finite, for functions $f \in \anf^{(1)}_{\omega}$.

For Statements~(\ref{pn-wstech-4}) and (\ref{pn-wstech-5}), finally, let $f \in \RRR^{(1)}$ so that $\beta[\range(f)]$ is directed. First, we show that the sequence $(r(\overline{f}(m)))_{m \in \omega}$ is unbounded. Assume that $r(\overline{f}(m))= \pair{t, c}$. Since the set $\set{u \in B}{u \ll \bigsqcup \beta[\range(f)]}$ is directed with respect to $\ll$ by Lemma~\ref{lem-propd}(\ref{lem-propd-5}), there is some $m' \ge m$ and a number $\pair{i, j} \le m+1$ so that
\begin{equation}\label{eq-incl}
\{ \pair{k(c, f(t)), k(j, f(i))}, \pair{p(\overline{f}(m)), k(j, f(i))} \} \subseteq \{ h(0), \ldots, h(m'+1) \}.
\end{equation}
Then $r(\overline{f}(m'+1)) > r(\overline{f}(m))$.

Suppose that $u \in B$ with $u \ll \beta(f(t))$. The there is some $c$ with $\beta(k(c, f(t)) = u$. Moreover, because the sequence $(r(\overline{f}(m)))_{m \in \omega}$ is unbounded, there is some $m'$ with $r(\overline{f}(m')) = \pair{t, c}$ as well as a smallest number $\pair{i, j} \le m'+1$ so that the inclusion in (\ref{eq-incl}) holds. Thus, $u \ll \beta(k(j , f(i)))$. Consequently, since $k(j, f(i)) = G(f)(\widehat{m})$, for some $\widehat{m} \ge m'+1$,  it follows that $\bigsqcup \beta[\range(f)] \sqsubseteq \bigsqcup_{m} \beta(G(f)(m))$.

As follows from the above definitions, we have for $g \in \widehat{\SSS}^{(1)}_{\omega}$ that for every $m$ there is some $i$ with $\beta(G(g)(m)) \ll \beta(g(i))$. Therefore, $\bigsqcup_{m} \beta(G(g)(m)) \sqsubseteq \bigsqcup \beta[\range(g)]$, if $\beta[\range(g)]$ is directed. This shows Statement~(\ref{pn-wstech-4}) and, with what we derived before, also Statement~(\ref{pn-wstech-5}). 
\end{proof}

\begin{theorem}\label{thm-surj}
Let $A \subseteq \omega$ be a c.e.\ infinite set and $D$ be an effectively given domain with basis $B$ and numbering $\beta$ of the base elements. Then there is a computable onto map $\fun{\Gamma}{\widehat{\SSS}^{(1)}_{A}}{D}$ so that for $f, g \in \widehat{\SSS}^{(1)}_{A}$,
\begin{enumerate}

\item \label{thm-surj-1}
If $f \in \anf^{(1)}_{A}$ then $\Gamma(f) \in B$.

\item \label{thm-surj-2}
If $\beta[\range(f)]$ is directed then $\Gamma(f) \sqsubseteq \bigsqcup \beta[\range(f)]$.

\item \label{thm-surj-3}
If $f \in \RRR^{(1)}$ and $\beta[\range(f)]$ is directed then $\Gamma(f) = \bigsqcup \beta[\range(f)]$.

\end{enumerate}
\end{theorem}
\begin{proof}
Let $\fun{G}{\widehat{\SSS}^{(1)}_{\omega}}{\widehat{\SSS}^{(1)}_{\omega}}$ be the computable operator constructed in the preceding Lemma. Then we define for $f \in \widehat{\SSS}^{(1)}_{A}$,
\[
\Gamma(f) \Def \bigsqcup_{a} \beta(G(f)(a)).
\]
Since the basis $B$ is countable, every domain element is the least upper bound of a sequence that is monotonically increasing with respect to $\ll$. Because of Property~(\ref{pn-wstech-5}) of the preceding lemma it therefore follows that $\Gamma$ is onto. We will show that $\Gamma$ is continuous. Let to this end $u \in B$. Since $G$ is continuous by Theorem~\ref{thm-coprop}, we have
\begin{align*}
u \ll \Gamma(f) 
&\Leftrightarrow u \ll \bigsqcup\nolimits_{a} \beta(G(f)(a))  \\
&\Leftrightarrow  (\exists a)\, u \ll \beta(G(f)(a)) \\
&\Leftrightarrow (\exists a)\, u \ll \beta((\bigsqcup \set{G(s)}{s \in \anf^{(1)}_{A} \wedge s \sqsubseteq f})(a)) \\
&\Leftrightarrow (\exists a) (\exists s \in \anf^{(1)}_{A})\, s \sqsubseteq f \wedge a \in \dom(G(s)) \wedge u \ll \beta(G(s)(a)).
\end{align*}

In the same way we obtain that
\begin{multline*}
u \ll \bigsqcup \set{\Gamma(s)}{s \in \anf^{(1)}_{A} \wedge s \sqsubseteq f}  \Leftrightarrow \\
 (\exists a) (\exists s \in \anf^{(1)}_{A})\, s \sqsubseteq f \wedge a \in \dom(G(s)) \wedge u \ll \beta(G(s)(a)).
\end{multline*}

Since every domain element $x$ is uniquely determined by the set $B_{x}$, it follows that $\Gamma$ is continuous. It remains to show that $\set{\pair{i, j}}{\beta_{j} \ll \Gamma(\alpha^{(1)}_{i})}$ is c.e. We have that
\begin{align*}
\beta_{j}& \ll \Gamma(\alpha^{(1)}_{i}) \\
&\Leftrightarrow (\exists a)\, \beta_{j} \ll \beta(G(\alpha^{(1)}_{i})(a)) \\
&\Leftrightarrow (\exists a) (\exists b)\, \beta_{j} \ll \beta_{b} \wedge b = G(\alpha^{(1)}_{i})(a) \\
&\Leftrightarrow (\exists a) (\exists b) (\exists c) (\exists m)\, \beta_{j} \ll \beta_{b} \wedge b = \alpha^{(1)}_{m}(c) \wedge \graph(\alpha^{(1)}_{m}) \subseteq \graph(G(\alpha^{(1)}_{i})) \\
&\Leftrightarrow   (\exists a) (\exists b) (\exists c) (\exists m)\, \beta_{j} \ll \beta_{b} \wedge \pair{\pair{m, c}, b+1} \in \egraph(\lambda \hat{\imath}, n.\ \alpha^{(1)}_{\hat{\imath}}(n)) \wedge  \mbox{}\\
& \hspace{1.8em} \graph(\alpha^{(1)}_{m}) \subseteq \graph(G(\alpha^{(1)}_{i})).
\end{align*}
Since the set of all $\pair{i, m}$ with $\graph(\alpha^{(1)}_{m}) \subseteq \graph(G(\alpha^{(1)}_{i}))$ is c.e.\ by Theorem~\ref{thm-coprop} and $\egraph(\lambda \hat{\imath}, n.\ \alpha^{(1)}_{\hat{\imath}}(n))$ is c.e.\ by Lemma~\ref{lem-alph}, the set of all $\pair{i, j}$ with $\beta_{j} \ll \Gamma(\alpha^{(1)}_{i})$ is c.e. as well. Therefore, $\Gamma$ is computable. Properties~(\ref{thm-surj-1})-(\ref{thm-surj-3}) are a consequence of the corresponding properties of $G$,
\end{proof}

Assume that the domain $D$ in the previous result is also algebraic. Then  we have for $u \in B$ and $x \in D$ that $u \ll x$ exactly if $u \sqsubseteq x$. As is readily seen, in this case for every base element $u$ of the domain one can find an initial segment function in $\anf^{(1)}_{A}$ that is mapped onto $u$ under $\Gamma$. This initial segment map can be chosen in such a way that it has only $\beta$-indices of $\bot$ and $u$ as values. This makes the idea of extending representations as they are used in computable analysis e.g. (cf.\ \cite{we00}) to maps from $\widehat{\SSS}^{(1)}_{A}$ to domains precise that was mentioned at the beginning of this paper: initial segment functions are used as names for approximating base elements, or the Scott-open sets they determine, and total functions are names for the limit elements they approximate. By the monotoniciy of $\Gamma$ the relation of approximation between the functions in $\widehat{\SSS}^{(1)}_{A}$, that is the names, is transferred to the domain.

Next, we will show that via the mapping $\Gamma$ numberings of $D_{c}$ can be generated from numberings of $\SSS^{(1)}_{A}$.

\begin{theorem} \label{thm-numgen}
Let $A \subseteq \omega$ be a c.e.\ infinite set and $D$ be an effectively given domain. Moreover, let $\theta$ be a numbering of $\SSS^{(1)}_{A}$. Define
\[
\eta_{i} \Def \Gamma(\theta_{i}) \quad (i \in \omega).
\]
Then $\eta$ is a numbering of the computable elements of $D$ such that:
\begin{enumerate}

\item  \label{thm-numgen-1}
If $\theta$ satisfies (QGN~I) then $\eta$ meets Condition~(A~I).

\item  \label{thm-numgen-2}
If $\theta$ satisfies (QGN~II) then $\eta$ meets Condition~(A~II).

\item  \label{thm-numgen-3}
If $\theta$ is quasi-G\"odel numbering then $\eta$ is admissible.

\end{enumerate}
\end{theorem}
\begin{proof}
By Corollary~\ref{cor-ibot}, $\Gamma$ maps $\RRR^{(1)}$ onto $D_{c}$. Therefore, the mapping $\eta$ defined above is a numbering of $D_{c}$. 

(\ref{thm-numgen-1})  Because of the continuity of $\Gamma$ we have that
\begin{align*}
\beta_{i} \ll \eta_{j}
&\Leftrightarrow \beta_{i} \ll \Gamma(\theta_{j}) \\
&\Leftrightarrow \beta_{i} \ll \Gamma(\bigsqcup \set{\alpha^{(1)}_{n}}{\alpha^{(1)}_{n} \sqsubseteq \theta_{j}}) \\
&\Leftrightarrow \beta_{i} \ll \bigsqcup \set{\Gamma(\alpha^{(1)}_{n})}{\alpha^{(1)}_{n} \sqsubseteq \theta_{j}}) \\
&\Leftrightarrow (\exists n)\, \beta_{i} \ll \Gamma(\alpha^{(1)}_{n}) \wedge \alpha^{(1)}_{n} \sqsubseteq \theta_{j}. \\
\end{align*}
Since $\Gamma$ is computable, the set $\set{\pair{i, j}}{\beta_{i} \ll \Gamma(\alpha^{(1)}_{n})}$ is c.e. By Lemma~\ref{lem-alph} the same is true for $\set{\pair{n, j}}{\alpha^{(1)}_{n} \sqsubseteq \theta_{j}}$, if $\theta$ satisfies (QGN~I). Thus, also $\set{\pair{i, j}}{\beta_{i} \ll \eta_{j}}$ is c.e.\ in this case. That is, $\eta$ meets the (A~I) requirement.

(\ref{thm-numgen-3})
Let $\theta$ be a quasi-G\"odel numbering. We need to construct a function $v \in \RRR^{(1)}$ so that $\theta_{v(i)}$ enumerates $W^{A}_{i}$ in such a way that $\theta_{v(i)} \in \RRR^{(1)}$, if $W^{A}_{i}$ is not empty. By (QGN~I) there is a function $h \in \RRR^{(1)}$ enumerating the universal function of $\theta$. Define
\begin{gather*}
\widehat{Q}(i, a, z) \Leftrightarrow \pi^{(2)}_{1}(h(a)) = \pair{i, \pi^{(2)}_{1}((z)_{\lth(z)})} \wedge \pi^{(2)}_{2}(h(a)) > 0, \\
Q(i, t, z) \Leftrightarrow (\exists a \le t)\, \widehat{Q}(i, a, z), \\
f(i, t, z) \Def \pi^{(2)}_{2}(h(\mu a \le t.\ \widehat{Q}(i, a, z))).
\end{gather*}
Now, construct $k \in \RRR^{(2)}$ according to Lemma~\ref{lem-meth1} so that $\lambda t.\ k(i, t)$ enumerates the extended graph of a function $g \in \RRR^{(1)}$ that enumerates $W^{A}_{i}$, if $W^{A}_{i}$ is not empty, and the extended graph of the nowhere defined function, otherwise. By Condition~(QGN~II) there is then a function $v \in \RRR^{(1)}$ that has the properties we were looking for.

Assume that $\beta[W_{i}]$ is directed. Then $W_{i}$ is not empty. As we have seen, $W$ and $W^{A}$ are computably isomorphic. Thus, there is some $s \in \RRR^{(1)}$ so that $W_{i} = W^{A}_{s(i)}$. Set $d \Def v \circ s$. Then it follows wth Theorem~\ref{thm-surj}(\ref{thm-surj-3}) that
\[
\eta_{d(i)} = \Gamma(\theta_{d(i)}) = \bigsqcup \beta[\range(\theta_{d(i)})] = \bigsqcup \beta[W^{A}_{s(i)}] = \bigsqcup \beta[W_{i}].
\]
Hence, $\eta$ satisfies Requirement~(A~II). By Statement~(\ref{thm-numgen-1}), $\eta$ also satisfies~(A~I). Therefore, $\eta$ is admissible.

(\ref{thm-numgen-2}) Let $\psi$ be a quasi-G\"odel numbering of $\SSS^{(1)}_{A}$, and let $\theta$ satisfy (QGN~II). By Theorem~\ref{thm-maxmin} there is some $g \in \RRR^{(1)}$ with $\psi = \theta \circ g$. Let $\gamma$ be the numbering of $D_{c}$ defined by $\gamma_{i} \Def \Gamma(\psi_{i})$. Then we have that $\gamma \le_{m} \eta$. Since $\gamma$ is admissible by what was shown in the previous step, it follows with Proposition~\ref{pn-advol} that $\eta$ satisfies (A~II).
\end{proof}

As we see, via the mapping $\Gamma$ numberings of $D_{c}$ with Properties~(A~I) and{\textbackslash}or~(A~II) can be obtained from numberings of $\SSS^{(1)}_{A}$. A natural question now is whether all numberings of $D_{c}$ with these properties can be obtained in this way. 

Set
\[
\widehat{\Gamma}(\theta)_{i} \Def \Gamma(\theta_{i}).
\]
Then $\widehat{\Gamma}$ is a mapping from the set of all numberings of $\SSS^{(1)}_{A}$ into the set of all numberings of $D_{c}$. As follows from the definition, $\widehat{\Gamma}$ is monotone with respect to $\le_{m}$. 

\begin{theorem}\label{thm-suff}
Let $A \subseteq \omega$ be a c.e.\ infinite set and $D$ be an effectively given domain. Then the following statements hold:
\begin{enumerate}

\item \label{thm-suff-1}
For every admissible numbering $\eta$ of $D_{c}$ there is a quasi-G\"odel numbering $\theta$ of $\SSS^{(1)}_{A}$ with $\eta = \widehat{\Gamma}(\theta)$.

\item \label{thm-suff-2}
For every numbering $\eta$ of $D_{c}$ satisfying (A~I) there is a numbering $\theta$ of $\SSS^{(1)}_{A}$ satisfying (QGN~I) with $\eta = \widehat{\Gamma}(\theta)$.

\end{enumerate}
\end{theorem}
\begin{proof}
(\ref{thm-suff-1})
Let $\psi$ be a quasi-G\"odel numbering of $\SSS^{(1)}_{A}$ and $\eta$ an admissible numbering of $D_{c}$. By the previous theorem, $\widehat{\Gamma}(\psi)$ is an admissible numbering as well. Hence, $\eta$ and $\widehat{\Gamma}(\psi)$ are computably isomorphic, by Proposition~\ref{pn-advol}. Thus, there is a computable permutation $g \in \RRR^{(1)}$ with $\eta = \widehat{\Gamma}(\psi) \circ g$, that is, $\eta_{i} = \Gamma(\psi_{g(i)})$. Let $\theta \Def \psi \circ g$. Then also $\theta$ is a quasi-G\"odel numbering. Moreover, $\eta =\widehat{\Gamma}(\theta)$.

(\ref{thm-suff-2}) As has already been mentioned, $D_{c}$ has an admissible numbering, say $\gamma$. Let $\psi$ be the quasi-G\"odel numbering with $\gamma = \widehat{\Gamma}(\psi)$, existing by Statement~(\ref{thm-suff-1}).  If $\eta$ is a numbering of $D_{c}$ satisfying~(A~I), then $\eta \le_{m} \gamma$, by Proposition~\ref{pn-advol}. Therefore, there is some $g \in \RRR^{(1)}$ with $\eta = \gamma \circ g$. Define $\theta \Def \psi \circ g$. Then it follows with Proposition~\ref{pn-qgnchar} that $\theta$ satisfies Condition~(QGN~I). Moreover, $\widehat{\Gamma}(\theta) = \eta$.
\end{proof}

The question whether an analogue of Statement~(\ref{thm-suff-2}) is true for numberings of $D_{c}$ that satisfy Requirement~(A~II) remains open.

We call numberings of $D_{c}$ that satisfy (A~I) \emph{computable}. As in the case of the numberings of $\SSS^{(1)}_{A}$ the $m$-equivalence classes of the numberings of $D_{c}$ form a Rogers semi-lattice in which the $m$-equivalence classes of the computable numberings form a Rogers sub-semi-lattice with the class of admissible numberings as greatest element. Let again $\eta \oplus \gamma$ be the direct sum of the numberings $\eta$ and $\gamma$ of $D_{c}$. Then it follows for numberings $\psi$ and $\theta$ of $\SSS^{(1)}_{A}$ that
\[
\widehat{\Gamma}(\psi \oplus \theta) = \widehat{\Gamma}(\psi) \oplus \widehat{\Gamma}(\theta).
\]
Moreover, let $\widehat{\Gamma}_{\equiv}$ be the quotient map, that is, $\widehat{\Gamma}_{\equiv}([\psi]) = [\widehat{\Gamma}(\psi)]$.

\begin{corollary}
$\widehat{\Gamma}_{\equiv}$ is a homomorphism of the Rogers semi-lattice of the numberings of $\SSS^{(1)}_{A}$ to the Rogers semi-lattice of the numberings of $D_{c}$ which maps the Rogers sub-semi-lattice of computable numberings of $\SSS^{(1)}_{A}$ onto the Rogers sub-semi-lattice of computable numberings of $D_{c}$.
\end{corollary}.

\section{Conclusion} \label{sec-concl}

Non-termination is a typical phenomenon of algorithms. It cannot be read of the program text whether and in which case it will happen. Approaches to avoid it have been studied and require advanced methods.  The question we dealt with in this paper was whether there is a class of algorithms that compute the total (computable) functions---in which one is actually only interested in---, and if non-termination occurs then the area of such inputs has a well-defined structure.

We presented such a class. The typical algorithms we had in mind when starting this research was list searching or the computation of approximations of infinite objects like the real numbers. The algorithms in this class are such that their domain of definition is either the set of all natural numbers, or a finite initial segment of this set of a length in a given set of possible lengths. It is shown that even though besides the total functions there are now only finite partial functions, a rich computability theory can be developed in which the same important results hold as in the classical theory dealing with all possible algorithms. In particular, the theory of computably enumerable sets remains unchanged, except that the domain characterisation of these set is no longer useful. What is presented is a development of computability theory based on the notion of enumeration.
The main ingredient in the new approach is the notion of a quasi-G\"odel numbering which takes on the role of G\"odel numberings in the classical approach. Every G\"odel numbering is also a quasi-G\"odel numbering.

Besides developing computability theory on the basis of quasi-G\"odel numberings, meta-investigations have been carried out: each of the new quasi-G\"odel numbered function classes is a retract of the G\"odel numbered class of all partial computable functions. Moreover, it is the class of computable elements of an effectively given algebraic domain (in the sense of Scott-Ershov) that can be computably mapped onto every other effectively given domain so that the finite functions in the function classes considered here are mapped onto base elements. This extends the use of representations in computable analysis in such a way that now also the basic open sets used for approximations obtain a finite function as name. Via the mapping every quasi-G\"odel numbering induces an admissible numbering of the computable domain elements. It was shown that every admissible numbering can even be obtained in this way.

The class of all partial computable functions, and several of its subclasses, can be characterised in a machine-independent way as smallest classes containing certain basic functions and being closed under operations like composition and primitive recursion. It would be interesting to know whether the function classes considered in this paper can be characterised in a similar way.  Among others this could be used to introduce a notion of relativised computability, in particular, a notion of Turing reducibility,  and compare it with the classical ones. Another way of doing so, is via the modified Turing machine model presented in Section~\ref{sec-syn}.


\begin{thebibliography}{00}

\bibitem{ab80}
O.~Aberth. \emph{Computable Analysis}. McGraw-Hill, New York, 1980.

\bibitem{aj94}
S.~Abramsky and A.~Jung. Domain theory. In S.~Abramsky et al.,  editors, \emph{Handbook of Logic in Computer Science, vol. 3, Semantic Structures}, Clarendon Press, Oxford, 1994, pages1--168.

\bibitem{ac98}
R.~M.~Amadio and P.-L.~Curien. \emph{Domains and Lambda-Calculi}, Cambridge University Press, Cambridge, 1998.

\bibitem{da58}
M.~Davis. \emph{Computability and Unsolvability}. Mc-Graw-Hill, New York, 1958.

\bibitem{de83}
T.~Deil. \emph{Darstellungen und Berechenbarkeit reeller Zahlen}. Ph.D.~thesis. Fernuniversit\"at Hagen, 1983.

\bibitem{er72}
Yu.~L.~Ershov. Computable functionals of finite types. \emph{Algebra i Logika}, 11: 367--437, 1972; English translation, \emph{Algebra and Logic}, 11: 203--242, 1973.

\bibitem{er73}
Yu.~L.~Ershov. Theorie der Numerierungen I. \emph{Zeitschrift f\"ur mathematische Logik und Grundlagen der Mathematik}, 19: 289--388, 1973.

\bibitem{er75}
Yu.~L.~Ershov. Theorie der Numerierungen II. \emph{Zeitschrift f\"ur mathematische Logik und Grundlagen der Mathematik}, 21: 473--584, 1975.

\bibitem{goe74}
B.~Goetze. Die Struktur des Halbverbandes der effektiven Numerierungen. \emph{Zeitschrift f\"ur mathematische Logik und Grundlagen der Mathematik}, 20, 183--188, 1974.

\bibitem{goe76}
B.~Goetze. The structure of the lattice of recursive sets. \emph{Zeitschrift f\"ur mathematische Logik und Grundlagen der Mathematik}, 22, 187--191, 1976.

\bibitem{gs83}
D.~Gordon and E.~Shamir. Computation of recursive functionals using minimal initial segments. \emph{Theoretical Computer Science}, 23: 305--315, 1983.

\bibitem{gh03}
G.~Gierz, K.~H.~Hoffmann, K.~Keimel, J.~D.~Lawson, M.~W.~Mislove and D.~S.~Scott, \emph{Continuous Lattices and Domains}, Cambridge University Press, Cambridge, 2003.

\bibitem{ds90}
C.~A.~Gunter and D.~S.~Scott. Semantic domains. In J.~van~Leeuwen, editor, \emph{Handbook of Theoretical Computer Science, vol. B, Formal Models and Semantics}, Elsevier, Amsterdam, 1990, pages 633--674.

\bibitem{ha71}
J.~Hauck. Zur Pr\"azisierung des Begriffs berechenbare reelle Funktion. \emph{Zeitschrift f\"ur mathematische Logik und Grundlagen der Mathematik},17: 295--300. 1971.

\bibitem{ha73}
J.~Hauck. Berechenbare reelle Funktionen.  \emph{Zeitschrift f\"ur mathematische Logik und Grundlagen der Mathematik}, 19:121--140, 1973.

\bibitem{ha76}
J.~Hauck. Berechenbare reelle Funktionenfolgen.  \emph{Zeitschrift f\"ur mathematische Logik und Grundlagen der Mathematik}, 22: 265--282, 1976.

\bibitem{ha78}
J.~Hauck. Konstruktive Darstellungen reeller Zahlen und Funktionen.  \emph{Zeitschrift f\"ur mathematische Logik und Grundlagen der Mathematik}, 24: 365--374, 1978.

\bibitem{hi84}
P.~G.~Hinman. Finitely approximable sets. In M.~M.~Richter et al., editors, \emph{Computation and Proof Theory}, Springer, Berlin, 1984. pages 233--258.

\bibitem{hy79}
J.~M.~E.~Hyland. Filter spaces and continuous functionals. \emph{Annals of Mathematical Logic}, 16: 101--143, 1979.

\bibitem{kt83}
T.~Kamimura and A.~Tang. Effectively given spaces. In J.~Diaz, editor, \emph{Automata, Languages and Programming}, Springer, Berlin, 1983,  pages 385--396.

\bibitem{ke82}
W.~H.~Kersjes. \emph{Rekursionstheorie auf Teilmengen von P}. Diplomarbeit. RWTH Aachen, 1982.

\bibitem{kle36} 
S.~C.~Kleene. General recursive functions of natural numbers. \emph{Mathematische Annalen}, 112: 727--742, 1936. 

\bibitem{kl52}
S.~C.~Kleene. \emph{Introduction to Metamathematics}. Van Nostrand, Princeton, N.J., 1952.

\bibitem{kh69}
A.~B.~Khutoretski\u{\i}. On the reducibility of computable numerations. \emph{Algebra i Logika}, 8: 251--264, 1969; English translation, \emph{Algebra and Logic}, 8: 145--151, 1970.

\bibitem{kw85}
C.~Kreitz and K.~Weihrauch. Theory of representation. \emph{Theoretical Computer Science}. 38: 35--53, 1985.

\bibitem{la64b}
A.~H.~Lachlan. Standard classes of recursively enumerable sets. \emph{Zeitschrift f\"ur mathematische Logik und Grundlagen der Mathematik}, 10: 23--42, 1964.


\bibitem{ma70}
A.~I.~Mal'cev. \emph{Algorithms and Recursive Functions}. Wolters-Noordhoff, Groningen, 1970. 

\bibitem{ma71}
A.~I.~Mal'cev. \emph{The Metamathematics of Algebraic Structures. Collected Papers: 1936 - 1967}. B.~F.~Wells~III, editor. North-Holland, Amsterdam, 1971.

\bibitem{mes84}
W.~Menzel and V.~Sperschneider. Recursively enumerable extensions of $R_{1}$ by finite functions. In E.~B\"orger et al., editors, \emph{Logic and Machines: Decision Problems and Complexity}, Springer, Berlin,1984, pages 62--76.

\bibitem{mo63}
Y.~N.~Moschovakis. \emph{Recursive Analysis}. Ph.D.~thesis. University of Wisconsin, Madison, Wis., 1963.

\bibitem{ms55}
J,~Myhill and J.~C.~Shepherdson. Effective operations on partial recursive functions. \emph{Zeitschrift f\"ur mathematische Logik und Grundlagen der Mathematik}, 1: 310--317, 1955.

\bibitem{od92}
P.~G.~Odifreddi, \emph{Classical Recursion Theory}. Elsevier, Amsterdam, 1992.

\bibitem{pe64}
M.~B.~Pour-El. G\"odel numberings versus Friedberg numberings. \emph{Proceedings of the American Mathematical Society}, 15: 252--256, 1964.

\bibitem{pr89}
M.~B.~Pour-El and J.~I.~Richards. \emph{Computability in Analysis and Physics}. Springer, Berlin, 1989.

\bibitem{ri53}
H.~G.~Rice. Classes of recursively enumerable sets and their decision problems. \emph{Transactions of the American Mathematical Society}, 74: 358--366, 1953.

\bibitem{ro58}
H.~Rogers, Jr. G\"odel numberings of partial recursive functions. \emph{Journal of Symbolic Logic}, 23: 331--341, 1958.

\bibitem{ro67}
H.~Rogers, Jr. \emph{Theory of Recursive Functions and Effective Computability}. McGraw-Hill, New York, 1967.

\bibitem{sc70}
D.~Scott. \emph{Outline of a Mathematical Theory of Computation}. Technical Monograph PRG-2. Oxford University Computing Laboratory, 1970.

\bibitem{sc81}
D.~Scott. \emph{Lectures on a Mathematical Theory of Computation}. Technical Monograph PRG-19. Oxford University Computing Laboratory, 1981.

\bibitem{sc82}
D.~Scott. Domains for denotational semantics. In M.~Nielsen and E.~M.~Schmidt, editors,  \emph{Automata, Languages and Programming}, Springer, Berlin, 1982, pages 577--613. 

\bibitem{sp85}
D.~Spreen. \emph{Rekursionstheorie auf Teilmengen partieller Funktionen}. Habilitation thesis. RWTH Aachen, 1985.

\bibitem{sp90}
D.~Spreen. Computable one-to-one enumerations of effective domains. \emph{Information and Computation}, 84(1): 26-45, 1990.


\bibitem{slg94}
V.~Stoltenberg-Hansen, I.~Lindstr\"om and E.~R.~Griffor. \emph{Mathematical Theory of Domains}. Cambridge University Press. Cambridge, 1994.

\bibitem{tu3637}
A,~Turing. On computable numbers, with an application to the Entscheidungsproblem. \emph{Proceedings of the London Mathematical Society}, s2-42(1): 230--265, 1936/37

\bibitem{tu37}
A.~Turing. A correction.  \emph{Proceedings of the London Mathematical Society}, s2-43: 544--546, 1937.

\bibitem{wd80}
K.~Weihrauch and T.~Deil. \emph{Berchenbarkeit auf cpo-s}. Schriften zur Angewandten Mathematik und Informatik Nr. 63. RWTH Aachen, 1980.

\bibitem{ws83}
K.~Weihrauch and G.~Sch\"afer. Admissible representations of effective cpo-s. \emph{Theoretical Computer Science}, 26: 131-141, 1983.

\bibitem{we87}
K.~Weihrauch. \emph{Computability}. Springer, Berlin, 1987.

\bibitem{we00}
K.~Weihrauch. \emph{Computable Analysis}. Springer, Berlin, 2000.

\end{thebibliography}
\end{document}